\documentclass[10pt]{article}
\usepackage{pgfplots}
\pgfplotsset{compat=1.14}
\usepackage{xspace}
\usepackage{tikz}
\usepackage{multirow}
\usepackage{marvosym}
\usepackage{enumitem}
\usepackage{fullpage}
\usepackage{amsthm}
\usepackage{amsmath}
\usepackage{amssymb}
\usepackage{mathtools}
\usepackage{amsfonts}
\usepackage{booktabs}
\usepackage{relsize}
\usepackage[T1]{fontenc}
\usepackage{adjustbox}

\usepackage{hyperref}
\pgfplotsset{compat=1.16}
\usepackage{tikz-3dplot}
\usepackage{xspace}
\usepackage{tikz}
\pgfdeclarelayer{background}
\pgfsetlayers{background,main}

\usepackage{xcolor,colortbl}

\usepackage{thmtools}

\usepackage{verbatim}
\usepackage{graphicx}
\usepackage{algorithmicx}
\usepackage[ruled]{algorithm}
\usepackage{algpseudocode}
\usepackage{xcolor}
\usepackage[center]{subfigure}
\usepackage{caption}
\usepackage{url}

\algnewcommand{\algorithmicand}{\textbf{ and }}
\algnewcommand{\algorithmicor}{\textbf{ or }}
\algnewcommand{\OR}{\algorithmicor}
\algnewcommand{\AND}{\algorithmicand}

\usepackage{xifthen}
\usepgfplotslibrary{ternary, units}
\usetikzlibrary{patterns,shapes,arrows,positioning,fit,decorations.pathreplacing}
\usetikzlibrary{decorations.pathmorphing, pgfplots.ternary, pgfplots.units}
\usetikzlibrary{shadows.blur}
\usepackage{pgfplotstable}
\usepgfplotslibrary{groupplots}
\usepackage{multicol}
\usepackage[normalem]{ulem}

\definecolor{goodgreen}{rgb}{0.1, 0.5, 0.1}
\definecolor{Purple}{HTML}{911146}

\newcommand{\nop}[1]{}
\theoremstyle{plain}                  
\newtheorem{hyp}{Hypothesis}
\newtheorem{assumption}{Assumption}
\newtheorem{theorem}{Theorem}
\newtheorem{lemma}[theorem]{Lemma}

\newtheorem{corollary}[theorem]{Corollary}         

\newtheorem{example}[theorem]{Example}

\newcommand{\revision}[1]{\textcolor{black}{#1}}

\date{}
\title{ADOPT: Adaptively Optimizing Attribute Orders for \\ 
Worst-Case
Optimal Join Algorithms via Reinforcement Learning}

\author{
\begin{tabular}{cccc}
Junxiong Wang$^1$  & Immanuel Trummer$^1$ & Ahmet Kara$^2$ &  Dan Olteanu$^2$ \\
junxiong@cs.cornell.edu &  itrummer@cornell.edu & kara@ifi.uzh.ch &  olteanu@ifi.uzh.ch 
\end{tabular}\\ \\
$^1$Cornell University  \enspace\enspace $^2$University of Zurich
}

\begin{document}
\maketitle

\begin{abstract}
The performance of worst-case optimal join algorithms depends on the order in which the join attributes are processed. Selecting good orders before query execution is hard, due to the large space of possible orders and unreliable execution cost estimates in case of data skew or data correlation. We propose ADOPT, a query engine that combines adaptive query processing with a worst-case optimal join algorithm, which uses an order on the join attributes instead of a join order on relations. ADOPT divides query execution  into episodes in which different attribute orders are tried. Based on run time feedback on attribute order performance, ADOPT converges quickly to near-optimal orders. It avoids redundant work across different orders via a novel data structure, keeping track of parts of the join input that have been successfully processed. It selects attribute orders to try via reinforcement learning, balancing the need for exploring new orders with the desire to exploit promising orders. 
In experiments with various data sets and queries, it outperforms baselines, including commercial and open-source systems using worst-case optimal join algorithms, whenever queries become complex and therefore difficult to optimize.
\end{abstract}

\section{Introduction}

The area of join processing has recently been revolutionized by  worst-case optimal join algorithms~\cite{DBLP:conf/icdt/Veldhuizen14,DBLP:journals/jacm/NgoPRR18}. LeapFrog TrieJoin (LFTJ) is a prime example of a worst-case optimal join algorithm~\cite{DBLP:conf/icdt/Veldhuizen14}. 
Such algorithms guarantee asymptotically worst-case optimal performance. Those formal guarantees set them apart from traditional join algorithms, which are known to be sub-optimal~\cite{DBLP:journals/sigmod/NgoRR13}. In practice, they often translate into orders of magnitude runtime improvements, specifically for cyclic queries, compared to traditional approaches. They are incorporated in  several recent query engines: for factorized databases~\cite{DBLP:journals/tods/OlteanuZ15,DBLP:journals/sigmod/OlteanuS16}, graph processing~\cite{Aberger2016,DBLP:conf/cidr/JinAS22} and general query processing~\cite{Freitag2020}, in-database machine learning~\cite{DBLP:conf/sigmod/SchleichOC16}, and in the commercial systems LogicBlox~\cite{Aref2015} and RelationalAI~\cite{DBLP:conf/datalog/Aref19}.

\revision{As pointed out in prior work~\cite{DBLP:conf/icdt/Veldhuizen14}, in practice, the} performance of worst-case optimal join algorithms often depends heavily on the order in which join attributes (i.e., groups of join columns that are linked by equality constraints) are processed. \revision{Yet this is not reflected in the formal analysis of worst-case optimal join algorithms~\cite{DBLP:conf/icdt/Veldhuizen14}. Worst-case optimality is defined with regards to worst-case assumptions about the database content. Under these assumptions, different attribute orders have asymptotically equivalent time complexity. On the other side, given the actual database content, some attribute orders may perform much better than others in practice. Similar to the classical join ordering problem, it is therefore important to aim for the instance-optimal order, e.g., using data statistics.} 

%Even though all orders are equally good to achieve optimality in the {\em worst case}, their performance may differ significantly in practice. Not unlike the join order for traditional join algorithms, the optimal attribute order depends on the database state. For a fixed database state, the performance gap between good and bad attribute orders is significant, even for moderately sized data sets. Whereas worst-case optimality guarantees (asymptotically) optimal performance, relative to a \textit{worst-case database state}, it does not guarantee optimality with regards to the \textit{actual state}. Hence, the theoretical guarantees offered by worst-case optimal algorithms do not remove the need for careful attribute order selection.

\begin{figure}
\centering
\begin{tikzpicture}
\begin{groupplot}[
            enlargelimits=false,
            ylabel= {Scaled Time},
            xlabel= {Ranked Attribute Orders},
            xtick={0, 20, 40, 60, 80, 100, 119},
            ytick={2, 4, 8, 16, 32},
            xtick pos=left,
            ytick pos=left,
            ymajorgrids,
            ymode=log,
            log basis y={2},
            /pgf/bar width=2pt,
            group style={
                rows=2,
                columns=1,
            },
            width=.45\textwidth,
            height=0.2\textwidth,
            clip=false,
        ]
        
        %\nextgroupplot[title={\large 5 clique} ]
        \nextgroupplot
            
        \addplot [ybar,draw=blue,fill=blue,] table [x expr=\coordindex, y index=0]{data/5clique_order.csv};
        
        % \nextgroupplot[title={\large 5 cycle} ]
        % \nextgroupplot
            
        % \addplot [ybar,draw=blue,fill=blue,] table [x expr=\coordindex, y index=0]{data/5cycle_order.csv};
        
\end{groupplot}

\end{tikzpicture}
\caption{Execution time of different attribute orders for the five-clique query on the ego-Twitter graph.}
\label{fig:intro}
\vspace*{-1em}
\end{figure}
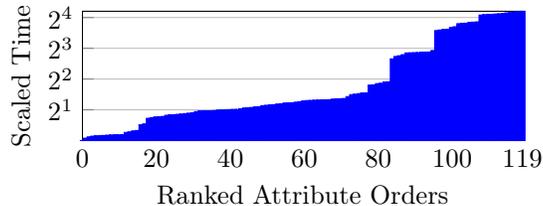

\begin{example}
\rm
Figure~\ref{fig:intro} illustrates the need for accurate attribute order selection. It compares LFTJ execution times (scaled to the time of the fastest order) for different attribute orders and the same query that asks for the number of cliques of five distinct nodes. 120 attribute orders (on the x-axis) are ranked by execution time. \nop{A more detailed description of the experimental setup is deferred to Section~\ref{sec:experiments}.} The performance gap between the best and worst attribute orders is more than 16x. The choice of an attribute order has thus significant impact on performance and near-optimal orders are sparse.
\end{example}

Execution engines using worst-case optimal join operators (e.g., the LogicBlox system~\cite{Aref2015}) typically select attribute orders via a query optimizer. Similar to traditional query optimizers selecting join orders, such optimizers exploit data statistics and simplifying cost models to pick an attribute order. This approach is however risky. Erroneous cost estimates (e.g., due to data skew not represented in data statistics) can lead to highly sub-optimal attribute order choices. Incorrect cost estimates are known to cause significant overheads in traditional query optimization~\cite{Lohman2014}. The experiments in Section~\ref{sec:experiments} show that this case appears in the context of optimization for worst-case optimal join algorithms as well. \revision{In particular, this applies to queries on non-uniform data with an elevated number of predicates, increasing the potential for inter-predicate correlations that make size and cost predictions hard.}

%In particular, optimizing queries with an elevated number of join attributes is difficult especially for large queries on non-uniform data where cost estimation is difficult. The following example enumerates concrete use cases that have these properties.}

\begin{example}
\label{ex:usecases}
\rm
\revision{In social network analysis, analysts are often interested in finding people who are mutually connected in cliques via links in the graph representing the social network~\cite{Khodadadi2021, Krebs2002, palla2005uncovering}. Specifically, prior analysis often considers cliques of up to five or six~\cite{Khodadadi2021, palla2005uncovering}, or more~\cite{Krebs2002} members. The experiments in Section~\ref{sec:experiments} show that such queries already create challenges in cost prediction, making methods that are robust to prediction errors preferable.}
\end{example}

%The guarantees, offered by worst-case optimal join algorithms, refer however to a worst-case scenario. They do not refer to the actual database state. Hence, worst-case optimal join algorithms do not guarantee optimal execution time

% \begin{figure}
% \centering
% \begin{tikzpicture}
% \begin{groupplot}[
%             enlargelimits=false,
%             ylabel= {Scaled Time},
%             xlabel= {Ranked Attribute Orders},
%             ytick={2, 4, 8, 16, 32},
%             xtick pos=left,
%             ytick pos=left,
%             ymajorgrids,
%             ymode=log,
%             log basis y={2},
%             /pgf/bar width=2pt,
%             group style={
%                 rows=2,
%                 columns=1,
%             },
%             width=\textwidth,
%             height=0.2\textwidth,
%             clip=false,
%         ]
        
%         \nextgroupplot[title={\large 5 clique} ]
            
%         \addplot [ybar,draw=blue,fill=blue,] table [x expr=\coordindex, y index=0]{data/5clique_order.csv};
        
%         \nextgroupplot[title={\large 5 cycle} ]
            
%         \addplot [ybar,draw=blue,fill=blue,] table [x expr=\coordindex, y index=0]{data/5cycle_order.csv};
% \end{groupplot}
% \end{tikzpicture}
% \end{figure}

To overcome these challenges, we propose an adaptive execution strategy for worst-case optimal join algorithms. The goal of adaptive processing is to enable attribute order switches, during query processing. The processing time is divided into episodes and in each episode we may choose an attribute order for the execution of the query over a fragment of the input data. By measuring execution speed for different attribute orders, the adaptive processing framework converges to near-optimal attribute orders over time. To the best of our knowledge, this is the first adaptive processing strategy for worst-case optimal join algorithms.

Adaptive processing for query processing based on attribute orders leads, however, to new challenges, discussed in the following. 

First, we must limit overheads due to attribute order switching. In particular, we must avoid redundant processing when applying multiple attribute orders to the same data. We solve this challenge by a task manager, capturing execution progress achieved by different attribute orders. Join result tuples are characterized by a value combination for join attributes. Hence, we generally represent execution progress by (hyper)cubes within the Cartesian product of value ranges over all join attributes. Having processed a cube implies that all contained result tuples, if any, have been generated. Data processing threads query the task manager to retrieve cubes not covered previously. Also, the task manager is updated whenever new results become available. It ensures that different threads process non-overlapping cubes, independently of the current attribute order. Query processing ends once all processed cubes, in aggregate, cover the full space of join attribute value ranges.

%is queried before each attribute order switch. It returns a hypercube, within the space formed by the Cartesian product of value ranges of all join attributes, that has not been processed according to any prior attribute order. Hence, all query result tuples contained within the hypercube (if any) have not been generated yet. Until the next switch, join processing is therefore restricted to tuples within the hypercube. Before switching to a new attribute order, all execution progress achieved until that point is registered in the aforementioned data structure. The unexplored subspace from the current hypercube is expressed by a disjoint union of a small number of  hypercubes and returned back to the data structure for future processing.

Second, we need a metric to compare different attribute orders, based on run time feedback. This metric must be applicable even when executing attribute orders for a very short amount of time. The number of result tuples generated per time unit may appear to be a good candidate metric. However, it is not informative in case of small results. Instead, we opt for a metric analyzing the size of the hypercube (within the Cartesian product of join attribute values) covered per time unit. Even if no result tuples are generated, this metric rewards attribute orders that quickly discard subsets of the output space.

Third, we must choose, in each episode, which attribute order to select next. This choice is challenging as it is subject to the so called \textit{exploration-exploitation dilemma}. Choosing attribute orders that obtained good scores in past invocation (exploitation) may seem beneficial to generate a full query result as quickly as possible. However, executing attribute orders about which little is known (exploration) may be better. It may lead to even better attribute orders that can be selected in future episodes. To balance between these two extremes in a principled manner, we employ methods from the area of reinforcement learning. Under moderately simplifying assumptions, based on the guarantees offered by these methods, we can show that ADOPT converges to optimal attribute orders.

% In summary, our original, scientific contributions
% are the following.

We have integrated our approach for adaptive processing with worst-case optimal join algorithms into ADOPT (ADaptive wOrst-case oPTimal joins), a novel, analytical SQL processing engine. We compare ADOPT to various baselines, including traditional database systems such as PostgreSQL and MonetDB, prior methods for adaptive processing such as SkinnerDB~\cite{Trummer2021c}, as well as commercial and open-source database engines that use worst-case optimal join algorithms. We evaluate all systems on acyclic and cyclic queries from public benchmarks, \revision{TPC-H, JCC-H~\cite{boncz2018jcc},} join order~\cite{Gubichev2015} and SNAP graph data~\cite{snapnets,nguyen2015join} workloads. \revision{For complex queries on skewed data, ADOPT outperforms all competitors.} In particular, it improves over a commercial database engine using the same worst-case optimal join algorithm as ADOPT. This demonstrates the benefit of adaptive attribute order selections.

%For small queries, adaptive processing does not result in performance improvements over competitors. However, for larger and more complex queries, ADOPT outperforms all competitors. 

In summary, the contributions in this paper are the following:

\begin{itemize}
    \item We propose the first adaptive processing strategy for worst-case optimal join algorithms using reinforcement learning.
    \item We describe specialized data structures, progress metrics, and learning algorithms that make adaptive processing in this scenario practical.
    \item We formally analyze worst-case optimality guarantees and convergence properties.
    \item We compare ADOPT experimentally against various baselines, showing that it outperforms them for a variety of acyclic and cyclic queries and datasets.
\end{itemize}

%implement our adaptive strategy in a system called ADOPT. We

The remainder of this paper is organized as follows. Section~\ref{sec:overview} presents an overview of the ADOPT system. Section~\ref{sec:algorithm} describes the algorithm used for adaptive processing in detail. Section~\ref{sec:analysis} analyzes our approach formally.
Section~\ref{sec:experiments} shows experimentally that ADOPT outperforms a range of competitors for both acyclic and cyclic queries.
Section~\ref{sec:related} discusses prior related work.
Appendices \ref{app:statistics_graph_datasets}-\ref{sec:attributescalability} report on further details on the experimental results and the datasets used in the experiments. 
Appendix \ref{sec:illustrating_LFTJ}  details the original LFTJ algorithm ADOPT is based on while
Appendix \ref{sec:lftjinadopt} explains the LFTJ variant used for ADOPT.

\section{Overview}
\label{sec:overview}

Figure~\ref{fig:overview} overviews the ADOPT system, illustrating its primary components. ADOPT \revision{supports SPJAG queries with sub-queries, covering the majority of TPC-H queries (see Section~\ref{sub:expsetup} for details)}. It performs in-memory data processing and uses a columnar data layout. \revision{The highly specific requirements of the ADOPT approach (e.g., support for high-frequency attribute order switching in a worst-case optimal join processing framework) motivate a customized system, rather than the integration into classical SQL execution engines.} The implementation uses Java and supports multi-threading via the Java ExecutorService API. It uses a worst-case optimal algorithm to process joins and selects attribute orders via reinforcement learning.

\tikzstyle{system}=[draw, rectangle, rounded corners=1mm, shade, top color=gray!10, bottom color=gray!20, blur shadow={shadow blur steps=5}]
\tikzstyle{innercomponent}=[draw, rectangle, rounded corners=1mm, shade, top color=blue!10, bottom color=blue!20, minimum width=1.7cm, font=\bfseries, blur shadow={shadow blur steps=5}, align=center]
\tikzstyle{outercomponent}=[draw, rectangle, rounded corners=1mm, shade, top color=red!10, bottom color=red!20, minimum width=1.7cm, font=\bfseries, blur shadow={shadow blur steps=5}, align=center]
\tikzstyle{system}=[draw, rectangle, shade, top color=gray!10, bottom color=gray!20, blur shadow={shadow blur steps=5}]

\tikzstyle{dataflow}=[thick, -latex]
\tikzstyle{mdataflow}=[thick, <->]
% \tikzstyle{label}=[font=\bfseries]

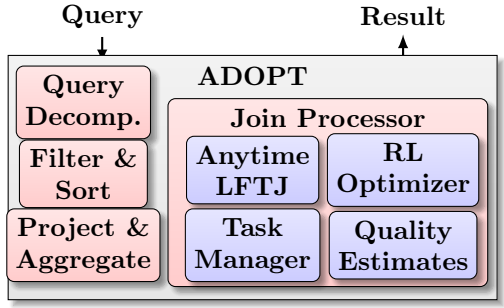
\begin{figure}
    \centering
    \begin{tikzpicture}
        \node[label,font=\bfseries] (query) at (-2,2.05) {Query};
        \node[label, font=\bfseries] (result) at (2,2.05) {Result};
        \draw[dataflow] (query.south) -- ++ (0,-0.35);
        \draw[dataflow] (result.south) ++ (0,-0.35) -- ++ (0,0.35);
        \node[system, minimum width=6.5cm, minimum height=3.25cm] at (0,-0.1) {};
        \node[label, font=\bfseries] at (0,1.25) {ADOPT};
        \node[outercomponent, minimum width=4.25cm, minimum height=2.5cm] at (1,-0.3) {};
        \node[label, font=\bfseries] at (1,0.7) {Join Processor};
        \node[innercomponent] at (0,0) {Anytime\\LFTJ};
        \node[innercomponent] at (2,0) {RL\\Optimizer};
        \node[innercomponent] at (0,-1) {Task\\Manager};
        \node[innercomponent] at (2,-1) {Quality\\Estimates};
        %\node[outercomponent] at (-2.25,0.7) {Unnesting};
        \node[outercomponent] at (-2.25,0.9) {Query\\Decomp.};
        \node[outercomponent] at (-2.25,-0.05) {Filter \&\\ Sort};
        \node[outercomponent] at (-2.25,-1) {Project \&\\ Aggregate};
    \end{tikzpicture}
    \caption{\revision{Overview of ADOPT system components.}}
    \label{fig:overview}
    \vspace*{-1em}
\end{figure}

%After parsing, each query is first unnested into a sequence of simple queries (i.e., SPJGA queries without sub-queries). 

\revision{For complex queries (i.e., queries with sub-queries), ADOPT first decomposes them into a sequence of simple SPJAG queries, using decomposition techniques proposed in prior work~\cite{Neumann, trummer2019skinnerdb}. After decomposition, it executes the resulting queries, storing query results in temporary tables that are referenced by later queries in the query sequence (as input tables).} For \revision{each simple} query, ADOPT first performs a pre-processing step to filter the tables using unary predicates from the query \revision{(the resulting tables are typically much smaller than the original ones)}. After that, the following join phase is executed on the filtered tables.

For worst-case optimal computation of equality joins, ADOPT uses  LeapFrog TrieJoin (LFTJ). LFTJ considers join attributes in a fixed order to find value combinations that satisfy all join predicates. ADOPT uses an anytime version of this algorithm, so it can suspend and resume execution with high frequency. This enables the adaptive processing strategy, allowing ADOPT to identify near-optimal attribute orders, based on run time feedback. \revision{Similar to LFTJ, ADOPT does not materialize intermediate join results: LFTJ stores at most one tuple, containing one value per attribute, as intermediate state and adds complete tuples directly to the join result. This makes suspend and resume operations very efficient. Appendix~\ref{sec:illustrating_LFTJ} provides further details on the original LFTJ algorithm that ADOPT is based upon, including several examples. Also, Appendix~\ref{sec:lftjinadopt} details on the LFTJ variant used for ADOPT. In particular, it discusses how ADOPT uses and maintains data structures enabling the system to perform fast seek operations on the input tables, retrieving tuples that satisfy inequality conditions on their attributes.}

%Similar to LFTJ, ADOPT does not materialize intermediate join results: LFTJ stores at most one tuple of currently selected value per attribute as intermediate state and outputs it right away in case it is in the join result. This makes suspend and resume operations very efficient.

Besides the join algorithm itself, ADOPT uses an optimizer based on reinforcement learning. The optimizer selects attribute orders, balancing the need for exploration (i.e., trying out new attribute orders) with the need for exploitation (i.e., trying out attribute orders that performed well in the past). Each selected attribute order is only executed for a limited number of steps, enabling ADOPT to try thousands of attribute orders per second. To compare different attribute orders, ADOPT generates quality estimates. These estimates judge the performance achieved via an attribute order during a single invocation. Performance may vary, for the same order, across different invocations (e.g., due to heterogeneous data distributions). However, by averaging over different invocations for the same attribute order, ADOPT obtains increasingly more precise quality estimates over time.

%(similar in intent but different from the constraint data structure introduced by Ngo and R\'e~\cite{Ngo})

Switching between attribute orders makes it challenging to avoid redundant work. ADOPT uses a task manager to keep track of remaining parts of the join input to process. More precisely, the task manager manages (hyper)cubes in the Cartesian product space, formed by value ranges of all join attributes. Each cube represents a part of the input space that still has to be processed by some attribute order (i.e., corresponding result tuples, if any, have not been added into a shared result set yet). The execution of the anytime LFTJ is restricted to cubes that have not been processed yet. More precisely, data processing threads query the task manager for cubes, called \textit{target cubes} in the following, that do not overlap with any cubes processed previously or concurrently (by other threads). Threads process the target cube until completion or until reaching the per-episode limit of computational steps. The task manager is notified of processed parts of the target cube (if the step limit is reached, only a subset of the target cube, represented by a small set of cubes contained in the target cube, was processed). The task manager removes processed cubes from the set of remaining cubes.
Join processing terminates once the entire input (i.e., the hypercube representing the full Cartesian product of join attribute values) has been covered. This can be verified efficiently using the task manager. If no unprocessed cubes remain, a complete result has been generated. Depending on the type of query, ADOPT executes a post-processing stage in which group-by clauses and aggregates are executed. Specifically for \revision{(count, max, min, sum, and avg) aggregates without grouping}, ADOPT integrates join processing with aggregation and does not need to perform a post-processing stage.

Several processing phases of ADOPT can be parallelized. Specifically, ADOPT parallelizes the join preparation phase (i.e., unary predicates are evaluated on different data partitions in general) and sorts data in parallel. During the join phase, ADOPT assigning non-overlapping hypercubes to different threads. Hence, using the same mechanism that avoids redundant work across attribute orders, ADOPT avoids redundant work across different threads as well.

\section{Algorithm}
\label{sec:algorithm}

We discuss the algorithm used by ADOPT in detail. Section~\ref{sub:main} discusses the top-level function, used to process queries. Section~\ref{sub:join} introduces ADOPT's parallel anytime join algorithm with worst-case optimality guarantees. Section~\ref{sub:cubes} discusses the mechanism by which ADOPT avoids redundant work across different attribute orders. Section~\ref{sub:rl} describes how ADOPT selects attribute orders via reinforcement learning. Finally, Section~\ref{sub:estimation} describes the reward metric used to guide the learning algorithm.

\subsection{Main Function}
\label{sub:main}

\begin{algorithm}[t!]
\caption{Main function of ADOPT, processing queries.\label{alg:main}}
\renewcommand{\algorithmiccomment}[1]{// #1}
\begin{small}
\begin{algorithmic}[1]
\State \textbf{Input:} Query $q$, number of threads $n$, per-episode budget $b$
\State \textbf{Output:} Query result
\Function{ADOPT}{$q, n, b$}
\State \Comment{Filter input tables via unary predicates}
\State $\{R_1,\ldots,R_m\}\gets$\Call{Prep.UnaryFilter}{$q$}
\State \Comment{Initialize join result set}
\State $R\gets\emptyset$
\State \Comment{Initialize reinforcement learning}
\State \Call{RL.Init}{$q$}
\State \Comment{Initialize constraint store}
\State \Call{TM.Init}{$q,n$}
\State \Comment{Iterate until result is complete}
\While {$\lnot$ \Call{TM.Finished}{}}
\State \Comment{Select attribute order via UCT algorithm}
\State $o \gets$ \Call{RL.Select}{}
\State \Comment{Use order for limited join steps}
\State $reward\gets$\Call{AnytimeWCOJ}{$q,o,n,b,R$}
\State \Comment{Update UCT statistics with reward}
\State \Call{RL.Update}{$o,reward$}
\EndWhile
\State \Comment{\revision{Return result after post-processing}}
\State \Return{\revision{\textproc{Post}$(q,R)$}}
\EndFunction
\end{algorithmic}
\end{small}
\end{algorithm}

ADOPT uses Algorithm~\ref{alg:main} to process \revision{simple SPJAG queries (i.e., without sub-queries). In addition to the query,} the algorithm also takes as input a number of data processing threads and a number of computational steps spent to evaluate a selected attribute order.

First, ADOPT filters the tables with the unary predicates \revision{(Line~5)}. \revision{ADOPT supports hash indexes on single columns and uses them, if available, to retrieve rows satisfying unary equality predicates. Without indexes, it scans and filters data, exploiting multi-threading. After that,} the only remaining predicates are then join predicates \revision{(including equality and other join predicates)}. Next, the algorithm initializes the set of join result tuples, the reinforcement learning algorithm by specifying the search space of attribute orders (which depends on the query), and the task manager with the input query and the number of processing threads \revision{(Lines~6 to 11)}. 
Internally, the task manager initializes the hypercube representing the total amount of work for each thread. More precisely, it divides the cube, representing the Cartesian product of all join attribute ranges, into equal shares for each thread. 

The task manager keeps track of cubes processed by the worker threads. Hence, query processing finishes once all processed cubes, in aggregate, cover the full input space. Iterations continue \revision{(Lines~13 to 20)} until that termination condition is satisfied. In each iteration, ADOPT first selects an attribute order via reinforcement learning \revision{(Line~15)}. Then, it executes that order, in parallel, for a fixed number of steps \revision{(Line~17)}. By executing the attribute order, \revision{the result set ($R$) may get updated. Note that $R$ only contains complete result tuples (mapping each attribute to a value) or partial values for aggregates. However, it does not contain any intermediate result tuples. Besides updating results}, executing an attribute order yields reward values, representing execution progress per time unit. Those reward values are used to update statistics \revision{(Line~19)}, maintained internally by the reinforcement learning optimizer, to guide attribute order selections in future iterations. Once \revision{the join finishes}, the algorithm \revision{performs post-processing (e.g., calculating per-group aggregates for group-by queries, based on join results in $R$) and returns the result (Line~22).}

%result tuples may get added into the result set ($R$). \revision{Note that only complete result tuples (i.e., combinations of values for all attributes) are added to that set but no intermediate result tuples.} 
\subsection{Anytime Join Algorithm}
\label{sub:join}

\begin{algorithm}[t!]
\caption{Parallel anytime version of worst-case optimal join algorithm.\label{alg:anytimeLFTJ}}
\renewcommand{\algorithmiccomment}[1]{// #1}
\begin{small}
\begin{algorithmic}[1]
\State \textbf{Input:} Query $q$, attribute order $o$, number of threads $n$, per-episode budget $b$, Result set $R$
\State \textbf{Output:} Reward $r$
\Function{AnytimeWCOJ}{$q,o,n,b,R$}
\State \Comment{Initialize accumulated reward}
\State $r\gets 0$
\State \Comment{Execute in parallel for all threads}
\For{$1\leq t\leq n$ in parallel}
\State \Comment{Initialize remaining cost budget}
\State $l_t\gets b$
\State \Comment{Iterate until per-episode budget spent}
\While{$l_t>0$}
\State \Comment{Retrieve unprocessed target cube}
\State $c_t\gets$\Call{TM.Retrieve}{}
\State \Comment{Process cube until timeout, add results}
\State $\langle P_t,s_t\rangle\gets$\Call{JoinOneCube}{$q,l_t,o,c_t,R$}
\State \Comment{Update constraints via processed cube}
\State \Call{TM.Remove}{$c_t,P_t$}
\State \Comment{Update accumulated reward (see Section~\ref{sub:estimation})}
\State $r\gets r+Reward(P_t,q)$
\State \Comment{Update remaining budget}
\State $l_t\gets l_t-s_t$
\EndWhile
\EndFor
\State \Comment{Return accumulated reward}
\State \Return{$r$}
\EndFunction
\end{algorithmic}
\end{small}
\end{algorithm}

Algorithm~\ref{alg:anytimeLFTJ} is the (worst-case optimal) join algorithm, used to execute a given attribute order for a fixed number of steps. Execution proceeds in parallel: different worker threads operate on non-overlapping cubes. Each worker thread iterates the following steps until its computational budget is depleted \revision{(Lines~11 to 22)}. First, it retrieves an unprocessed cube, the target cube, from the task manager \revision{(Line~13)}. Then, it uses a sub-function (an anytime version of the LFTJ) to process the retrieved target cube \revision{(Line~15)}. In practice, it is often not possible to process the entire target cube under the remaining computation budget. Hence, the result of the triejoin invocation (Function~\textproc{JoinOneCube}) reports the set of cubes, contained within the target cube, that were successfully processed. In addition, it returns the number of computation steps spent. The task manager is notified of successfully processed cubes which will be excluded from further consideration \revision{(Line~17)}. Also, a reward value is calculated that represents progress towards generating a full join result \revision{(Line~19)}. We postpone a detailed discussion of the reward function to Section~\ref{sub:rl}. Finally, Algorithm~\ref{alg:anytimeLFTJ} returns the reward value, accumulated over all threads and iterations \revision{(Line~25)}.

\begin{algorithm}[t!]
\caption{Worst-case optimal join algorithm with timeout, joining a single cube.\label{alg:joinOneCube}}
\renewcommand{\algorithmiccomment}[1]{// #1}
\begin{small}
\begin{algorithmic}[1]
\State \textbf{Input:} Query $q$, remaining budget $b$, attribute order $o$, target cube to process $c$, result set $R$, attribute counter $a$, value mappings $M$
\State \textbf{Effect:} Iterates over attribute values and possibly adds results to $R$
\Procedure{JoinOneCubeRec}{$q,b,o,c,R,a,M$}
\If{$a\geq|q.A|$} \Comment{Check for completed result tuples}
\State Insert tuple with current attribute values $M$ into $R$
\Else
\State \Comment{Initialize value iterator (do not evaluate it!)}
\State $V\gets$ iterator over values for $o_a$ in $[c.l_{o_{a}},c.u_{o_{a}}]$ \revision{that satisfy}
\Statex $\quad\quad\quad\quad$ \revision{all applicable join predicates in $q$.}
\State \Comment{Iterate over values until timeout}
\For{$v\in V$}
\State \Comment{Select values for remaining attributes}
\State \Call{JoinOneCubeRec}{$q,l,o,c,R,a+1,M\cup\{\langle o_a,v\rangle\}$}
\State \Comment{Check for timeouts}
\If{Total computational steps $>b$}
\State \textbf{Break}
\EndIf
\EndFor
\EndIf
\EndProcedure
\vspace{0.25cm}
\State \textbf{Input:} Query $q$, remaining budget $b$, attribute order $o$, target cube to process $c$, result set $R$
\State \textbf{Output:} Processed cube $p$, computational steps performed $s$
\Function{JoinOneCube}{$q,b,o,c,R$}
\State \Comment{Resume join for fixed number of steps}
\State \Call{JoinOneCubeRec}{$q,b,o,c,R,0,\emptyset$}
\State \Comment{Retrieve state from \textproc{JoinOneCubeRec} invocation}
\State $s\gets$ Number of computational steps spent
\State $v\gets$ Vector s.t.\ $v_a$ is last value considered for attribute $o_a$
\State \Comment{Calculate processed cubes}
\State $P\gets\emptyset$
\For{$0\leq a< |q.A|$}
\State Create new cube $p$ s.t.\
\State $\quad$ $\forall i<a:p_i=[v_i,v_i];$
\State $\quad$ $\quad$ $p_a=[c.l_{o_a},v_a);$
\State $\quad$ $\quad$ $\quad$ $\forall a<i:p_i=[c.l_{o_i},c.u_{0_i}]$
\State $P\gets P\cup\{p\}$
\EndFor
\State \Return{$\langle P,s\rangle$}
\EndFunction
\end{algorithmic}
\end{small}
\end{algorithm}

Algorithm~\ref{alg:joinOneCube} describes the sub-function, used to process a single cube, at a high level of abstraction. The actual join is performed by Procedure~\textproc{JoinOneCubeRec}. This procedure is based on the leapfrog triejoin~\cite{DBLP:conf/icdt/Veldhuizen14}, a classical, worst-case optimal join algorithm\footnote{\revision{A detailed example of the LFTJ execution is given in Appendix \ref{sec:illustrating_LFTJ}.}}. For conciseness, the pseudo-code describes the algorithm as a recursive function (whereas the actual implementation does not use recursion). The input to the algorithm is the join query, the remaining computational budget, an attribute order, a target cube to process, the result set, and the index of the current attribute. The algorithm considers query attributes sequentially, in the given attribute order. The attribute index marks the currently considered attribute. Once the attribute index reaches the total number of attributes (represented as $q.A$), the algorithm has selected one value for each attribute. Furthermore, at that point, it is clear that the combination of attribute values satisfies all applicable join conditions. Hence, the algorithm adds the corresponding result tuple into the result set \revision{(Line~5)}. \revision{As a variant (not shown in Algorithm~\ref{alg:joinOneCube}), for queries with simple aggregates without grouping, ADOPT does not store result tuples but merely updates partial aggregate values for each aggregate.} If the attribute index is below the total number of attributes, the algorithm iterates over values for that attribute (i.e., attribute $o_a$ where $o$ is the order and $a$ the attribute index) \revision{in the loop from Line~10 to 17}. 

In Line~8, Algorithm~\ref{alg:joinOneCube} creates an iterator over values for the current attribute that \revision{satisfy all \textit{applicable} join predicates and are within the target cube,} i.e., values contained in the interval $[c.l_{o_a},c.u_{o_a}]$ for attribute number $a$ within order $o$ ($c.l$ and $c.u$ designate vectors, indexed by attribute, that represent lower and upper target cube bounds respectively). The algorithm does not assemble the full set of matching values before iterating (as that would create significant overheads when switching attribute orders before being able to try all collected values). Instead, Line~8 is meant to represent the initialization of data structures that allow iterating over matching values efficiently. \revision{Join predicates are applicable if, beyond the current attribute $o_a$, they only refer to attributes whose values have been fixed previously (i.e., a corresponding value assignment is contained in $M$). For equality join predicates, ADOPT uses the same mechanism as LFTJ~\cite{DBLP:conf/icdt/Veldhuizen14} to efficiently iterate over satisfying values. This mechanism is described in detail in Appendix~\ref{sec:illustrating_LFTJ}. It is based on data structures that support fast seek operations on query relations. Whenever required data structures are not available, ADOPT dynamically creates them at run time. For base relations, but not for relations filtered via unary predicates, ADOPT caches and reuses those data structures across queries.}

Join processing via Procedure~\textproc{JoinOneCube} terminates once the computational budget is depleted \revision{(check in Line~14)}, or if the current cube is entirely processed. Function~\textproc{JoinOneCube} retrieves the number of computational steps, spent during join processing, as well as the last selected value for each attribute. It uses the latter to calculate the set of processed cubes (to be removed from the set of unprocessed cubes). Procedure~\textproc{JoinOneCubeRec} does not advance from one value of an attribute to the next, unless all value combinations for the remaining attributes have been fully considered. Hence, if value range  $c.l_{o_{a}}$ to $v_a$ was covered for the current attribute $a$, the cube representing processed value combinations reaches the full cube dimensions for all attributes that appear later than $a$ in the order $o$, and is fixed to the currently selected value for all attributes appearing before $a$ in $o$. Note that the pseudo-code uses a shortcut to assign both cube bounds at once (e.g., $p_i=[v_i,v_i]$ is equivalent to $[p.l_i,p.u_i]=[v_i,v_i]$) \revision{in Lines~32 to 34}.

\tikzstyle{cube}=[draw=black, thick]

\begin{figure}
    \centering
    \begin{tikzpicture}
        \begin{axis}[width=5cm, height=5cm, legend pos=outer north east, area legend, legend columns=1, legend entries={Entire Cube, {Target Cube: $A_1,A_2$}, {Target Cube: $A_2,A_1$}, Processed Cubes}, ticks=none, xlabel={Attribute $A_1$}, ylabel={Attribute $A_2$}, ylabel near ticks, xlabel near ticks]
            \addplot[cube, fill=gray!20] coordinates {(0,0) (10,0) (10,10) (0,10) (0,0)};
            \addplot[cube, fill=blue] coordinates {(2,1) (5,1) (5,8) (2,8) (2,1)};
            \addplot[cube, fill=red] coordinates {(6,3) (9,3) (9,9) (6,9) (6,3)};
            
            \addplot[cube, fill=yellow] coordinates {(2,1) (4,1) (4,8) (2,8) (2,1)};
            \addplot[cube, fill=yellow] coordinates {(4,1) (5,1) (5,6) (4,6) (4,1)};
            
            \addplot[cube, fill=yellow] coordinates {(6,3) (9,3) (9,5) (6,5) (6,3)};
            \addplot[cube, fill=yellow] coordinates {(6,5) (7,5) (7,6) (6,6) (6,5)};
        \end{axis}
    \end{tikzpicture}
    \caption{Illustration of containment relationships between hypercubes when processing a query with two attributes.}
    \label{fig:cubes}
\end{figure}
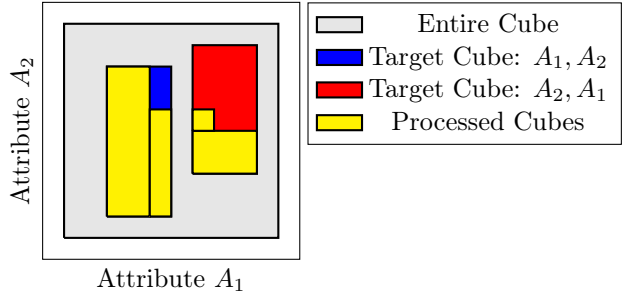

\begin{example}
\rm
Figure~\ref{fig:cubes} illustrates the containment relationships between different cubes when processing a query with two attributes. Processed cubes are contained within target cubes and target cubes are contained within the entire query cube. The figure represents target cubes that were processed, in different episodes, according to both possible attribute orders. The first one (left) was processed using order $A_1,A_2$. Hence, values for the first attribute change only after trying all values for the second attribute. Therefore, processed cubes fill the target cube ``column by column''. The other target was processed using the order $A_2,A_1$. Hence, processed cubes fill the target cube ``row by row''.
\end{example}

% The cube representing the scope of a single thread is located within the overall cube representing the entire query. The target cube is contained within the thread scope and represents the (maximal) goal for one invocation of Algorithm~\ref{alg:joinOneCube}. Processed cubes are contained within the target cube and may not entirely cover it, in case of a timeout. The figure indicates that a timeout occurred while processing the target cube. Also, it indicates that Attribute~1 comes, indeed, first in the processed attribute order. Hence, the cube associated with the value range considered for the first attribute covers the entire value range (of the target cube) for the second attribute.
\subsection{Avoiding Redundant Work}
\label{sub:cubes}

\begin{algorithm}[t!]
\caption{Managing cubes representing unprocessed join input.\label{alg:cubes}}
\renewcommand{\algorithmiccomment}[1]{// #1}
\begin{small}
\begin{algorithmic}[1]
\State $U\gets\emptyset$ \Comment{Global variable representing unprocessed cubes}
\vspace{0.25cm}
\State \textbf{Input:} Query $q$, number of threads $n$.
\State \textbf{Effect:} Initialize set of unprocessed cubes.
\Procedure{TM.Init}{$q, n$}
\State $A\gets$ attributes that appear in $q$ in equality join conditions
\State $[l_a,u_a]\gets$ attribute value ranges for all attributes $a\in A$
\State \Comment{Identify attribute with largest value domain}
\State $a^*\gets\arg\max_{a\in A}(u_a-l_a)$
\State \Comment{Use full value range for all but that attribute}
\State $f\gets\times_{a\in A:a\neq a^*}[l_a,u_a]$
\State \Comment{Divide largest value domain into per-thread ranges}
\State $\delta\gets (u_{a^*}-l_{a^*})/n$
\State \Comment{Form one unprocessed cube per thread}
\State $U\gets\{f\times [l_{a^*}+i\cdot\delta,l_{a^*}+(i+1)\cdot\delta|0\leq i<n]\}$
\EndProcedure
\vspace{0.25cm}
\State \textbf{Output:} Returns an unprocessed hypercube.
\Function{TM.Retrieve}{}
\State \Return Randomly selected cube from $U$
\EndFunction
\vspace{0.25cm}
\State \textbf{Input:} Target cube $c$ to subtract, processed cube set $P$.
\State \textbf{Effect:} Updates set of unprocessed cubes.
\Procedure{TM.Remove}{$c, P$}
\State \Comment{Subtract target cube from unprocessed cubes}
\State $U \gets U \setminus c$
\State \Comment{Add complement of processed cubes as unprocessed}
\For{$p \in P$}
\State \Comment{Get dimensions where $p$ fully covers $c$}
\State $F\gets$ indexes $i$ s.t.\ $p.l_i=c.l_i$ and $p.u_i=c.u_i$
\State \Comment{Get dimensions where $p$'s bounds collapse}
\State $S\gets$ indexes $i$ s.t.\ $p.l_i=p.u_i$
\State \Comment{Get single remaining dimension}
\State $d\gets$ single remaining dimension not in $F$ or $S$
\State Create new cube $u$ s.t.\
\State $\quad$ $u_{d}=(p.u_d,c.u_d]$; $\forall f\in F:u_f=p_f$; $\forall s\in S:u_s=p_s$
\State \Comment{Add newly created cube to unprocessed cubes}
\If{$u$ is not empty}
\State $U\gets U\cup\{u\}$
\EndIf
\EndFor
\EndProcedure
\vspace{0.25cm}
\State \textbf{Output:} True iff no unprocessed cubes are left.
\Function{TM.Finished}{}
\State \Return{\textbf{true} iff $U=\emptyset$}
\EndFunction
\end{algorithmic}
\end{small}
\end{algorithm}

ADOPT changes between different attribute orders over the course of query processing. This creates the risk of redundant work across different orders. ADOPT avoids redundant work by keeping track of cubes, in the space of join attribute values, that have not been considered yet. More precisely, ADOPT keeps track, at any point in time, of remaining, i.e.\ unprocessed, cubes. Whenever one of the processing threads requests a new cube to work on, ADOPT returns an unprocessed cube, thereby avoiding redundant work.

Algorithm~\ref{alg:cubes} gives functions used to manipulate cubes. At the beginning (Procedure~\textproc{TM.Init}), it initializes the set of unprocessed cubes to cover the entire attribute space. To do so, ADOPT first retrieves all join attributes \revision{(Line~5)}, then their value ranges \revision{(Line~6)}. Forming one single cube (i.e., the Cartesian product of all value ranges) diminishes chances for parallelization, at least at the start of query processing. Hence, ADOPT divides the attribute value space into equal-sized cubes with one cube per thread \revision{(Lines~7 to 14)}. To do so, it uses the attribute with maximal value domain, dividing its range equally across threads \revision{(Line~12)}. \revision{Note that, as discussed in the following, threads are not restricted to processing cubes initially assigned to them over the entire course of query evaluation. Instead, at the end of each episode, unprocessed parts of cubes assigned to a specific thread may get re-assigned to other threads.}

Whenever a worker threads requests a cube to work on \revision{(Line~13 in Algorithm~\ref{alg:anytimeLFTJ})}, a randomly selected cube from the set of unprocessed cubes is returned \revision{(Line~18 in Algorithm~\ref{alg:cubes})}. Note that the pseudo-code is slightly simplified, compared to the implementation, by omitting checks used to avoid concurrent changes to the set of unprocessed cubes (by multiple threads). 

Whenever a worker threads finished processing, it registers a set of cubes that was processed. It calls Procedure~\textproc{TM.Remove} to update the set of unprocessed cubes. This function takes two parameters, representing the set of processed cubes as well as the target cube, as input. All processed cubes are contained within the target cube and have a special structure, explained in the following. As a first step, ADOPT removes the target cube from the set of unprocessed cubes \revision{in Line24} (the target cube was selected by an invocation of the \textproc{TM.Retrieve} function and is therefore contained in the set $U$). If the set of processed cubes, in aggregate, do not cover the target cube (in general, that is the case), the set of unprocessed cubes is now missing all cubes contained in the target cube but not covered by the processed cubes. Hence, ADOPT adds more unprocessed cubes to reflect the difference.

Each processed cube has a special form, due to the structure of the join algorithm generating it (Lines~23 to 28 in Algorithm~\ref{alg:joinOneCube}). All processed cubes are generated according to the same attribute order and based on the same, final values selected for each attribute. Consider one single processed cube, using the selected attribute values $v_s$ for a prefix $S$ of the attribute order, the range of values up to the selected value $v_d$ for a single attribute $d$, and the full target cube range for the remaining attributes $F$. Clearly, given the selected values for attributes $S$, none of the values greater than $v_d$ for attribute $d$ has been considered by the join algorithm (instead, such value combinations would have been considered later by the join algorithm). Hence, the corresponding cube is added to the set of unprocessed cubes \revision{(Line~37)}. Also note that these unprocessed cubes cannot overlap (as, for each pair of unprocessed cubes, there is at least one attribute $a$ for which one cube fixes a value $v_a$, the other cube covers only values greater than $v_a$). This preserves the invariant that elements of $U$, representing unprocessed cubes, do not overlap. It also means that work done by different threads does not overlap. The processing finishes (Procedure~\textproc{TM.Finished}) whenever no unprocessed cubes are left.

\tikzstyle{values}=[only marks, mark=x, draw=black, mark size=6, ultra thick]
\tikzstyle{processedcube}=[fill=blue!20]
\tikzstyle{unprocessed}=[ultra thick, draw=red]

\begin{figure}
    \centering
    \begin{tikzpicture}
    \begin{groupplot}[group style={group size=3 by 1, x descriptions at=edge bottom, y descriptions at=edge left}, width=3cm, height=3cm, xmin=0, xmax=3, ymin=0, ymax=5, xlabel=Attribute, ylabel=Value, xtick=\empty, ylabel near ticks, xlabel near ticks, ymajorgrids, ytick={1,2,3,4,5}, yticklabel style={yshift=-0.2cm}]
    \nextgroupplot
    \draw[processedcube] (axis cs:0,0) rectangle (axis cs:1,4);
    \draw[processedcube] (axis cs:1,0) rectangle (axis cs:2,5);
    \draw[processedcube] (axis cs:2,0) rectangle (axis cs:3,5);
    \addplot[values] coordinates {(0.5,4.5) (1.5,2.5) (2.5,3.5)};
    \draw[unprocessed] (axis cs:0,4) rectangle (axis cs:1,5);
    \draw[unprocessed] (axis cs:1,0) rectangle (axis cs:2,5);
    \draw[unprocessed] (axis cs:2,0) rectangle (axis cs:3,5);
    % Currently selected values in attributes A1, A2, A3: 5, 3, 4
    % \addplot[fill=blue] coordinates {(0,3) (0,4) (1,4) (1,2) (2,2) (2,3) (3,3) (3,0) (2,0) (2,1) (1,1) (1,3) (0,3)};

    \nextgroupplot
    \draw[processedcube] (axis cs:0,4) rectangle (axis cs:1,5);
    \draw[processedcube] (axis cs:1,0) rectangle (axis cs:2,2);
    \draw[processedcube] (axis cs:2,0) rectangle (axis cs:3,5);
    \addplot[values] coordinates {(0.5,4.5) (1.5,2.5) (2.5,3.5)};
    \draw[unprocessed] (axis cs:0,4) rectangle (axis cs:1,5);
    \draw[unprocessed] (axis cs:1,3) rectangle (axis cs:2,5);
    \draw[unprocessed] (axis cs:2,0) rectangle (axis cs:3,5);

    \nextgroupplot
    \draw[processedcube] (axis cs:0,4) rectangle (axis cs:1,5);
    \draw[processedcube] (axis cs:1,2) rectangle (axis cs:2,3);
    \draw[processedcube] (axis cs:2,0) rectangle (axis cs:3,4);
    \addplot[values] coordinates {(0.5,4.5) (1.5,2.5) (2.5,3.5)};
    \draw[unprocessed] (axis cs:0,4) rectangle (axis cs:1,5);
    \draw[unprocessed] (axis cs:1,2) rectangle (axis cs:2,3);
    \draw[unprocessed] (axis cs:2,4) rectangle (axis cs:3,5);

    \end{groupplot}
    \end{tikzpicture}
    \caption{Illustrating cube processing in Example~\ref{ex:cube-removal}:  
    The initial target cube $([1,5], [1,5], [1,5])$ is processed up to $(5,3,4)$  (marked by X). Processed cubes are represented by blue rectangles, complementary unprocessed cubes by red rectangles.}
    \label{fig:removal}
    \vspace*{-1em}
\end{figure}
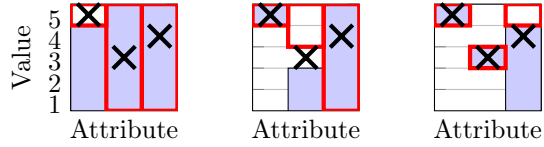

\begin{example}
\label{ex:cube-removal}
\rm
Figure~\ref{fig:removal} illustrates the processing of a target cube $([1,5], [1,5], [1,5])$ for an attribute order ($A_0$,$A_1$,$A_2$). In each sub-plot, the x-axis represents attributes while the y-axis represents attribute values.  Assume the timeout for this episode occurs after considering the values $(5,3,4)$ (marked by X). This means that we managed to process the following sub-cubes, left: $([1-4], [1-5], [1-5])$, middle: $(5, [1-2], [1-5])$, right: $(5, 3, [1-4])$. We infer the remaining unprocessed sub-cubes that complement these processed sub-cubes with respect to the target cube, left: $(5, [1-5], [1-5])$, middle: $(5, [4-5], [1-5])$, right:$(5, 3, 5)$.

\nop{
A timeout occurs during join processing after considering values (5,3,4) for the three attributes (in attribute order, selected values are marked by X). 
Each of the three sub-plots, from left to right, represents one of the processed cubes generated in that order via the loop from Lines~30 to 36 in Algorithm~\ref{alg:joinOneCube}. The extent of processed cubes, for each attribute, is marked up in blue. E.g., for the left-most plot, the processed cube is defined as $[1,4]\times[1,5]\times[1,5]$. For each processed cube, an unprocessed cube is added in Lines~26 to 39 in Algorithm~\ref{alg:cubes}, such that the unprocessed cubes together cover all unprocessed values. The value ranges of unprocessed cubes are marked as red rectangles. 
}

% For instance, in the left-most plot, the associated, unprocessed cubes is empty.
\end{example}

% \begin{example}
% Figure~\ref{fig:removal} illustrates removal of processed cubes. We consider a query with three attributes. In each sub-plot, the x-axis represents the attribute while the y-axis represents attribute value ranges. Hence, each sub-plot describes a cube which corresponds to the Cartesian product of the value ranges colored in blue. E.g., the plot in the left-upper corner represents the cube $[1,4]\times[1,5]\times[1,5]$.
% A timeout occurs during join processing after selecting values (5,3,4) for the three attributes (in attribute order). The upper row in Figure~\ref{fig:removal} shows the three processed cubes that result from this final state. From left to right, processed cubes are generated in that order via the loop from Lines~30 to 36 in Algorithm~\ref{alg:joinOneCube}. The cubes on the bottom represent the corresponding complement of unprocessed cubes. Note that the left-most cube is empty (it is therefore not added to the unprocessed cubes). The right-most unprocessed cube, for instance, captures, for fixed values in the first two attributes, the fact that value five for the last attribute was not considered yet.
% \end{example}
\subsection{Learning Attribute Orders}
\label{sub:rl}

ADOPT uses reinforcement learning to learn near-optimal attribute orders, over the course of a single query execution. At the beginning of each time slice, ADOPT selects an attribute order that maximizes the tradeoff between exploration and exploitation. It uses the Upper Confidence Bounds on Trees (UCT) algorithm~\cite{Kocsis2006} to choose an attribute order. This requires mapping the scenario (of attribute order selection) into a Markov-Decision Problem. Next, we discuss the algorithm as well as the problem model.

An episodic Markov Decision Process (MDP) is generally defined by a tuple $\langle s_0,S,A,T,R\rangle$ where $S$ is a set of states, $s_0\in S$ the initial state in each episode, $A$ a set of actions, and $T:S\times A\rightarrow S$ a transition function, linking states and action pairs to target states. Component $R$ represents a reward function, assigning states to a reward value. In our scenario, the transition function is deterministic while the reward function is probabilistic (i.e., states are associated with a probability distribution over possible rewards, rather than a constant reward that is achieved, every time the state is visited). The transition function is partial, meaning that certain actions are not available in certain states. Implicitly, we assume that all states without available actions are end states of an episode. After reaching and end state, the current episode ends and the next episode starts (from the initial state $s_0$ again). Given an MDP, the goal in reinforcement learning~\cite{Sutton2018} is to find a policy, describing behavior that results in maximal (expected) reward. In order to leverage reinforcement learning algorithms for our scenario, we must therefore map attribute order selection into the MDP formalism.

Our goal is to learn a policy that describes an attribute order. The policy generally recommends actions to take in a specific state. Here, we introduce one action for each query attribute. States are associated with attribute order prefixes (i.e., each state represents an order for a subset of attributes). To simplify the notation, we will refer to states by the prefix they represent, to actions by the attribute they correspond to. The transition function connects a first state $s_1$ to a second state $s_2$ via action $a$, if the second state can be reached by appending the attribute, represented by the action, to the prefix represented by the first state. More precisely, using the notation introduced before, the transition function links the state-action pair $\langle s_1,a\rangle$ to state $s_2=s_1\circ a$ (where $\circ$ represents concatenation). Each state represents a prefix of an attribute order in which each attribute appears at most once. Hence, the actions available in a state correspond to attributes that do not appear in the prefix represented by the state. This means that all states representing a complete attribute order are end states, implicitly. As a further restriction, we do not allow actions representing attributes that do not connect to any attributes in the prefix represented by the current state. This is similar to the heuristic of avoiding Cartesian product joins, used almost uniformly in traditional query optimizers. The reward function is set to zero for all states, except for end states. States of the latter category represent complete attribute orders. Upon reaching such a state, ADOPT executes the corresponding attribute order for a limited number of steps, measuring execution progress. The process by which execution process is measured is described in the following subsections.

ADOPT applies the UCT algorithm to solve the resulting MDP. As the MDP represents the problem of attribute ordering, linking rewards to execution progress, solving the MDP (i.e., finding a policy with maximal expected reward) yields a near-optimal attribute order. The UCT algorithm represents the state space as a search tree. Nodes represent states while tree edges represent transitions. Tree nodes are associated with statistics, establishing confidence bounds on the average reward associated with the sub-tree rooted at that node. Confidence bounds are updated as new reward samples become available. In each episode, the UCT algorithm selects a path from the search tree root to one of the leaf nodes. At each step, the UCT algorithm selects the child node with maximal upper confidence bound (hence the name of the algorithm). This approach converges to optimal policies~\cite{Kocsis2006}. After selecting a path to a leaf and calculating the associated reward, the UCT algorithm updates confidence bounds for each node on that path.

\revision{ADOPT grows the UCT search tree gradually over the course of query execution. At the start of execution, the tree only contains the root node. Then, in each episode, the tree is expanded by at most one node. Which nodes are added depends on the selected attribute orders. Each attribute order corresponds to a sequence of states in the MDP (a state represents an attribute order, each state appending one attribute, compared to its predecessor). In the fully grown search tree, each state is associated with one node. If, for the currently selected attribute order, some of the states do not have associated nodes in the tree yet, ADOPT expands the tree by adding a node for the first such state. ADOPT uses the partial tree to select attribute orders as follows. Given a state for which all possible successor states have associated nodes in the tree (i.e., reward statistics are available), ADOPT uses the aforementioned principle and selects the attribute that maximizes the upper confidence bound on reward values. If some of the successor states do not have associated nodes yet, ADOPT transitions to a randomly selected state among them (which will create a corresponding node). As a special case, if no nodes are available for any of the successor states, ADOPT selects the next attribute with uniform random distribution.}

%In our scenario, it is crucial to avoid generating the entire search tree at once, as this may cause non-negligible overheads. The number of possible attribute orders for a given query can be very large. Instead, we use a UCT variant that builds the search tree gradually. Specifically, starting from a tree containing only the root node, the algorithm expands the search tree by at most one node per sample. If the algorithm reaches an MDP state that has no associated node in the search tree, it performs random transitions until reaching an end state. In our scenario, this means that the remaining attributes are selected in random order. The resulting reward value corresponds to a uniform random sample for rewards in the corresponding sub-tree. If the algorithm encounters states with no associated tree nodes, it creates the first ``missing'' nodes and adds it to the tree. As time progresses, the tree becomes most refined along paths associated with interesting attribute orders.

\tikzstyle{uctnode}=[draw, circle, shade, top color=gray!10, bottom color=gray!20, blur shadow={shadow blur steps=5}, minimum width=1cm]
\tikzstyle{episode}=[color=red]

\begin{figure}[t]
    \centering
    \begin{tikzpicture}
        \node[uctnode] (start) at (0,0) {};
        \node[episode] at (0.5, 0.5) {0};
        \node[uctnode] (a) at (-2,-1) {A};
        \node[episode] at (-1.5, -0.5) {1};
        \node[uctnode] (ab) at (-3,-2) {AB};
        \node[episode] at (-2.5, -2.5) {4};
        \node[uctnode] (ac) at (-1,-2) {AC};
        \node[episode] at (-0.5, -2.5) {5};
        \node[uctnode] (b) at (0,-1.5) {B};
        \node[episode] at (0.5, -1) {2};
        \node[uctnode] (c) at (2,-1) {C};
        \node[episode] at (2.5, -0.5) {3};
        \node[uctnode] (cb) at (3,-2) {CB};
        \node[episode] at (3.5, -2.5) {6};
        \node[uctnode] (ca) at (1,-2) {CA};
        \node[episode] at (1.5, -2.5) {7};
        \node[uctnode, double] (cba) at (2,-3) {CBA};
        \node[episode] at (2.5, -3.5) {8};
        
        \draw[dataflow] (start) -- (a);
        \draw[dataflow] (start) -- (b);
        \draw[dataflow] (start) -- (c);
        \draw[dataflow] (a) -- (ab);
        \draw[dataflow] (a) -- (ac);
        \draw[dataflow] (c) -- (ca);
        \draw[dataflow] (c) -- (cb);
        \draw[dataflow] (cb) -- (cba);
    \end{tikzpicture}
    \caption{\revision{UCT search tree for a query with three attributes: nodes are labeled with partial attribute orders, transitions append one attribute. Red numbers next to nodes represent the episode number at which they are added when selecting attribute orders ABC,
BCA, CBA, ABC, ACB, CBA, CAB, and CBA (in that order).}}
    \label{fig:uct}
    \vspace*{-2em}
\end{figure}
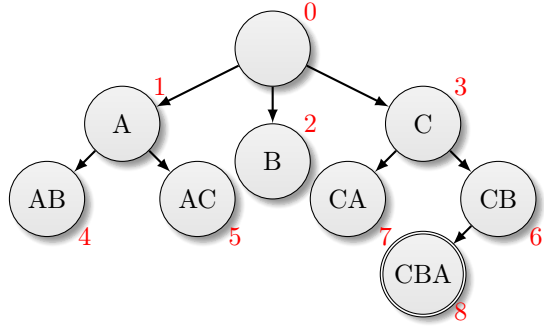

\begin{example}
\label{ex:expansion}
\rm
\revision{Given a query with three attributes (A, B, and C), assume that ADOPT selects the following attribute orders in the first episodes (some orders are selected in multiple episodes): ABC, BCA, CBA, ABC, ACB, CBA, CAB, CBA. Figure~\ref{fig:uct} shows the UCT search tree after those episodes. Nodes represent partial attribute orders and edges represent the addition of one attribute. Next to each node, in red, the figure shows the number of the episode in which the node was added. Initially (episode zero), the tree contains only the root node. In the first episode, ADOPT selects order ABC, adding a node for the first prefix (A) without corresponding node in the tree. Later, in episode four, ADOPT selects order ABC and, again, adds a node for the first prefix (AB) for which no node has been created. Once nodes are added, ADOPT starts collecting reward statistics for all attribute orders extending the corresponding prefix. These statistics are used to select attribute orders in future episodes.}
\end{example}

% \begin{algorithm}[t!]
% \caption{Reward Function.\label{alg:reward}}
% \renewcommand{\algorithmiccomment}[1]{// #1}
% \begin{small}
% \begin{algorithmic}[1]
% \State \textbf{Input:} Selected hypercube $c$, processed hypercube $p$
% \State \textbf{Output:} Reward $r$
% \Function{Reward}{$c, p$}
% % \State $c_i \gets c.u_i - c.l_i, \forall i = 1 \to n$
% % \State $\delta_i \gets p.u_i - p.l_i, \forall i = 1 \to n$
% \For{$i \gets 1$ to $m$}
%     \State $c_i \gets c.u_i - c.l_i$
%     \State $\delta_i \gets p.u_i - p.l_i$
% \EndFor
% \State $r \gets \sum_{1\leq i \leq m} (\delta_i / \Pi_{1 \leq k \leq m} (c_k) ) \frac{c.volume}{join\_space.volume}$
% % \State \Comment{Scale reward with the relative hypercube volume ratio}
% % \State $r \gets \frac{c.volume}{join\_space.volume}$ 
% \State \Return $r$
% \EndFunction
% \end{algorithmic}
% \end{small}
% \end{algorithm}

\subsection{Estimating Order Quality}
\label{sub:estimation}

%Those values serve as quality samples, judging the quality of an attribute order when processing one specific partition of the data. While quality estimates may vary for the same order, across different invocations, ADOPT converges to the order with maximal average quality over time. 

The reinforcement learning, described in Section~\ref{sub:rl}, is guided by reward values. Next, we discuss the definition of the reward function. Before that, we introduce an auxiliary function, measuring the volume of a cube as the product of range sizes over all dimensions:

\begin{equation}
    Volume(c)=\prod_{i}(c.u_i-c.l_i)
\end{equation}

With a slight abuse of notation, we  write $Volume(q)$ to denote the volume of the cube, spanned by all join attributes of a query $q$. 

In order to fully process a query, ADOPT must cover the cube representing the entire space of attribute value combinations. Hence, the more volume of that cube we cover per time unit, the faster query processing is. \revision{Even for a fixed attribute order, the volume processed per time unit may vary across different parts of the data (e.g., since the number of result tuples per volume varies). However, the fastest order processes most volume in average, averaging over the entire data set, and the UCT algorithm converges to decisions with highest average reward, even if the reward function is noisy~\cite{Kocsis2006}.} This implies that volume covered is a useful measure of progress. The reward function, presented next, follows that intuition. Given a set of processed cubes $P$ for query $q$, it uses the aggregate volume covered, scaled to the total volume to process (scaling ensures reward values between zero and one, consistent with the requirements of the UCT algorithm):

\begin{equation}
    Reward(P,\revision{q})=(\sum_{p\in P}Volume(p))/Volume(q)
\end{equation}

%\junxiong{However, $Reward(P, c)$ is not a precious measurement of the execution progress. During the execution, ADOPT splits cubers into smaller unbalanced cubes and small cubes are more likely to finish during the timeout. Thus, the progress of a target cube $c$ can not reflect the quality of selected join order. To tackle this issue, we scale this reward using the fraction between the volume of target cube and the volume of entire cube. The final reward $r$ is calculate as }

% \begin{equation}
%     r =  Reward(P,c) Volume(c) / Volume(Entire Cube)
% \end{equation}

% \junxiong{Directly using $Reward(P, c)$ as final reward $r$ to update the statistics of UCT leads  }
\section{Analysis}
\label{sec:analysis}

% In this section, we prove three deseirable properties of ADOPT. After proving correctness (Section~\ref{sub:correct}), we show that ADOPT maintains worst-case optimality guarantees (Section~\ref{sub:wcoproof}), then we prove that ADOPT converges to optimal attribute orders (Section~\ref{sub:convergence}).

In Section~\ref{sub:convergence}, we prove that ADOPT converges to optimal attribute orders. Two further properties, correctness and worst-case optimality, are analyzed in Sections 
\ref{sub:correct} and respectively \ref{sub:wcoproof}.

 \subsection{Convergence to Optimal Orders}
 \label{sub:convergence}
First, we show that ADOPT must finish processing once the accumulated rewards reach a precise threshold.
% Different reward functions can be used. For instance, given a variable order, denote by $l$ and $u$ lower and upper bound on the value domain of the first variable. The LFTJ algorithm iterates over variable values in ascending order. Let $d=e-s$ the distance between the value for the first variable (e.g., the maximal value over all iterators) at the end ($e$) and at the start ($s$) of a time slice. Let $r=d/(u-b)$ be the relative distance covered (relative to the entire value range). For the following lemma, we assume that relative distance is used as reward function. 

\begin{theorem}\label{th:threshold}
Join processing finishes once the sum of accumulated rewards over all threads and episodes reaches one.
\end{theorem}
\begin{proof}
Reward is proportional to the volume of the cube covered, scaled to the size of the full cube. Hence, accumulating a reward sum of one means that a volume equal to the full cube has been processed. Furthermore, ADOPT avoids covering overlapping cubes by different threads and in different episodes (independently of the attribute order). Hence, once the accumulated reward reaches one, processed cubes must cover the full cube.
\end{proof}

This implies that the reward function is a good measure of attribute order quality indeed.

\begin{theorem}
The attribute order with the highest average reward per episode minimizes the number of computational steps.\label{th:rewardgood}
\end{theorem}
\begin{proof}
For any attribute order $o$, processing finishes once the accumulated rewards reach one (Theorem~\ref{th:threshold}). Therefore, the average reward  $r_o$ per episode for $o$ is inversely proportional to the number of episodes $e_o$ needed by $o$, i.e.\ $r_o=1/e_o$. Also, the number of computational steps per episode is constant. Therefore, minimizing the number of episodes needed maximizes the average reward.
\end{proof}

This implies convergence to optimal attribute orders.

\begin{corollary}
ADOPT converges to an optimal attribute order.
\end{corollary}
\begin{proof}
Following Theorem~\ref{th:rewardgood}, the order with the highest average reward is also the fastest one to process. Furthermore, the UCT algorithm used by ADOPT converges to a solution with maximal expected reward~\cite{Kocsis2006}. Hence, ADOPT converges to an attribute order that minimizes the number of processing steps.
\end{proof}

\subsection{Correctness}
\label{sub:correct}

We show that ADOPT generates a complete and correct result.

\begin{theorem}
ADOPT does not produce incorrect join results.\label{th:noincorrect}
\end{theorem}
\begin{proof}
ADOPT inserts join results in Line~5 of Algorithm~\ref{alg:joinOneCube}. For each attribute, Algorithm~\ref{alg:joinOneCube} only iterates over values that appear in all relations with that attribute (Line~10 in Algorithm~\ref{alg:joinOneCube}). Hence, join results must satisfy all equality join conditions. Furthermore, Algorithm~\ref{alg:joinOneCube} is only applied to input tuples satisfying all unary predicates (due to the filter in Line~5 in Algorithm~\ref{alg:main}). Hence, join results satisfy all applicable predicates and are correct.
\end{proof}

%Join results are incorrect if they ADOPT considers only combinations of tuples from the join tables that satisfy all unary predicates (due to the pre-processing stage). Hence, tuples combinations that do not satisfy join predicates are the only possibility to generate incorrect output. However, the anytime LFTJ algorithm only advances from one attribute to the next, if all selected tuples have the same value for the current attribute. In that case, the associated equality join condition is satisfied. Hence, result tuples are only generated if all join conditions are satisfied. Hence, the output is correct.

\begin{lemma}
All results contained within processed cubes, returned by Algorithm~\ref{alg:joinOneCube}, have been inserted into the result set.\label{lm:processedcorrect}
\end{lemma}
\begin{proof}
Assume there is a vector $r$ of join attribute values, matching all join predicates, that is contained in a processed cube $p$ but not in the join result. There is an attribute $a$ such that $r$ equals the last selected attribute values $v$ up to attribute $a$ (in attribute order), then takes a value below the last selected value for the next attribute. However, Algorithm~\ref{alg:joinOneCube} does not advance from one value to the next for an attribute, before considering all value combinations for the remaining attributes. Hence, $r$ must have been added to the result, leading to a contradiction.
\end{proof}

\begin{lemma}
The task manager only removes processed cubes.\label{lm:removalcorrect}
\end{lemma}
\begin{proof}
In each invocation of  \textproc{TS.Remove}, the task manager removes only the target cube. Assume a vector $l$ of attribute values, within the target cube, is ``lost'', i.e.\ it is neither contained in any processed cube nor in any of the newly added, unprocessed cubes. Denote by $v$ the last selected values for all attributes in the join invocation, immediately preceding removal, and by $o$ the corresponding attribute order. Assume $l$ matches the values in $v$ for some prefix (possibly of size zero) of order $o$. Denote by $a$ the first attribute in $o$ for which $l$ does not match $v$. Denote by $p$ the processed cube added when reaching $a$ in the loop from Line~30 in Algorithm~\ref{alg:joinOneCube}. If $l_a<v_a$ then $l$ must be contained in $p$. However, if $l_a>v_a$, $l$ must be contained in the unprocessed cube added when reaching $p$ in the loop from Line~26 in Algorithm~\ref{alg:cubes}, leading to a contradiction.
\end{proof}

\begin{theorem}
ADOPT produces a complete join result.\label{th:complete}
\end{theorem}
\begin{proof}
The join phase terminates only once no unprocessed cubes are left. Result tuples contained in processed cubes are inserted into the result set (Lemma~\ref{lm:processedcorrect}) and no unprocessed cubes are erroneously removed (Lemma~\ref{lm:removalcorrect}). Hence, processing cannot terminate before all result tuples are inserted.
\end{proof}

%  more target cubes can be retrieved from the constraint store. This is only possible once the set of cubes, processed via joins, covers the cube representing the full Cartesian product among all joined tables. However, the Cartesian product of all joined tables forms a super set of any valid join result. Therefore, the join result must be complete after termination.

The next result follows immediately.

\begin{corollary}
ADOPT produces a correct join result.
\end{corollary}
\begin{proof}
ADOPT produces all tuples in the join result (Theorem~\ref{th:complete}) without generating any incorrect tuple (Theorem~\ref{th:noincorrect}).
\end{proof}

% \begin{theorem}
% WCOL produces the correct query result.
% \end{theorem}

\subsection{Worst-Case Optimality}
\label{sub:wcoproof}

We analyze whether ADOPT maintains worst-case optimality guarantees. We focus on join processing overheads, neglecting without loss of generality the preparation overheads. These overheads include the sorting of the relations to support LFTJ leapfrogging following the orders of attributes picked by ADOPT. Our analysis is based on the worst-case optimality properties of LFTJ and makes the same assumptions as the corresponding proof~\cite{DBLP:conf/icdt/Veldhuizen14} (e.g., restriction to equality joins). We consider the number of threads a constant. Additionally, our analysis makes the following assumption.

%Hence, it is valid under the same assumptions as the LFTJ analysis Also, we focus on queries that contain only equality joins but no other join predicates. 

\begin{assumption}\label{as:constantepisodes}
We assume that the number of episodes is bounded by a constant that does not depend on the data size.
\end{assumption}

% \begin{assumption}
% We assume that all management overheads due to adaptive processing, including time required for selecting target cubes and initializing execution state at the end of each time slice, are negligible compared to join processing overheads.
% \end{assumption}

% The latter assumption is reasonable for the following reason. Management overheads are incurred only when the attribute order changes. However, we can freely choose the number of steps between join order changes. If management overheads are non-negligible, we can simply increase the step count. 

The latter assumption can be ensured by increasing the number of steps per episode, proportional to the maximal output size. 

\begin{lemma}\label{lemma:adopt-cubes}
The number of cubes processed by ADOPT does not depend on the data size.
\end{lemma}
\begin{proof}
Initially, the number of unprocessed cubes is proportional to the number of threads (i.e., constant). Each invocation of Procedure~\textproc{CS.Remove} may create up to $m$ cubes where $m$ is the number of query attributes (i.e., a constant). Each episode may process multiple cubes, i.e.\ invoke that function multiple times. However, whenever the target cube was fully processed, no unprocessed cubes are added. There can be at most one target cube per episode and thread that was not fully processed. Hence, the number of unprocessed cubes added per episode is bounded by a constant. Due to Assumption~\ref{as:constantepisodes}, the number of generated (and processed) cubes is therefore bounded by a constant as well.
\end{proof}

\begin{lemma}\label{lemma:adopt-dominate}
The time complexity of ADOPT's join phase is dominated by the complexity of \textproc{JoinOneCube}.
\end{lemma}
\begin{proof}
The per-episode complexity of all operations of the reinforcement learning algorithm are bounded by the number of join attributes. Similarly, the number of operations required to retrieve or remove cubes is bounded by the number of join  attributes. Hence, the time complexity for join processing dominates.
\end{proof}

We are now ready to prove our main result.

\begin{theorem}
ADOPT is worst-case optimal.
\end{theorem}
\begin{proof}
The number of cubes processed by ADOPT does not depend on the input data size (Lemma~\ref{lemma:adopt-cubes}). Furthermore, time complexity for joins dominates (Lemma~\ref{lemma:adopt-dominate}). ADOPT uses the LFTJ algorithm to process cubes. This algorithm is worst-case optimal~\cite{DBLP:conf/icdt/Veldhuizen14}. The time for processing the largest cube, which occurs over the execution of a query, is upper-bounded by the time required by LFTJ for processing the entire query cube. The number of cubes is independent of the data size. Hence, total processing overheads are asymptotically equivalent to the time required by LFTJ, therefore worst-case optimal.
\end{proof}

\section{Experimental Evaluation}
\label{sec:experiments}

% we vary the benchmark scenario and test the robustness of our system.
% compare ADOPT to multiple baselines (Section~\ref{sub:baselines}). In Section~\ref{sub:variants}, we compare the hypercube approaches proposed by us and prefix share approaches proposed by ~\cite{trummer2019skinnerdb}.

%ADOPT outperforms its competitors on all graph benchmarks by 2-30x. For the JOB benchmark, ADOPT is close (0.91x) to MonetDB and outperforms the others by 1.4-6.3x.

We confirm experimentally that ADOPT outperforms a range of competitors for both acyclic and cyclic queries from the join order benchmark~\cite{Gubichev2015}, \revision{standard decision support benchmarks (i.e., TPC-H and JCC-H ~\cite{boncz2018jcc}),} and graph data~\cite{snapnets,nguyen2015join} workloads. The robustness of ADOPT's query evaluation becomes more evident for queries with an increasingly larger number of joins and with filter conditions whose joint selectivity is hard to assess correctly at optimization time.
\nop{
 We confirm experimentally that our query engine ADOPT outperforms a variety of competitors in both runtime (Sec.~\ref{sub:baselines}) and robustness (Sec.~\ref{sub:robustness}) for graph data workloads involving cyclic queries on a number of publicly-available datasets that are typically used in the literature \cite{snapnets} for benchmarking graph queries~\cite{nguyen2015join}.
 }
The superior performance of ADOPT over its competitors is due to the interplay of its four key features: worst-case optimal join evaluation; reinforcement learning that eventually converges to near-optimal attribute orders (Sec.~\ref{sub:optimalattributeorder}); hypercube data decomposition (Sec.~\ref{sub:variants}); and domain parallelism (Sec.~\ref{sub:parallelization}). \revision{Appendices~\ref{app:comparison_system_X}-\ref{sec:attributescalability}  report further experiments on: the performance comparison of ADOPT and System-X, memory consumption, sorting and synchronization overhead, and scalability with the number of join attributes per table.}

\nop{
We introduce our experimental setups in Section~\ref{sub:expsetup} and benchmark ADOPT and multiples systems on various graph datasets in ~\ref{sub:baselines}. In Section~\ref{sub:variants}, we compare our hypercube approach and alternatives using the prefix share progress tracker and the offset progress tracker. In Section~\ref{sub:ablation}, we study the attribute order of ADOPT and break down the execution time. In Section~\ref{sub:parallelization}, we scale up ADOPT system and evaluate its speedup.
}

\subsection{Experimental Setup}
\label{sub:expsetup}

We benchmark the query engines on acyclic and cyclic queries.

\subsubsection*{Benchmark for acyclic queries} 
The join order benchmark (JOB)~\cite{Gubichev2015} consists of 113 queries over the highly-correlated IMDB real-world dataset. This benchmark shows an orders-of-magnitude performance gap between different join orders for the same query. \revision{TPC-H (JCC-H \cite{boncz2018jcc}) is a benchmark used for decision support, comprising of 22 queries that incorporate standard SQL predicates. In TPC-H, data is synthetically generated with uniform distribution, whereas in JCC-H, the data is highly skewed, which makes JCC-H a harder benchmark to optimize. In our experiments, we use TPC-H/JCC-H with scaling factor ten. We omit four queries in TPC-H (JCC-H) queries, Q2, Q13, Q15, and Q22, for lack of support for non-integer join columns, outer joins, views, and substring functions.}

\subsubsection*{Benchmark for cyclic queries}
We follow prior work on benchmarking worst-case optimal join algorithms against traditional join plans~\cite{nguyen2015join} and consider the evaluation of clique and cycle queries over the binary edge relations of four graph datasets from the SNAP network collection~\cite{snapnets}. The considered queries are as follows:

\begin{itemize}
    \item $n$-clique: Compute the cliques of  $n$ distinct vertices. Such a clique has an edge between any two of its vertices. For instance, the 3-clique is the triangle:
    
    $edge(a,b), edge(b,c), edge(a,c), a<b<c$
    
    \item $n$-cycle: Compute the cycles of  $n$ distinct nodes. Such a cycle has an edge between the $i$-th and the $(i+1)$-th vertices for $1\leq i< n$ and an edge between the first and the last vertices. For instance, the $4$-cycle query is: 
    
    $edge(a,b), edge(b,c), edge(c,d), edge(a,d), a<b<c<d$
    
    % \item $(m,n)$-lollipop: Count the number of distinct $m$-length paths conected to an $n$-clique, where the start nodes are sampled at random. The $(2,3)$-lollipop with the start nodes from the unary relation $v_1$ is:
    
    % $v_1(a), edge(a, b), edge(b, c), edge(c, d), edge(d, e), edge(c, e)$
    
    % \item $n$-barbell: Count the number of distinct subgraphs consisting of two non-overlapping $n$-vertex cliques connected by an edge. The  $3$-barbell query is: 
    
    % $edge(a,b), edge(b,c), edge(a,c), edge(a, d), edge(d,e),$ 
    
    % $edge(e,f), edge(d,f)$
\end{itemize}

The inequalities in the above queries enforce that each node in the clique/cycle is distinct. 
Instead of returning the list of all distinct cliques/cycles, all systems are instructed to return their count. ADOPT  counts the result tuples as they are computed.
The reason for returning the count is to avoid the time to list the result tuples and only report the time to compute them.

\subsubsection*{Systems}
ADOPT is implemented in JAVA (jdk 1.8). It uses 10,000 steps per episode and UCT exploration ratio 1E-6. The competitors are: the open-source engines MonetDB~\cite{boncz2008breaking} (Database Server Toolkit v11.39.7, Oct2020-SP1) and PostgreSQL 10.21~\cite{postgres} that employ traditional join plans; a commercial engine System-X (implemented in C++)\nop{LogicBlox (version 4.40.0)} that uses the worst-case optimal LFTJ  algorithm~\cite{DBLP:conf/icdt/Veldhuizen14}; the open-source engine EmptyHeaded that uses a worst-case optimal join algorithm~\cite{Aberger2016}; and SkinnerDB~\cite{trummer2019skinnerdb} (implemented in Java jdk 1.8) that uses reinforcement learning to learn an optimal join order for traditional query plans.

% 1.00	1.00	1.00	1.00	1.00
% 6.38	5.09	2.90	2.55	2.42
% 0.00	1.54	2.64	4.69	2.09
% 7.11	15.35	17.94	7.23	3.84
% 0.91	14.99	22.03	6.38	3.65
% 1.44	15.71	33.00	10.29	3.85

%PostgreSQL was allocated enough memory to evaluate its query plans in memory. 

\subsubsection*{Setup}
We run each experiment five times and report the average execution time. We used a server with 2 Intel Xeon Gold 5218 CPUs with 2.3 GHz (32 physical cores)/384GB RAM/512GB hard disk. ADOPT, EmptyHeaded, MonetDB, SkinnerDB, and System-X were set to run in memory. By default, all engines use 64 threads\nop{ (using hyperthreading); for ADOPT, we also investigate its runtime as a function of the number of threads (Sec.~\ref{sub:parallelization})}. \revision{For all systems, we create indexes to optimize performance (index creation overheads are reported separately in Appendix~\ref{sec:indexcreation}). For systems such as MonetDB that create indexes automatically, based on properties of observed queries, we perform one warm-up run before starting our measurements.}

\subsection{Runtime Performance}
\label{sub:baselines}

\nop{We verified experimentally the following two hypotheses.}

%%%%%%%%%%%%%%%%%%%%
\begin{table*}[t]
\centering 
\caption{Overall runtime (in seconds) to compute all queries for each benchmark. For the JOB benchmark, ">" indicates the time is only for some of the 113 queries. For the four graph datasets, ">" indicates the time exceeded the six-hour (21,600 seconds) timeout for some of the cyclic queries. The multiplicative factors in parentheses after the runtimes of systems are the speedups of ADOPT over these systems.}
\resizebox{1\textwidth}{!}{
\begin{tabular}{l|rrrrrrr}
\toprule[1pt] Systems & JOB & ego-Facebook & ego-Twitter & soc-Pokec & soc-Livejournal1 & \revision{TPC-H} & \revision{JCC-H} \\
\midrule[1pt]
ADOPT       & 45   & \textbf{4,414}     & \textbf{3,931}   & \textbf{9,268}     & \textbf{26,350} & 141 & \textbf{194} \\
System-X    & > 287 (6.38x)    & > 22,459 (5.09x)  & 11,384 (2.90x) & > 23,623 (2.55x)  & > 63,878 (2.42x) & -- & --  \\
EmptyHeaded & --               & 6,783  (1.54x)    & 10,381 (2.64x) & > 43,444 (4.69x)  & > 55,144 (2.09x) & -- & -- \\
PostgreSQL  & 285 (6.33x)      & > 67,774 (15.35x) & > 70,515 (17.94x)  & > 67,016 (7.23x)    & > 101,193 (3.84x) & 182 (1.53x)  & > 216,122 \\
&&&&&&& (1,114x) \\
MonetDB     & \textbf{41} (0.91x) & > 66,165 (14.99x) & > 86,596 (22.03x) & > 59,131 (7.23x)  & > 96,222 (3.84x)  & \textbf{17} (0.12x) & > 216,035 \\
&&&&&&& (1,114x) \\
SkinnerDB   & 65 (1.44x)       & > 69,366 (15.71x) & > 129,741 (33.00x) & > 95,374  (10.29x)  & > 101,392 (3.85x) & 173 (1.23x) & 320 (1.6x) \\
\bottomrule[1pt]
\end{tabular}
}
\label{tab:overall}
\end{table*}
%%%%%%%%%%%%%%%%%%%%%

\revision{ADOPT puts together worst-case optimal join algorithm, which is primarily motivated by cyclic queries, and adaptive processing, which is motivated by scenarios in which size and cost prediction for query planning is difficult (e.g., due to data skew or complex queries). This motivates the following hypotheses.}

\begin{hyp}
    \revision{ADOPT outperforms baselines without worst-case optimal join algorithms on cyclic queries.}
\end{hyp}

\begin{hyp}
    \revision{ADOPT outperforms non-adaptive baselines for complex queries on skewed data.}
\end{hyp}

\begin{hyp}\label{hyp:bad}
    \revision{ADOPT performs worse, compared to baselines, if queries are simple, acyclic, and are executed on uniform data.}
\end{hyp}

% \begin{hyp}
% ADOPT demonstrates comparable or better performance than its competitors for both cyclic and acyclic queries.
% \end{hyp}

\pgfplotstableread[col sep=comma,]{data/facebook_twitter_clique_queries.csv}\fbtwittercliquequery
\pgfplotstableread[col sep=comma,]{data/pokec_livejournal_clique_queries.csv}\pokeclivejournalcliquequery
\pgfplotstableread[col sep=comma,]{data/cycle_queries.csv}\cyclequery

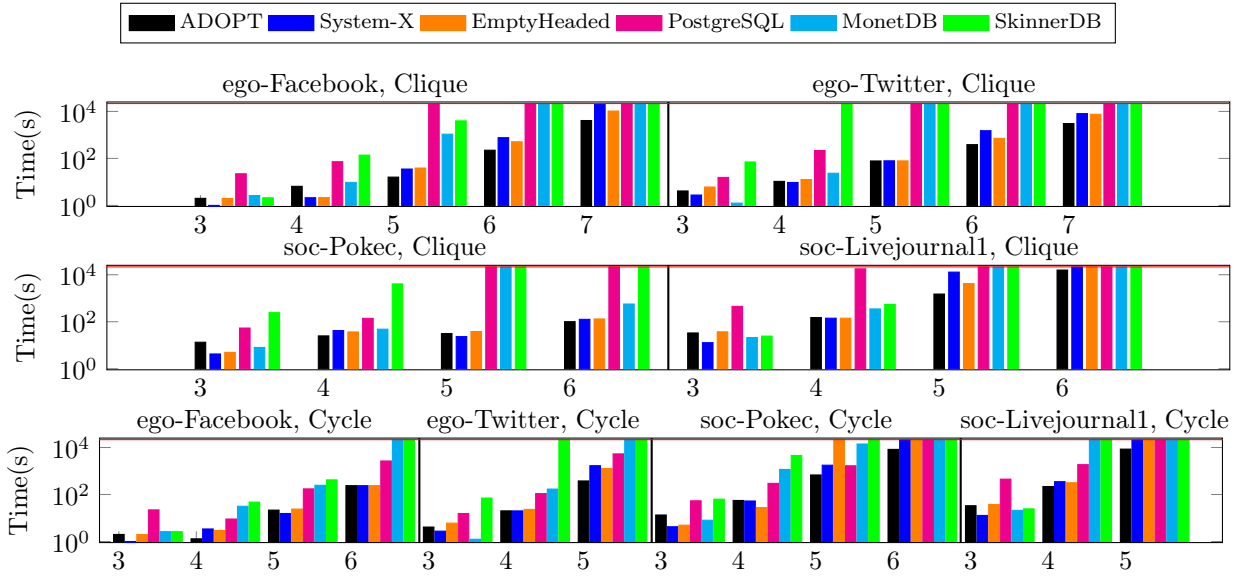
\begin{figure*}[t]
\centering
\ref{baselineLegend}
\vspace{-2.5em}
\begin{tikzpicture}
    \pgfmathsetmacro{\BarOffset}{0.7}
    \begin{axis}[
        width=\textwidth,
        height=0.18\textwidth,
        ymin=0.9,
        ymax=25200,
        ymode=log,
        log basis y={10},
        xtick=data,
        xtick distance=10,
        xtick pos=bottom,
        % x tick label style={anchor=north west,align=right},
        xticklabels from table={\fbtwittercliquequery}{Query},
        ylabel={Time(s)},
        xlabel style={
            yshift=-2ex,
            anchor=west,
        },
        legend entries={ADOPT, System-X, EmptyHeaded, PostgreSQL, MonetDB, SkinnerDB}, legend columns=6, legend to name=baselineLegend, legend style={font=\footnotesize}, legend style={cells={align=left,anchor=west}},
        area legend,
        bar width=4pt,
        clip mode=individual,
    ]

    % plot the "black" ybars
    \addplot [
        ybar,
        draw=black,
        fill=black,
    ] table [
        x expr=\coordindex,
        y=ADOPT
    ] {\fbtwittercliquequery};
        
    % plot the "blue" ybars
    \addplot [
        ybar,
        draw=blue,
        fill=blue,
    ] table [
        x expr=\coordindex+0.2*\BarOffset,
        y=Logicblox
    ] {\fbtwittercliquequery};

    % plot the "cyan" ybars
    \addplot [
        ybar,
        draw=orange,
        fill=orange,
    ] table [
        x expr=\coordindex+0.4*\BarOffset,
        y=EmptyHeaded
    ] {\fbtwittercliquequery};

    % plot the "magenta" ybars
    \addplot [
        ybar,
        draw=magenta,
        fill=magenta,
    ] table [
        x expr=\coordindex+0.6*\BarOffset,
        y=PostgreSQL
    ] {\fbtwittercliquequery};

    % plot the "cyan" ybars
    \addplot [
        ybar,
        draw=cyan,
        fill=cyan,
    ] table [
        x expr=\coordindex+0.8*\BarOffset,
        y=MonetDB
    ] {\fbtwittercliquequery};
    
     % plot the "cyan" ybars
    \addplot [
        ybar,
        draw=green,
        fill=green,
    ] table [
        x expr=\coordindex+\BarOffset,
        y=SkinnerDB
    ] {\fbtwittercliquequery};

    \begin{scope}[=
        every label/.append style={
            label distance=2ex,
        },
    ]
        \node [label=below:{ego-Facebook, Clique}]
            at (axis cs:1.5, 2500000) {};
        \node [label=below:{ego-Twitter, Clique}]
            at (axis cs:7.5, 2500000) {};
            
        \draw[thick] ({axis cs:4.85,0}|-{rel axis cs:0,0}) -- ({axis cs:4.85,0}|-{rel axis cs:0,1});
        
    \end{scope}
    
    \addplot[red,sharp plot,update limits=false,]
    coordinates {(-50, 21600) (50, 21600)}
    node[midway,below,font=\bfseries\sffamily]{};
    \end{axis}
\end{tikzpicture}
\vspace{-2.5em}
\begin{tikzpicture}
    \pgfmathsetmacro{\BarOffset}{0.6}
    \begin{axis}[
        width=\textwidth,
        height=0.18\textwidth,
        ymin=0.9,
        ymax=25200,
        ymode=log,
        log basis y={10},
        xtick=data,
        xtick distance=10,
        xtick pos=bottom,
        % x tick label style={rotate=-30,anchor=north west,align=left},
        xticklabels from table={\pokeclivejournalcliquequery}{Query},
        ylabel={Time(s)},
        xlabel style={
            yshift=-2ex,
        },
        bar width=4pt,
        clip mode=individual,
    ]

    % plot the "black" ybars
    \addplot [
        ybar,
        draw=black,
        fill=black,
    ] table [
        x expr=\coordindex,
        y=ADOPT
    ] {\pokeclivejournalcliquequery};

    % plot the "blue" ybars
    \addplot [
        ybar,
        draw=blue,
        fill=blue,
    ] table [
        x expr=\coordindex+0.2*\BarOffset,
        y=Logicblox
    ] {\pokeclivejournalcliquequery};

    % plot the "cyan" ybars
    \addplot [
        ybar,
        draw=orange,
        fill=orange,
    ] table [
        x expr=\coordindex+0.4*\BarOffset,
        y=EmptyHeaded
    ] {\pokeclivejournalcliquequery};

    % plot the "magenta" ybars
    \addplot [
        ybar,
        draw=magenta,
        fill=magenta,
    ] table [
        x expr=\coordindex+0.6*\BarOffset,
        y=PostgreSQL
    ] {\pokeclivejournalcliquequery};

    % plot the "cyan" ybars
    \addplot [
        ybar,
        draw=cyan,
        fill=cyan,
    ] table [
        x expr=\coordindex+0.8*\BarOffset,
        y=MonetDB
    ] {\pokeclivejournalcliquequery};
    
     % plot the "cyan" ybars
    \addplot [
        ybar,
        draw=green,
        fill=green,
    ] table [
        x expr=\coordindex+\BarOffset,
        y=SkinnerDB
    ] {\pokeclivejournalcliquequery};
    
    \begin{scope}[=
        every label/.append style={
            label distance=2ex,
        },
    ]
        
    \node [label=below:{soc-Pokec, Clique}]
        at (axis cs:1.5,2500000) {};
    \node [label=below:{soc-Livejournal1, Clique}]
        at (axis cs:6,2500000) {};
            
    \draw[thick] ({axis cs:3.8,0}|-{rel axis cs:0,0}) -- ({axis cs:3.8,0}|-{rel axis cs:0,1});
            
    \end{scope}
    
    \addplot[red,sharp plot,update limits=false,]
    coordinates {(-50, 21600) (50, 21600)}
    node[midway,below,font=\bfseries\sffamily]{};
    
    \end{axis}
\end{tikzpicture}
% \vspace{-2em}
\begin{tikzpicture}
    \pgfmathsetmacro{\BarOffset}{0.5}
    \begin{axis}[
        width=\textwidth,
        height=0.18\textwidth,
        ymin=0.9,
        ymax=25200,
        ymode=log,
        ylabel near ticks,
        xmin=-0.5,
        xmax=14,
        log basis y={10},
        xtick=data,
        xtick distance=10,
        xtick pos=bottom,
        % x tick label style={rotate=-30,anchor=north west,align=left},
        xticklabels from table={\cyclequery}{Query},
        % extra x ticks={3.8},
        % extra x tick labels={},
        % extra x tick style={
        %     grid=major,
        %     ultra thick line,
        % },
        % ylabel={Execution Time(s)},
        ylabel={Time(s)},
        xlabel style={
            yshift=-2ex,
        },
        legend cell align=left,
        area legend,
        bar width=4pt,
        clip mode=individual,
    ]
    
    % plot the "blue" ybars
    \addplot [
        ybar,
        draw=black,
        fill=black,
    ] table [
        x expr=\coordindex-0.5*\BarOffset,
        y=ADOPT
    ] {\cyclequery};

    % plot the "black" ybars
    \addplot [
        ybar,
        draw=blue,
        fill=blue,
        % pattern color=black,
        % pattern=north east lines,
    ] table [
        x expr=\coordindex-0.2*\BarOffset,
        y=Logicblox
    ] {\cyclequery};

    \addplot [
        ybar,
        draw=orange,
        fill=orange,
    ] table [
        x expr=\coordindex+0.1*\BarOffset,
        y=EmptyHeaded
    ] {\cyclequery};
        
    % plot the "magenta" ybars
        \addplot [
            ybar,
            draw=magenta,
            fill=magenta,
            % pattern color=magenta,
            % pattern=vertical lines,
        ] table [
            x expr=\coordindex+0.4*\BarOffset,
            y=PostgreSQL
        ] {\cyclequery};

    % plot the "cyan" ybars
        \addplot [
            ybar,
            draw=cyan,
            fill=cyan,
            % pattern color=cyan,
            % pattern=grid,
        ] table [
            x expr=\coordindex+0.7*\BarOffset,
            y=MonetDB
        ] {\cyclequery};
    
     % plot the "cyan" ybars
        \addplot [
            ybar,
            draw=green,
            fill=green,
            % pattern color=cyan,
            % pattern=grid,
        ] table [
            x expr=\coordindex+\BarOffset,
            y=SkinnerDB
        ] {\cyclequery};

        \begin{scope}[=
            every label/.append style={
                label distance=2ex,
            },
        ]
            \node [label=below:{ego-Facebook, Cycle}]
                at (axis cs:1.5, 2500000) {};
            \node [label=below:{ego-Twitter, Cycle}]
                at (axis cs:5.1, 2500000) {};
            \node [label=below:{soc-Pokec, Cycle}]
                at (axis cs:8.5, 2500000) {};
            \node [label=below:{soc-Livejournal1, Cycle}]
                at (axis cs:12.35, 2500000) {};
                
            \draw[thick] ({axis cs:3.63,0}|-{rel axis cs:0,0}) -- ({axis cs:3.63,0}|-{rel axis cs:0,1});
            
            \draw[thick] ({axis cs:6.63,0}|-{rel axis cs:0,0}) -- ({axis cs:6.63,0}|-{rel axis cs:0,1});
            
            \draw[thick] ({axis cs:10.62,0}|-{rel axis cs:0,0}) -- ({axis cs:10.62,0}|-{rel axis cs:0,1});
            
        \end{scope}
    
    \addplot[red,sharp plot,update limits=false,]
    coordinates {(-50, 21600) (50, 21600)}
    node[midway,below,font=\bfseries\sffamily]{};

    \end{axis}
\end{tikzpicture}
\vspace*{-1em}
\caption{The execution (wall-clock) time for clique and cycle queries on four graphs (the x axis represents query size).}
\label{fig:clique_cycle_results}
\end{figure*}

Table~\ref{tab:overall} reports the total time in seconds for different systems and benchmarks. \revision{ADOPT performs best for the four benchmarks on graphs, featuring cyclic queries. Figure~\ref{fig:clique_cycle_results} breaks those results down by query size and query type. Compared to other baselines using worst-case optimal joins, ADOPT's gains derive from larger queries with more predicates, creating the potential for inter-predicate correlations that are hard to predict. This makes it difficult to select optimal attribute orders before execution. PostgreSQL, MonetDB, and SkinnerDB suffer from over-proportionally large intermediate results when processing cyclic queries as they do not implement worst-case optimal joins.}

\revision{The join order benchmark (JOB) features acyclic queries but non-uniform data (i.e., it contains some elements that should benefit ADOPT in the comparison and some that have the opposite effect). Here, ADOPT performs comparably but slightly worse to the best baseline: MonetDB. For System-X, Table~\ref{tab:overall} only reports time for executing a subset of the queries (39 out of 113). The remaining queries have IS/NOT NULL and IN predicates that are not supported by System-X.} EmptyHeaded needs more than five days to construct the data indices (tries) it needs for the non-binary JOB tables so we were not able to report its runtime on the JOB queries.

\revision{TPC-H and JCC-H share the same query templates and database schema but differ in the database content: TPC-H uses uniform data whereas JCC-H uses highly correlated data. On TPC-H, MonetDB performs best and outperforms ADOPT significantly. This is consistent with prior work~\cite{Aberger2018}, showing that systems with worst-case optimal joins (specifically: the LFTJ that ADOPT uses internally) perform significantly worse than MonetDB on TPC-H. Given those prior results and limited support for TPC-H queries in System~X and EmptyHeaded, we compare only to MonetDB as the strongest baseline. Besides drawbacks due to the join algorithm, ADOPT incurs overheads due to adaptive processing which is unnecessary on TPC-H: predicting sizes of intermediate results and plan execution cost is relatively easy due to uniform data.}

\revision{On the other hand, ADOPT outperforms all other systems on JCC-H. Despite sharing the same query templates with TPC-H, JCC-H makes query  optimization hard due to highly correlated data. Here, both adaptive baselines (SkinnerDB and ADOPT) benefit, with ADOPT being significantly faster, whereas all other systems reach the timeout of six hours. This means, even on acyclic queries, traditionally not considered the sweet spot for LFTJ-based joins~\cite{Aberger2018}, ADOPT is preferable if data is sufficiently correlated.}

\begin{figure}
\centering
%\hspace*{-0.5cm}
\begin{tikzpicture}
    
    \begin{groupplot}[
             ylabel near ticks,
             view={0}{90},
             point meta min=0,
             point meta max=3,
             width = 4 cm,
             height = 3 cm,
             colormap={custom}{color(0)=(blue) color(1)=(white) color(3)=(red)},
             group style = {group name=robustness plots, group size = 2 by 3,}
        ]
        
        \nextgroupplot
    
        \addplot[scatter, only marks] table [x index=0, y index=1, scatter src=\thisrowno{2}] {data/twitter_graph_three_clique_selectivity.csv};
        
        \coordinate (top) at (rel axis cs:0,1);
        
        \nextgroupplot
    
        \addplot[scatter, only marks] table [x index=0, y index=1, scatter src=\thisrowno{2}] {data/twitter_graph_four_clique_selectivity.csv};
        
        \nextgroupplot
    
        \addplot[scatter, only marks] table [x index=0, y index=1, scatter src=\thisrowno{2}] {data/twitter_graph_five_clique_selectivity.csv};
        
        \nextgroupplot
    
        \addplot[scatter, only marks] table [x index=0, y index=1, scatter src=\thisrowno{2}] {data/twitter_graph_four_cycle_selectivity.csv};
        
        \nextgroupplot
    
        \addplot[scatter, only marks] table [x index=0, y index=1, scatter src=\thisrowno{2}] {data/twitter_graph_five_cycle_selectivity.csv};
        
        \nextgroupplot
        
        \addplot[scatter, only marks] table [x index=0, y index=1, scatter src=\thisrowno{2}] {data/twitter_graph_six_clique_selectivity.csv};

        \coordinate (bot) at (rel axis cs:1,0);

    \end{groupplot}
    
    \node[text width=6cm,align=center,anchor=north] at ([yshift=0mm]robustness plots c1r1.south) {\caption*{3 clique}};
    
    \node[text width=6cm,align=center,anchor=north] at ([yshift=0mm]robustness plots c2r1.south) {\caption*{4 clique}};
    
    \node[text width=6cm,align=center,anchor=north] at ([yshift=0mm]robustness plots c1r2.south) {\caption*{4 cycl}};
    
    \node[text width=6cm,align=center,anchor=north] at ([yshift=0mm]robustness plots c2r2.south) {\caption*{5 clique \label{subplot:5_clique_robustness}}};
    
    \node[text width=6cm,align=center,anchor=north] at ([yshift=0mm]robustness plots c1r3.south) {\caption*{5 cycle \label{subplot:5_cycle_robustness}}};
    
    \node[text width=6cm,align=center,anchor=north] at ([yshift=0mm]robustness plots c2r3.south) {\caption*{6 clique \label{subplot:6_clique_robustness}}};

    \path (top|-current bounding box.north)--
                    coordinate(legendposabove)
                    (bot|-current bounding box.north);
                    
    \begin{axis}[%
            hide axis,
            scale only axis,
            height=0.05\linewidth,
            width=0.8\linewidth,
            yshift=1cm,
            at={(legendposabove.south)},
            anchor=south,
            point meta min=0,
            point meta max=3,
            colormap={custom}{color(0)=(blue) color(1)=(white) color(3)=(red)},
            colorbar horizontal,                  % Active colorbar
            colorbar style={
                separate axis lines,
                samples=256,                        % Number of steps
                xticklabel pos=upper
            },
        ]   
        \addplot [draw=none] coordinates {(0,0)};
    \end{axis}

\end{tikzpicture}
% \vspace{-2em}
\caption{Speedup of ADOPT over System-X when varying the selectivity of newly added unary predicates on three randomly chosen attributes (along the x-axis,  y-axis, and the circles for an (x,y)-point). More intense red (blue) means higher (lower) speedup. All queries are executed on ego-Twitter.}
\label{fig:robustness}
% \vspace{-1em}
\end{figure}

\subsection{Robustness}
\label{sub:robustness}

ADOPT does not rely on query optimization to pick the best attribute order. This can be a significant advantage for queries with user-defined functions or selection predicates, for which there are no available selectivity estimates. Mainstream systems pick a query plan that may be arbitrarily off from a good one. In contrast, ADOPT may quickly realize that such a plan is subpar and switch to a different one. To benchmark this observation, we consider experiments to assess the {\em  robustness} of ADOPT and System-X, which are the two systems we use that rely on attribute orders, when adding to the join queries very simple (unary) yet arbitrary selection conditions that can throw off standard query optimizers.

\begin{hyp}
ADOPT outperforms System-X consistently when varying the selectivity of unary predicates. 
\end{hyp}

% \vspace{-1em}
Figure~\ref{fig:robustness} shows the relative speedup of ADOPT over System-X as we vary the selectivity of unary predicates (selections with constants) on three randomly chosen attributes: we choose the five selectivities 0.2, 0.4, 0.6, 0.8, and 1 for the three attributes along the x-axis,  y-axis, and the circles for an (x,y)-point.
The color of each 3D point in the plot varies from blue to red: The more intense the red is, the higher is the speedup of ADOPT over System-X. System-X mostly outperforms ADOPT for 3-cliques. ADOPT is up to three times faster than System-X for all other cliques and cycles. This is due to the difficulty of optimizers to pick a good query plan in the absence of selectivity estimates, here even for unary predicates.

% Join orders with the same prefix is shared. In skinerDB.
% Hypercube approach divide executions into hypercubes. 

\pgfplotstableread[col sep=space,]{data/share_progress_clique_ratio.csv}\shareprogessclique
\pgfplotstableread[col sep=space,]{data/share_progress_cycle_ratio.csv}\shareprogesscycle

\captionsetup[figure]{font=small}

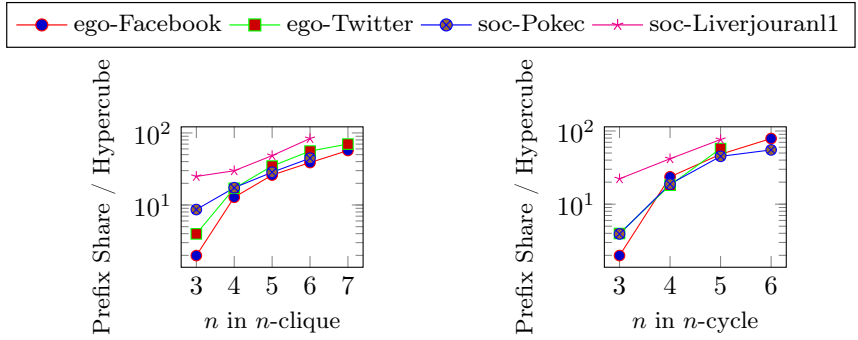
\begin{figure}
\center
\resizebox{0.7\linewidth}{!}{
\ref{graphLegend2}
}

\begin{minipage}{0.22\textwidth}
\begin{tikzpicture}
\begin{axis}
[xlabel={$n$ in $n$-clique}, ylabel={Prefix Share / Hypercube},
width=4cm, ylabel near ticks, xlabel near ticks, y label style={font=\small}, x label style={font=\small}, xtick={3,4,5,6,7}, ymode=log, log basis y={10},
legend entries={ego-Facebook, ego-Twitter, soc-Pokec, soc-Liverjouranl1}, legend columns=4, legend to name=graphLegend2, legend style={font=\footnotesize}, 
]
\addplot+[red] table[x index=0, y index=1] {\shareprogessclique};
\addplot+[green] table[x index=0, y index=2] {\shareprogessclique};
\addplot+[blue] table[x index=0, y index=3] {\shareprogessclique};
\addplot+[magenta] table[x index=0, y index=4] {\shareprogessclique};
\addplot+[black,sharp plot,update limits=false,densely dotted]
    coordinates {(-50, 1) (50, 1)}
    node[midway,below,font=\bfseries\sffamily]{};
\end{axis}
\end{tikzpicture}
% \captionof{figure}{Clique results of different graphs}
% \label{fig:clique_relative}
\end{minipage}
\hspace{5em}
\begin{minipage}{0.22\textwidth}
% \hspace*{-0.4cm}
\begin{tikzpicture}
\begin{axis}
[xlabel={$n$ in $n$-cycle}, width=4cm, ylabel={Prefix Share / Hypercube}, ylabel near ticks, xlabel near ticks, y label style={font=\small}, x label style={font=\small}, ymode=log, log basis y={10},]
\addplot+[red] table[x index=0, y index=1] {\shareprogesscycle};
\addplot+[green] table[x index=0, y index=2] {\shareprogesscycle};
\addplot+[blue] table[x index=0, y index=3] {\shareprogesscycle};
\addplot+[magenta] table[x index=0, y index=4] {\shareprogesscycle};
\addplot+[black,sharp plot,update limits=false,densely dotted]
    coordinates {(-50, 1) (50, 1)}
    node[midway,below,font=\bfseries\sffamily]{};
\end{axis}
\end{tikzpicture}
\end{minipage}
%\vspace{-1em}
\captionof{figure}{Speedup of using our hypercube approach versus using prefix+offset share progress tracker in ADOPT.}
\label{fig:relative2}
\end{figure}

%\subsection{Comparison of ADOPT Variants}
\vspace{-0.5em}
\subsection{Hypercube Data Partitioning}
\label{sub:variants}

We next benchmark the effect of our hypercube partitioning scheme and verify that it indeed leads to faster execution time than  SkinnerDB's alternatives called shared prefix+offset progress tracker~\cite{trummer2019skinnerdb}. 

\begin{hyp}
Hypercube partitioning leads to faster execution than shared prefix progress tracker and offset progress tracker.
\end{hyp}

SkinnnerDB shares progress between all join orders with the same prefix (iterating over all possible prefix lengths). Given a join order, it restores a state by comparing execution progress between the current join order and all other orders with the same prefix and by selecting the most advanced state. Offset progress tracker keeps the last tuples of each table that have been joined with all other tuples already. Using hypercube partitioning, ADOPT executes the episodes on disjoint parts of the input data so it can trivially compute distributive aggregates such as count. This is not the case for SkinnerDB's partitioning: To avoid recomputation of the same result in different episodes, it has to maintain a data structure (concurrent hash map). In the multi-thread environment, the prefix share progress tracker blocks the concurrent execution and causes significant synchronization overhead.

Figure~\ref{fig:relative2} shows the speedup of using the hypercube  partitioning over using the prefix+offset share progress tracker in ADOPT. The hypercube approach consistently has significant smaller overhead than prefix+offset share progress tracker. For larger (above 4) clique and cycle queries, the speedup is 10x to 100x.

\subsection{Time Breakdown by Attribute Order}
\label{sub:optimalattributeorder}

\nop{We run experiments to verify the following hypothesis.}

\begin{hyp}
ADOPT spends most time on executing near-optimal attribute orders.
\end{hyp}

\pgfplotstableread[col
sep=space,]{data/five_clique_order.csv}\fivecliqueorder

\pgfplotstableread[col
sep=space,]{data/five_cycle_order.csv}\fivecycleorder

\pgfplotstableread[col
sep=space,]{data/four_clique_order.csv}\fourcliqueorder

\pgfplotstableread[col
sep=space,]{data/four_cycle_order.csv}\fourcycleorder

\begin{figure}[h]
\begin{center}
%\hspace*{-0.5cm}
\begin{tikzpicture}
    
    \begin{groupplot}[
            xlabel={\#Order selection}, 
            ylabel={Relative time}, 
            ylabel near ticks, xlabel near ticks, ymajorgrids, 
            ymode=log, log basis y={10}, 
            xmode=log, log basis x={10},
            view={0}{90},
            group style = {group name=robustness plots, group size = 2 by 2, vertical sep=2cm, horizontal sep=4cm},
            width = 4 cm,
            label style={font=\small},
        ]
        
        \nextgroupplot
    
        \addplot[scatter, only marks] table [x index=1, y index=0, scatter] {\fourcliqueorder};
        
        \coordinate (top) at (rel axis cs:0,1);
        
        \nextgroupplot
    
        \addplot[scatter, only marks] table [x index=1, y index=0, scatter] {\fourcycleorder};
        
        \nextgroupplot[xmax=100000000]
    
        \addplot[scatter, only marks] table [x index=1, y index=0, scatter] {\fivecliqueorder};
        
        \nextgroupplot
    
        \addplot[scatter, only marks] table [x index=1, y index=0, scatter] {\fivecycleorder};
        
         \coordinate (bot) at (rel axis cs:1,0);

    \end{groupplot}
    
    \node[text width=6cm,align=center,anchor=north] at ([yshift=-5mm]robustness plots c1r1.south) {\caption*{4 clique}};
    
    \node[text width=6cm,align=center,anchor=north] at ([yshift=-5mm]robustness plots c2r1.south) {\caption*{4 cycle}};
    
    \node[text width=6cm,align=center,anchor=north] at ([yshift=-5mm]robustness plots c1r2.south) {\caption*{5 clique}};
    
    \node[text width=6cm,align=center,anchor=north] at ([yshift=-5mm]robustness plots c2r2.south) {\caption*{5 cycle}};

\end{tikzpicture}
\end{center}
\caption{Selections of orders with different quality on ego-Twitter.}
\label{fig:orderdetail}
%\vspace{-2em}
\end{figure}

We verified this hypothesis for $n$-clique and $n$-cycle queries with $n\in\{4,5\}$, since for these queries it was feasible to generate and execute all possible attribute orders. This was necessary to understand which orders are better than others and assess whether ADOPT uses predominantly good or poor orders.
We plot the orders that we select and their quality relative to the optimal orders (i.e., with lowest execution time) in Figure~\ref{fig:orderdetail}. The x-axis is the number of time slices  that use an order: the larger the x-value, the more we use an order. The y-axis is execution time of an order relative to the optimal one: The smaller the y-value, the closer to the optimal the order is. For $4$-clique and $4$-cycle, ADOPT spends more than $10^6$ (over 95\% frequency) times on executing an order with near-optimal performance. For $5$-cycle and $5$-clique,  ADOPT picks a near-optimal order more than $10^8$ times (over 98\% frequency). ADOPT thus quickly converges to a near-optimal order and then uses it for most of the processing, which confirms our hypothesis.

% We show orders that we select and its quality in Figure~\ref{fig:orderdetail}.

\begin{table}[t]
\centering 
\caption{Execution times (sec) for clique and cycle queries on ego-Twitter of: ADOPT, LFTJ with optimal attribute order (OPT), average runtime of LFTJ over all attribute orders (AVG).  Relative speedup of OPT over ADOPT (last column).}
\begin{tabular}{l|rrrr}
\toprule[1pt]  & ADOPT & OPT & AVG & ADOPT/OPT \\
\midrule[1pt]
3 clique & 4.1	 & 1.6	  & 3.7	   & 2.52 \\
4 clique & 10.5  & 6.8	  & 23.9   & 1.54 \\
5 clique & 77.9	 & 52.5	  & 275.6  & 1.48 \\
3 cycle	 & 4.1	 & 1.6	  & 3.5	   & 2.52 \\
4 cycle	 & 20.1  & 17.4	  & 58.9   & 1.16 \\
5 cycle	 & 377.9 & 328.8  &	3618.1 & 1.14 \\
\bottomrule[1pt]
\end{tabular}
\label{tab:ratio_with_opt}
\end{table}

Table~\ref{tab:ratio_with_opt} compares ADOPT and LFTJ with an optimal attribute order: The runtime gap decreases from 2.52x for 3-clique/cycle to 1.48x (1.14x) for 5-clique (5-cycle). This is remarkable, given that ADOPT tries out several attribute orders and switches between them, whereas LFTJ only uses one attribute order, which is optimal. Table~\ref{tab:ratio_with_opt} also shows that ADOPT takes significantly less time than the average runtime of LFTJ over all attribute orders\nop{ the larger the query gets}.

% R_1(a, b), R_2(b, c), R_3(c, a)
% R_1(a, b, c), R_2(b, c, d), R_3(c, d, a), R_4(d, a, b)
% R_1(a, b, c, d), R_2(b, c, d, e), R_3(c, d, e, a), R_4(d, e, a, b), R_5(e, a, b, c)
% R_1(a, b, c, d, e), R_2(b, c, d, e, f), R_3(c, d, e, f, a), R_4(d, e, f, a, b), R_5(e, f, a, b, c), R_6(f, a, b, c, d)

% E(a, b, a), E(b, a, b), E(a, b, a), E(b, a, b)
% R_1(a, b, a, b), R_2(b, a, b, a), R_3(a, b, a, a), R_4(b, a, a, b), R_5(a, a, b, a)

\subsection{Parallelization}
\label{sub:parallelization}

\begin{hyp}
    ADOPT achieves almost linear speedup for large cyclic queries. 
\end{hyp}

\pgfplotstableread[col
sep=comma,]{data/speedup_twitter_clique.csv}\parallelspeedupclique
\pgfplotstableread[col
sep=comma,]{data/speedup_twitter_cycle.csv}\parallelspeedupcycle

\pgfplotsset{
legend image code/.code={
\draw[mark repeat=2,mark phase=2]
plot coordinates {
(0cm,0cm)
(0.05cm,0cm)        %% default is (0.3cm,0cm)
(0.05cm,0cm)         %% default is (0.6cm,0cm)
};%
}
}

\begin{figure}
\center
% \ref{threadLegend}
\begin{minipage}{0.22\textwidth}
\begin{tikzpicture}
\begin{axis}
[xlabel={Number of threads}, ylabel={Speedup},
width=5cm, ylabel near ticks, xlabel near ticks, y label style={font=\small}, x label style={font=\small}, legend entries={3-clique, 4-clique, 5-clique, 6-clique}, legend columns=1, legend style={font=\tiny}, legend cell align=left, legend pos=north west,
]
\addplot+[red] table[x index=0, y index=1] {\parallelspeedupclique};
\addplot+[green] table[x index=0, y index=2] {\parallelspeedupclique};
\addplot+[blue] table[x index=0, y index=3] {\parallelspeedupclique};
\addplot+[magenta] table[x index=0, y index=4] {\parallelspeedupclique};
\end{axis}
\end{tikzpicture}
\end{minipage}
\hspace{7em}
\begin{minipage}{0.22\textwidth}
\hspace*{-0.4cm}
\begin{tikzpicture}
\begin{axis}
[xlabel={Number of threads}, width=5cm, ylabel near ticks, xlabel near ticks, y label style={font=\small}, x label style={font=\small}, legend entries={3-cycle, 4-cycle, 5-cycle}, legend columns=1, legend style={font=\tiny}, legend cell align=left, legend pos=north west,
]
\addplot+[red] table[x index=0, y index=1] {\parallelspeedupcycle};
\addplot+[green] table[x index=0, y index=2] {\parallelspeedupcycle};
\addplot+[blue] table[x index=0, y index=3] {\parallelspeedupcycle};
\end{axis}
\end{tikzpicture}
\end{minipage}
\caption{Speedup of multi-threaded ADOPT over single-threaded ADOPT for clique and cycle queries on the ego-Twitter dataset.}
\label{fig:speedup}
% \vspace{-1em}
\end{figure}
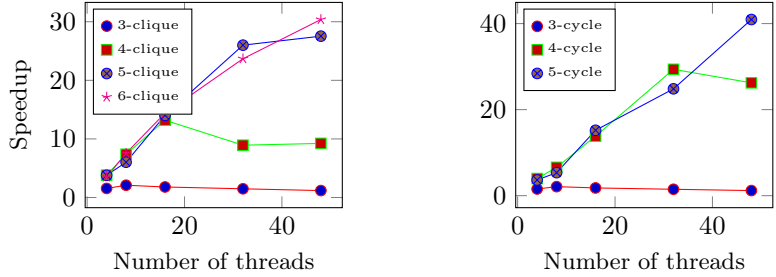

Figure~\ref{fig:speedup} plots the speedup of ADOPT as a function of the number of threads. ADOPT achieves significant speedups for large clique and cycle queries. In particular, it achieves nearly 30x speedup on 5- and 6-clique, and 40x speedup on 5-cycle (with 48 threads). The main reason is that the hypercube approach partitions disjointly the workload across threads, minimizing synchronization overheads. %and makes the computation synchronization-free and therefore fully parallelizable.

%avoids concurrent data structures to remove duplicate tuples from different orders and thus can be fully parallelized.

% For complex queries (e.g., 5 cycle and 6 clique), ADOPT 

\section{Related Work}
\label{sec:related}

%ADOPT relates to prior work on worst-case optimal join algorithms, as well as prior work in query optimization.

The choice of an attribute order, for worst-case optimal join algorithms, resembles the problem of join order selection for traditional join algorithms~\cite{Selinger1979}. Both tuning decisions have significant impact on processing performance. At the same time, it is hard to find good attribute orders before query processing starts, mainly due to challenges in estimating execution cost for specific orders (e.g., due to challenges in estimating sizes of intermediate results). The latter problem has been well documented for traditional query optimizers~\cite{Gubichev2015, Lohman2014}. Our experiments demonstrate that it appears in the context of worst-case optimal join algorithms as well.

Adaptive processing~\cite{Avnur2000, Deshpande2004, Quanzhong2007b, Tzoumas2008, Wei2022, Zhu2017} has been proposed as a \nop{possible} remedy to this problem, allowing the engine to switch to a different join order during query execution based on run time feedback. While early work has focused on stream data processing~\cite{Avnur2000, Deshpande2004, Quanzhong2007b, Tzoumas2008} (where query execution times are assumed to be longer), adaptive processing has recently also gained traction for classical query processing~\cite{Menon2020, trummer2019skinnerdb}. 
SkinnerDB~\cite{trummer2019skinnerdb} is the closest in spirit to ADOPT: \revision{both use reinforcement learning and adaptive processing. However, ADOPT uses an anytime version of a worst-case optimal join algorithm, whereas SkinnerDB's join algorithm is not optimal. The learning problems (i.e., actions and states of the corresponding MDPs) differ between the systems as ADOPT optimizes attribute orders whereas SkinnerDB orders tables. Most importantly: ADOPT introduces a novel data structure, characterizing precisely the cubes in the space of attribute value combinations that have not been processed yet, along with operators for updating it after each episode. This data structure avoids redundant work across episodes and attribute orders as well as across threads. This property is crucial to be able to maintain optimality guarantees for equi-joins when switching between attribute orders. Instead, SkinnerDB uses a tree-based data structure that reduces but does not completely avoid redundant work across join orders that are dissimilar. As the amount of redundant work is hard to bound, it is difficult to maintain worst-case optimality guarantees with such mechanisms.}

Our work uses reinforcement learning to select attribute orders. It  relates to recent works that employ learning for database tuning in general~\cite{udodemo, bait, Hilprecht2019a, Li2018, Trummer2022, VanAken2021, Wang2021} and, specifically, for query optimization~\cite{Krishnan2018, Marcus, Marcus2018a, Yu2020b}. Our work differs as it focuses on learning and specialized data structures for worst-case optimal join algorithms.

Prior work on query optimization for worst-case optimal joins investigates "model-free" information-theoretic cardinality estimation. A seminal work, which enabled reasoning about worst-case optimal join computation, established tight bounds on the worst-case size of join results~\cite{DBLP:journals/siamcomp/AtseriasGM13}, the so-called AGM bound that is defined as the cost of the optimal solution of a linear program derived from the joins and the sizes of the input tables. This is further refined in the presence of functional dependencies~\cite{DBLP:journals/jacm/GottlobLVV12} and for succinct factorized representations of query results~\cite{DBLP:journals/tods/OlteanuZ15}. The latest development 
extends this line of work with data degree constraints and histograms~\cite{DBLP:conf/icdt/000122}. Classical approaches to query optimization based on heuristics~\cite{Freitag2020} and data statistics~\cite{Aberger2016, Aref2015} have also been considered.
To the best of our knowledge, ADOPT is the first adaptive approach for optimization in the context of worst-case optimal join algorithms. Our approach is free from cost-based heuristics.
% \vspace{-0.5em}
\section{Conclusion}

Worst-case optimal join algorithms and adaptive processing strategies have been two of the most exciting advances in join processing over the past decades. Worst-case optimal joins enable efficient processing of cyclic queries. Adaptive processing allows handling complex queries where a-priori optimization is hard. For the first time, ADOPT brings together these two techniques, resulting in attractive performance for both acyclic and cyclic queries and in particular excellent performance for large cyclic queries.

ADOPT is an adaptive framework readily applicable to further query processing techniques, e.g., factorized databases~\cite{DBLP:journals/pvldb/BakibayevKOZ13} and functional aggregate queries~\cite{DBLP:journals/sigmod/Khamis0R17}. These works combine worst-case optimal joins with effective techniques to push aggregates past joins to achieve the best known computational complexity for query evaluation. In future work, we plan to merge this line of work with ADOPT-style adaptivity.

\bibliographystyle{abbrv}
\bibliography{library}

\appendix
\appendix
\section{Statistics about Graph Datasets}
\label{app:statistics_graph_datasets}
Table~\ref{tab:graphtable} describes the graph datasets used in our experiments. 
%%%%%%%%%%%%%%%%%%%%
\begin{table}[h]
\centering
\caption{Graphs used in the experiments.}
\begin{tabular}{lrr}
\toprule[1pt]
\textbf{Graph} & \textbf{\#Vertices} & \textbf{\#Edges} \\
\midrule[1pt]
ego-Facebook & 4,039 & 88,234 \\ 
ego-Twitter & 81,306 & 2,420,766 \\  
soc-Pokec & 1,632,803 & 30,622,564  \\
soc-LiveJournal1 & 4,847,571 & 689,937,732 \\
\bottomrule[1pt]
\end{tabular}
\label{tab:graphtable}
%\vspace*{-1em}
\end{table}
%%%%%%%%%%%%%%%%%%%%

\section{Relative performance comparison of ADOPT over System-X}
\label{app:comparison_system_X}
\pgfplotstableread[col sep=comma,]{data/clique_relative.csv}\cliquerelative
\pgfplotstableread[col sep=comma,]{data/cycle_relative.csv}\cyclerelative

% \captionsetup[figure]{font=small}

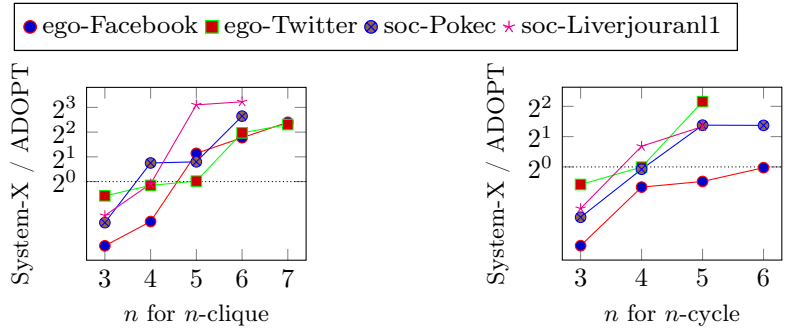
\begin{figure}
\center
\resizebox{0.6\linewidth}{!}{
\ref{graphLegend}
}

\begin{minipage}{0.22\textwidth}
\begin{tikzpicture}
\begin{axis}
[xlabel={$n$ for $n$-clique}, ylabel={System-X / ADOPT},
width=4.5cm, ylabel near ticks, xlabel near ticks, y label style={font=\small}, x label style={font=\small}, xtick={3,4,5,6,7}, ytick={1,2,4,8}, ymode=log, log basis y={2},
legend entries={ego-Facebook, ego-Twitter, soc-Pokec, soc-Liverjouranl1}, legend columns=4, legend to name=graphLegend, legend style={font=\footnotesize}, 
]
\addplot+[red] table[x index=0, y index=1] {\cliquerelative};
\addplot+[green] table[x index=0, y index=2] {\cliquerelative};
\addplot+[blue] table[x index=0, y index=3] {\cliquerelative};
\addplot+[magenta] table[x index=0, y index=4] {\cliquerelative};
\addplot+[black,sharp plot,update limits=false,densely dotted]
    coordinates {(-50, 1) (50, 1)}
    node[midway,below,font=\bfseries\sffamily]{};
\end{axis}
\end{tikzpicture}
% \captionof{figure}{Clique results of different graphs}
% \label{fig:clique_relative}
\end{minipage}
\hspace{7em}
\begin{minipage}{0.22\textwidth}
\begin{tikzpicture}
\begin{axis}
[xlabel={$n$ for $n$-cycle}, ylabel={System-X / ADOPT}, width=4.5cm, ylabel near ticks, xlabel near ticks, y label style={font=\small}, x label style={font=\small}, ytick={1,2,4,8}, ymode=log, log basis y={2},]
\addplot+[red] table[x index=0, y index=1] {\cyclerelative};
\addplot+[green] table[x index=0, y index=2] {\cyclerelative};
\addplot+[blue] table[x index=0, y index=3] {\cyclerelative};
\addplot+[magenta] table[x index=0, y index=4] {\cyclerelative};
\addplot+[black,sharp plot,update limits=false,densely dotted]
    coordinates {(-50, 1) (50, 1)}
    node[midway,below,font=\bfseries\sffamily]{};
\end{axis}
\end{tikzpicture}
\end{minipage}
\captionof{figure}{Relative speedup of ADOPT over System-X.}
\label{fig:relative}
% \vspace*{-1.5em}
\end{figure}

Figure~\ref{fig:relative} examines the relative performance of ADOPT over System-X\nop{, the two systems that implement LFTJ}. The speedup of ADOPT over System-X increases with the query parameter $n$; System-X times out for large $n$. This speedup reaches: 4x for both 5- and 6-clique on ego-Twitter and soc-Pokec; 8x for 5- and 6-clique on soc-Livejournal1; 2x for 5-cycle on ego-Twitter and 4x on both soc-Pokec and soc-Liverjournal1. A reason for this speedup increase is the difficulty of System-X's optimizer to pick a good attribute order for increasingly larger queries. The average performance of the attribute orders of ADOPT is better than the one attribute order of System-X.

%\color{red}
\section{Disk and Memory Consumption}
\label{sec:diskmemory}

%%%%%%%%%%%%%%%%%%%% disk space %%%%%%%%%%%%%%%%%%%%
\begin{table*}[t]
\caption{Disk space of different systems on each benchmark.}
% \vspace{-1em}
\centering
\begin{tabular}{l|rrrrrrr}
\toprule[1pt] Systems & JOB & ego-Facebook & ego-Twitter & soc-Pokec & soc-Livejournal1 & TPC-H & JCC-H \\
\midrule[1pt]
ADOPT       & 3.2G & 836K & 48M  & 405M & 1.1G & 8.5G & 8.5G \\
System-X    & 2.8G & 1.6M & 46M  & 387M & 759M & -    & -    \\
EmptyHeaded & -    & 6.0M & 44M  & 496M & 1.3G & -    & -    \\
PostgreSQL  & 5.4G & 11M  & 101M & 1.1G & 2.3G & 20GB & 20GB \\
MonetDB     & 2.9G & 3.5M & 58M  & 247M & 554M & 8.5G & 8.5G \\
SkinnerDB  & 3.1G & 834K & 47M  & 403M & 882M & 8.4G & 8.4G \\
\bottomrule[1pt]
\end{tabular}
\label{tab:disk}
\end{table*}

%%%%%%%%%%%%%%%%%%%% memory space %%%%%%%%%%%%%%%%%%%%
\begin{table*}[t]
\caption{Memory space of different systems on each benchmark.}
% \vspace{-1em}
\centering
\begin{tabular}{l|rrrrrrr}
\toprule[1pt] Systems & JOB & ego-Facebook & ego-Twitter & soc-Pokec & soc-Livejournal1 & TPC-H & JCC-H \\
\midrule[1pt]
ADOPT       & 22G & 10G	 & 17G	& 22G  & 45G  & 52G  & 57G  \\
% ADOPT Sort Arrays  & 2G  & 689K & 19M  & 233M & 5G   & 2G   & 2G   \\
% ADOPT Unprocessed Cubes  & 37K & 445K & 878K & 3M   & 11M  & 104K & 105K \\
% ADOPT Uct Nodes   & 21M & 4M   & 4M   & 4M   & 4M   & 109K & 110K \\ 
System-X    & 38G & 16G	 & 16G	& 16G  & 26G  & -    & -    \\
EmptyHeaded & -   & 68G	 & 74G	& 85G  & 89G  & -    & -    \\
PostgreSQL  & 28G & 15G	 & 17G	& 25G  & 110G & 56G  & 110G \\
MonetDB     & 42G & 122G & 243G	& 345G & 345G & 18G  & 280G \\
SkinnerDB   & 38G & 26G  & 69G  & 89G  & 125G  & 86G  & 86G \\ 
\bottomrule[1pt]
\end{tabular}
\label{tab:memory}
\end{table*}

\begin{table}[b]
\caption{Breakdown of ADOPT's main memory consumption for each benchmark: total memory consumption, memory for storing sort orders, unprocessed cubes, and the UCT search tree.}
% \vspace{-1em}
\centering
\begin{tabular}{l|rrrr}
\toprule[1pt] Benchmark & Total & Sort & Cubes & UCT \\
\midrule[1pt]
JOB & 22G & 2G & 37K & 21M \\
ego-Facebook & 10G & 689K & 445K & 4M \\
ego-Twitter & 17G & 19M & 878K & 4M \\
soc-Pokec & 22G & 233M & 3M & 4M \\
soc-Livejournal1 & 45G & 5G & 11M & 4M \\
TPC-H & 52G & 2G & 104K & 109K \\
JCC-H & 57G & 2G & 105K & 110K \\
\bottomrule[1pt]
\end{tabular}
\label{tab:memorybreakdown}
\end{table}

We report the disk space and memory usage of different systems in Table~\ref{tab:disk} and ~\ref{tab:memory} respectively. We use the command \texttt{du -sh} to measure disk space of the respective data folder. For measuring maximal memory consumption for each benchmark, we use the \texttt{ps -p pid -o rss=} command. ADOPT is implemented in Java. Hence, we increased the default settings for the Xmx and Xms parameters of the Java virtual machine as follows: \texttt{-Xmx15G -Xms15G} on ego-Facebook, \texttt{-Xmx20G -Xms20G} on ego-Twitter, \texttt{-Xmx25G -Xms25G} on soc-Pokec, \texttt{-Xmx50G -Xms50G} on soc-Livejournal1, and \texttt{-Xmx80G -Xms80G} on JOB, TPC-H, and JCC-H. For all other systems, to maximize their performance in terms of run time (results reported in Table~\ref{tab:overall}), we increased buffer space to 350~GB, the amount of main memory available on our test machine (e.g., for Postgres, we increased the setting for the shared\_buffer\_pool parameter). Note that all systems typically exploit only a small part of the total buffer space available to them. As discussed in Section~\ref{sec:experiments} in more detail, some systems were only evaluated on a subset of benchmarks (missing values are marked by ``-'' in the tables).

% For ADOPT, we use \texttt{-Xmx15G -Xms15G} on ego-Facebook, \texttt{-Xmx20G -Xms20G} on ego-Twitter, \texttt{-Xmx25G -Xms25G} on soc-Pokec, \texttt{-Xmx50G -Xms50G} on soc-Livejournal1, \texttt{-Xmx80G -Xms80G} on JOB, TPC-H and JCC-H. 

Table~\ref{tab:memory} shows that ADOPT consumes an amount of main memory that is approximately comparable to System~X, i.e., ADOPT's main memory consumption is within a factor of 0.6 to 1.7 of the corresponding value for System~X for all evaluated benchmarks. This seems reasonable as the execution engine of System~X is the most similar to the one of ADOPT (due to the use of LFTJ variants). EmptyHeaded, another worst-case optimal system, consumes more main memory than ADOPT. MonetDB consumes significantly more main memory than ADOPT on the graph benchmarks. Here, using non worst-case optimal joins, MonetDB produces large intermediate results for cyclic queries that are stored in main memory. The use of worst-case optimal join algorithms avoids these overheads. Also, MonetDB consumes more main memory on benchmarks that use skewed data (in particular JCC-H). Here, large intermediate results can be avoided by using the right join order. However, due to data skew, reliably identifying near-optimal query plans without adaptive processing is hard. On the other hand, MonetDB consumes only moderate amounts of main memory on TPC-H. Here, worst-case optimal join algorithms are unnecessary to avoid large intermediate results and query planning is easier (due to uniform data distributions). SkinnerDB incurs high memory overhead on some of the graph benchmarks due to cyclic queries. Postgres incurs high memory overheads for some of the graph benchmarks and for the JCC-H benchmark, due to binary joins and non-adaptive optimization. 

ADOPT stores several data structures in main memory that are specific to its adaptive approach: different sort orders for each table to support the LFTJ (each sort order is stored as one integer array with row indexes, of the same length as the table), the set of unprocessed cubes, maintained by the task manager (and represented as variable $U$ in Algorithm~\ref{alg:cubes}), and data structures used by the reinforcement learning algorithm, in particular the UCT search tree with associated reward statistics. Table~\ref{tab:memorybreakdown} shows the amount of memory consumed by each of these data structures for each benchmark (in addition to the total main memory consumption). For most benchmarks, sort orders consume most main memory, among all auxiliary data structures, reaching up to 11\% of total memory consumption for one benchmark. Main memory consumption for storing cubes and UCT statistics is lower by several orders of magnitude, making their contribution to total main memory consumption negligible.

\section{Index Creation Time}
\label{sec:indexcreation}

%%%%%%%%%%%%%%%%%%%% index time %%%%%%%%%%%%%%%%%%%%
\begin{table*}[t]
\caption{Index creation time (in seconds) of different systems for each benchmark it was evaluated on (- if the corresponding system was not evaluated on the benchmark).}
% \vspace{-1em}
\centering
\resizebox{\linewidth}{!}{
\begin{tabular}{l|rrrrrrr}
\toprule[1pt] Systems & JOB & ego-Facebook & ego-Twitter & soc-Pokec & soc-Livejournal1 & TPC-H & JCC-H \\
\midrule[1pt]
ADOPT             & 63 & 0.2  & 1.6 & 28  & 64  & 155 & 153  \\
% ADOPT Sort Index  & 42 & 0.2  & 1.6 & 28  & 64  & 97 & 94  \\
% ADOPT Hash Index  & 21 & 0    & 0   & 0   & 0   & 58 & 59  \\
EmptyHeaded & > 432000  & 17  & 24  & 38  & 115 & -   & -   \\
PostgreSQL  & 78 & 0.36 & 14  & 144 & 328 & 169 & 172 \\
SkinnerDB   & 23 & 0.4  & 1.1 & 19  & 25  & 82  & 84 \\
\bottomrule[1pt]
\end{tabular}
}
\label{tab:index}
\end{table*}

\begin{table}[b]
\caption{Breakdown of index generation overheads for ADOPT: total indexing time, time for generating sort orders, and time for generating hash indexes.}
% \vspace{-1em}
\centering
\begin{tabular}{l|rrr}
\toprule[1pt] Benchmark & Total & Sorting & Hashing \\
\midrule[1pt]
JOB & 63 & 42 & 21 \\
ego-Facebook & 0.2 & 0.2 & 0 \\
ego-Twitter & 1.6 & 1.6 & 0 \\
soc-Pokec & 28 & 28 & 0 \\
soc-Livejournal1 & 64 & 64 & 0 \\
TPC-H & 155 & 97 & 58 \\
JCC-H & 153 & 94 & 59 \\
\bottomrule[1pt]
\end{tabular}
\label{tab:indexbreakdown}
\end{table}

%Table~\ref{tab:index} reports index creation overheads

We report on index creation overheads, referring to the indexes created before the experiments discussed in Section~\ref{sub:baselines}. For Postgres, we created indexes on all primary and foreign key columns. In addition, for TPC-H and JCC-H, we additionally index the o\_orderdate column of the Orders table and build an index with composite search key on the l\_shipdate, l\_discount, and l\_quantity columns of the Lineitem table (the additional indexes improved performance). For SkinnerDB, we run the ``index all'' command, indexing all suitable columns and thereby optimizing its query evaluation times. For ADOPT, we index the same columns as SkinnerDB via hash indexes (supporting evaluation of unary equality predicates), except for join columns (which SkinnerDB indexes via hash indexes as well). Additionally, ADOPT creates data structures representing different sort orders of base tables (see Appendix~\ref{sec:lftjinadopt} for details). System~X and MonetDB create indexes automatically, based on observed queries (note that MonetDB supports the ``create index'' command but it is only treated as a suggestion, according to the online manual\footnote{\url{https://www.monetdb.org/documentation-Sep2022/user-guide/sql-summary/\#create-index}}). To give those two systems the opportunity to create suitable indexes, we ran each benchmark once before starting the actual measurements. As those systems interleave query execution and index creations, making it hard to measure index creation time separately, we do not report indexing overheads for those systems. EmptyHeaded creates all relevant indexes in a pre-processing step. 

%For PostgreSQL, we set the \texttt{\textbackslash timing} flag to measure the index creation time. The index creation time for EmptyHeaded is automatically logged and displayed in the console. As for ADOPT and SkinnerDB, the index creation is measured by calling \texttt{System.currentTimeMillis()} before and after building indices and computing the time difference. 

Table~\ref{tab:index} reports corresponding results. ADOPT has typically higher index creation overheads than SkinnerDB due to the added overhead for creating sort orders. On the other hand, ADOPT's index creation overheads are below the ones of Postgres. EmptyHeaded incurs higher index creation overheads, compared to ADOPT, as it creates indexes for all permutations of table columns. On JOB, this approach incurs very high index generation overheads, making the approach impractical. Table~\ref{tab:indexbreakdown} shows a breakdown of ADOPT's index generation overheads into two components: time for generating indexes representing row orders (to support LFTJ) and time for generating hash indexes on single columns (to support evaluation of unary equality predicates). Clearly, time for generating row orders dominates, even though hashing time is non-negligible for the TPC-H and JCC-H benchmarks. 

%On JOB, EmptyHeaded constructs tries for each column permutation, including string columns, of every table. As a result, it is unable to complete within a reasonable time. We do not report measurements for MonetDB and SystemX, as both interleave query evaluation with index creations, determining indexes to create automatically\footnote{\url{https://www.monetdb.org/documentation-Sep2022/user-guide/sql-summary/\#create-index}}. This makes it hard to separate execution time from index creation time.
% MonetDB and System-X automatically build indices and materialize caches, making precise measurement challenging. 

%Queries usually exhibit slower performance in their initial runs and significantly improve in subsequent runs. The performance that we report in Table~\ref{tab:overall} and Figure~\ref{fig:clique_cycle_results} is their stable performance after running queries multiple times. Based on the results in Table~\ref{tab:index}, ADOPT demonstrates reasonable time consumption for index creation compared to all other baselines.

%Moreover, we breakdown the index creation time in ADOPT to the creation time of the hash index and the creation time of the sort index. In graph benchmarks, we only build sort index, while in JOB, TPC-H and JCC-H benchmarks, we create both the sort index and the hash index (to support complex unary predicates or aggregation function). From Table~\ref{tab:index},  it is evident that the time required for sorting the indices surpasses the time taken for creating the hash index.

% From table~\ref{tab:index}, ADOPT spends reasonable time on indices creation compared with all baselines.

\section{Sorting and Synchronization Overheads}
\label{sec:ablation}

\nop{
 \begin{table}[t]
\centering 
\caption{Sorting and synchronization times for 5 clique/cycle.}
\begin{tabular}{llccc}
\toprule[1pt] Graph & Query & Sorting & Synchronization \\
\midrule[1pt]
\multirow{2}{*}{ego-Facebook} & 5 clique  & 0.022 & 0.046 \\
                              & 5 cycle	& 0.003 & 0.077 \\
\hline                       
\multirow{2}{*}{ego-Twitter}  & 5 clique & 0.015 & 0.017 \\
                              & 5 cycle  & 0.001	& 0.013 \\
\hline                           
\multirow{2}{*}{soc-Pokec}   & 5 clique  & 0.021	& 0.067 \\
                             & 5 cycle  & 0.018	& 0.063 \\
\hline                       
\multirow{2}{*}{soc-LiveJournal1} & 5 clique  & 0.021	& 0.049\\
                                  & 5 cycle  & 0.001	& 0.059 \\
\bottomrule[1pt]
\end{tabular}
\label{tab:time_breakdown}
\end{table}
}

\begin{hyp}
    The times required by ADOPT for sorting and  thread synchonization are small relative to the total execution time.
\end{hyp}

% Graph	Query	Sort Time	Wait Time	Execution Time	Sort Ratio	Wait Ratio
% ego-Facebook	5 clique	354	749	16204	0.022	0.046
% 	5 cycle	70	2052	26743	0.003	0.077
% ego-Twitter	5_clique	1318	1473	87255	0.015	0.017
% 	5_cycle	299	4088	311227	0.001	0.013
% soc-Pokec	5_clique	18148	5681	84933	0.214	0.067
% 	5_cycle	6287	21869	345698	0.018	0.063
% soc-LiveJournal1	5_clique	37676	86187	1770447	0.021	0.049
% 	5_cycle	11406	518615	8719080	0.001	0.059

\nop{Table ~\ref{tab:time_breakdown} confirms this hypothesis for two queries.}

To implement efficient seek operations in the context of its LFTJ variant, ADOPT requires data structures representing different table sort orders (see Appendix~\ref{sec:lftjinadopt} for details). For base tables, ADOPT creates those sort orders at pre-processing time, corresponding overheads are reported in Appendix~\ref{sec:indexcreation}. However, ADOPT creates temporary tables during query evaluation, representing base tables after filtering via unary predicates. For those tables, ADOPT creates all required sort orders at run time. We measured the relative overhead of run time sorting, compared to total query evaluation time. Over all queries and benchmarks, sorting overheads reach at most 2.5\% of total query evaluation time. This means that time required for run time sorting is fairly modest (which is explained, in part, by the fact that tables resulting from filter operations tend to be quite small, compared to the source tables).

%Whenever ADOPT starts a new episode with an attribute order different from the order used in the previous episodes, it might need to re-sort the tables so as to support the intersection of sorted lists as required by the leapfrog triejoin. We quantify the sorting time by invoking \texttt{System.currentTimeMillis()} before and after the sorting process, and then calculating the duration between the two timestamps. In our experiments, the time for sorting only accounts for at most 2.5\% of the total execution time of ADOPT for all queries. Sorting has thus negligible overhead, even though ADOPT frequently switches orders.

For each episode and thread, we can sum up the time spent by that thread in processing the LFTJ join on different cubes (considering all cubes processed by the thread during the episode). When measuring the total duration of an episode, the episode time typically exceeds the accumulated join processing time. The difference is due to various bookkeeping and synchronization overheads, e.g., waiting for the lock on the data structure containing unprocessed cubes (locking is necessary to avoid redundant work across threads). Considering all queries and benchmark, the maximal percentage of such overheads, relative to total query execution time, was 7\%. This means that the largest part of run time is spent doing useful work.

\section{Scalability in the Number of Join Attributes per Relation}
\label{sec:attributescalability}

Many of the benchmarks presented so far have an elevated number of tables and join attributes (e.g., up to 16 tables for JOB). However, the number of join attributes per table is typically small (two join attributes per table for graph benchmarks, reaching up to four attributes for some JOB queries). The number of attributes per table influences the number of sort orders ADOPT has to maintain, possibly influencing its relative performance. Next, we study the impact of the number of attributes per table on the relative performance of ADOPT. 

%This number influences the number of sort orders ADOPT may have to create during processing, thereby potentially. 

%To test the performance of ADOPT when scaling to larger number of attributes, we vary the number of degree in Loomis-Whitney query.

% degree 5 edge(a, b, c, d), edge(b, c, d, e), edge(a, c, d, e), edge(a, b, d, e), edge(a, b, c, e).

We use Loomis-Whitney queries with varying degree for this purpose. A Loomis-Whitney query with degree $n$ (i.e., the number of join attributes in each table is $n-1$) is defined as,

% \vspace{-1em}
\begin{equation*}
\begin{split}
    & edge(a_1, a_2, \cdots, a_{n-2}, a_{n-1}), edge(a_2, \cdots, a_{n-2}, a_{n-1}, a_{n}), \\
    & edge(a_1, a_3, \cdots, a_{n-1}, a_n), edge(a_1, a_2, a_4, \cdots, a_n), \cdots, \\
    & edge(a_1, a_2, \cdots, a_{n-2}, a_n) \\
\end{split}
\end{equation*}

\pgfplotstableread[col sep=comma,]{data/lw.csv}\generalizedquery

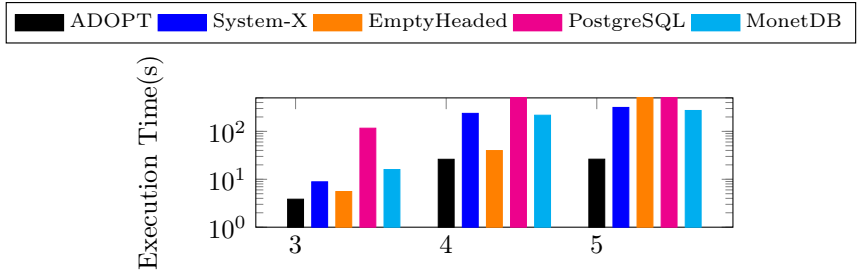
\begin{figure}
\begin{center}
\resizebox{0.7\linewidth}{!}{
\ref{baselineLegend2}
}

\begin{tikzpicture}
    \pgfmathsetmacro{\BarOffset}{0.8}
    \begin{axis}[
        width=0.48\textwidth,
        height=0.2\textwidth,
        ymin=1,
        xtick=data,
        ymode=log,
        ymax=500,
        % log basis y={10},
        % xtick distance=0.5,
        xticklabels from table={\generalizedquery}{Query},
        ylabel={Execution Time(s)},
        ylabel style={yshift=10pt},
        xlabel style={
            anchor=west,
        },
        legend entries={ADOPT, System-X, EmptyHeaded, PostgreSQL, MonetDB}, legend columns=5, legend to name=baselineLegend2, legend style={at={(0.5,-0.1)},anchor=north,cells={align=left,anchor=west},font=\scriptsize},
        area legend,
        bar width=6pt,
        clip mode=individual,
    ]

    % plot the "black" ybars
    \addplot [
        ybar,
        draw=black,
        fill=black,
    ] table [
        x expr=\coordindex,
        y=ADOPT
    ] {\generalizedquery};
    
    % plot the "blue" ybars
    \addplot [
        ybar,
        draw=blue,
        fill=blue,
    ] table [
        x expr=\coordindex+0.2*\BarOffset,
        y=Logicblox
    ] {\generalizedquery};

    % plot the "cyan" ybars
    \addplot [
        ybar,
        draw=orange,
        fill=orange,
    ] table [
        x expr=\coordindex+0.4*\BarOffset,
        y=EmptyHeaded
    ] {\generalizedquery};

    % plot the "magenta" ybars
    \addplot [
        ybar,
        draw=magenta,
        fill=magenta,
    ] table [
        x expr=\coordindex+0.6*\BarOffset,
        y=PostgreSQL
    ] {\generalizedquery};

    % plot the "cyan" ybars
    \addplot [
        ybar,
        draw=cyan,
        fill=cyan,
    ] table [
        x expr=\coordindex+0.8*\BarOffset,
        y=MonetDB
    ] {\generalizedquery};

    % \begin{scope}[=
    %     every label/.append style={
    %         label distance=2ex,
    %     },
    % ]
    %     \node [label=below:{Copy columns}]
    %         at (axis cs:1.5, 3000) {};
    %     \node [label=below:{Use previous result}]
    %         at (axis cs:6, 3000) {};
            
        % \draw[thick] ({axis cs:3.85,0}|-{rel axis cs:0,0}) -- ({axis cs:3.85,0}|-{rel axis cs:0,1});
        
    % \end{scope}
    
    % \addplot[red,sharp plot,update limits=false,]
    % coordinates {(-50, 150) (50, 150)}
    
    \end{axis}
\end{tikzpicture}

\caption{Execution (wall-clock) time for Loomis-Whitney queries on ego-Twitter (x-axis has the number of join attributes in each table).}
\label{fig:attributescalability}
\end{center}
\end{figure}

% For the Loomis-Whitney with degree $n + 1$, the input table is Loomis-Whitney with degree $n$.

We use the query result of Loomis-Whitney with degree $n$ as the input table of Loomis-Whitney query with degree $n + 1$ (truncating the query result to 200M rows as, otherwise, none of the compared systems finish for the highest degree within the timeout of ten minutes). For the different degrees, in ascending order, the table sizes are (approximately) 13M rows, 105M rows, and 200M rows.%For the degree 5, we truncate the input table to 200,000,000, since all systems cannot finish using the origin input table which contains roughly 800,000,000 rows.

The wall-clock execution time for Loomis-Whitney queries on ego-Twitter is depicted in Figure~\ref{fig:attributescalability}. ADOPT performs better than the other baselines and the relative performance gap grows as the number of join attributes increases. Clearly, without adaptive processing, finding good query plans becomes harder as queries become more complex. In addition, we measure overheads for index creation before run time for all systems separating index creation from query evaluation (see Appendix~\ref{sec:indexcreation} for details). For Postgres, we create indexes on each column of the input table. EmptyHeaded automatically selects indexes to create. For ADOPT, we create indexes to support all possible attribute orders. Table~\ref{tab:lmindexing} reports corresponding results. For both baselines implementing worst-case optimal join algorithms (ADOPT and EmptyHeaded), index creation time is higher than for Postgres and grows faster with increasing degree (note that, as discussed previously, the size of the input data increases as well). However, among the two baselines with worst-case optimal join algorithms, ADOPT generates indexes significantly faster. For all systems, the resulting indexes can be reused across all future queries.

\begin{table}[b]
    \centering
        \caption{Index generation times for different systems (in seconds), preparing evaluation of Loomis-Whitney queries with varying degree.}
    \begin{tabular}{l|rrr}
    \toprule[1pt]
    \textbf{System} & \textbf{Degree 3} & \textbf{Degree 4} & \textbf{Degree 5} \\
    \midrule[1pt]
         Postgres & 14 & 161 & 407\\
         ADOPT & 3 & 81 & 1003 \\
         EmptyHeaded & 99 & 949 & 3719 \\
    \bottomrule[1pt]
    \end{tabular}
    \label{tab:lmindexing}
\end{table}

%ADOPT outperforms all other baseline methods once the number of join attributes in each table is at least three. This is because non-adaptive approaches face significant challenges in selecting a good execution plan, making it difficult for them to compete with ADOPT's adaptive approach.

%\newpage
\section{Illustration of LeapFrog TrieJoin}
\label{sec:illustrating_LFTJ}
In this section, we illustrate the LFTJ algorithm.
In Section~\ref{sec:leapfrog_join},
we describe leapfrog join on unary relations, which is the basic building block of LFTJ.
We explain in Section~\ref{sec:trie_representation} how LFTJ traverses non-unary relations. 
In Sections~\ref{sec:lftj_acyclic_query} and \ref{sec:lftj_cyclic_query}, we 
illustrate LFTJ for an acyclic and respectively a cyclic query.

\subsection{Leapfrog Join on Unary Relations}
\label{sec:leapfrog_join}
Assume we want to  
join several unary relations over the same attribute. 
This amounts to computing the intersection of the relations. 
The leapfrog join algorithm navigates each relation
using an iterator that sees the relation as an ordered list. 
The iterators provide the following operations:
%$value()$ returns the value at the current position
%of the iterator;
$next()$ moves the iterator to the next position 
in the list, or to \textbf{EOF} if no such position exists;
given a value $v$, $seek(v)$ moves the iterator to the position with the least 
value $w$ such that $w \geq  v$, 
%(using binary search) 
or to \textbf{EOF} if no such value exists.
%$atEnd()$ returns true if the iterator is at the end.
At the beginning, each iterator is at the first position of its list.
As long as all values  
at the current positions of the iterators do not match,
the leapfrog join algorithm proceeds as follows. Given that the largest 
value at the current positions of the iterators is $v$,
the algorithm   calls $seek(v)$ for one of the iterators with the smallest current value. In case the values 
at the current iterator positions match,
the common value is added to the output. Then, the algorithm 
calls $next()$ for one of the list and repeats the above strategy 
to find the next common value. The algorithm stops once 
one of the iterators reaches \textbf{EOF}.  
%The algorithm runs in time 
%$\Theta(N_{\text{min}}(1 + \log (\frac{N_{\text{max}}}{N_{\text{min}}})))$, 
%where $N_{\text{min}}$ and $N_{\text{max}}$ are the sizes 
%of the smallest and respectively largest relation. 

\begin{figure}[t]
\centering
\begin{center}
%\begin{tabular}{c}
\begin{tabular}{ccc@{\hspace*{2em}}c}
$r(a)$ & $s(a)$ & $t(a)$ & $\text{join of } r, s,\text{ and } t$\smallskip \\
\begin{tabular}{c}\toprule
$0$  \\
$1$  \\
$3$  \\
$4$  \\
$5$  \\
$6$ \\
$7$ \\
$8$ \\
$9$ \\
$11$ \\
\bottomrule
\end{tabular}
&
\begin{tabular}{c}\toprule
$0$  \\
$2$  \\
$6$  \\
$7$  \\
$8$  \\
$9$  \\
\bottomrule
\\
\\
\\
\\
\end{tabular}
&
\begin{tabular}{c}\toprule
$2$  \\
$4$  \\
$5$  \\
$8$  \\
$10$  \\\bottomrule
\\
\\
\\
\\
\\
\end{tabular}&
\begin{tabular}{c}\toprule
8  \\\bottomrule
\\
\\
\\
\\
\\
\\
\\
\\
\\
\end{tabular}
\end{tabular}
\end{center}
\hspace*{-0.5cm}
\caption{
Unary relations and their join.
}
\label{fig:join_unary_relations}
\end{figure}

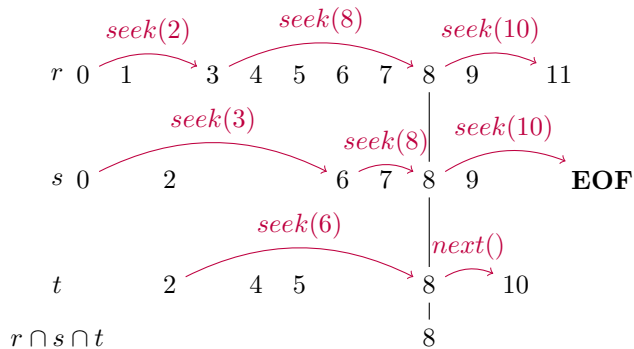
\begin{figure}[b]
\begin{center}
\begin{tikzpicture}[xscale=1.15]
%Relation r
\node (A) at (-0.3,0) {$r$};
\node (A0) at (0,0) {$0$};
\node (A1) at (0.5,0) {$1$};
\node (A3) at (1.5,0) {$3$};
\node (A4) at (2,0) {$4$};
\node (A5) at (2.5,0) {$5$};
\node (A6) at (3,0) {$6$};
\node (A7) at (3.5,0) {$7$};
\node (A8) at (4,0) {$8$};
\node (A9) at (4.5,0) {$9$};
\node (A11) at (5.5,0) {$11$};

\draw[->,purple] (A0) to[bend left] node[above] {$seek(2)$} (A3);
\draw[->,purple] (A3) to[bend left] node[above] {$seek(8)$} (A8);
\draw[->,purple] (A8) to[bend left] node[above] {$seek(10)$} (A11);

%Relation B
\node (B) at (-0.3,-1.4) {$s$};
\node (B0) at (0,-1.4) {$0$};
\node (B2) at (1,-1.4) {$2$};
\node (B6) at (3,-1.4) {$6$};
\node (B7) at (3.5,-1.4) {$7$};
\node (B8) at (4,-1.4) {$8$};
\node (B9) at (4.5,-1.4) {$9$};
\node (Beof) at (6,-1.4) {\textbf{EOF}};

\draw[->,purple] (B0) to[bend left] node[above] {$seek(3)$} (B6);
\draw[->,purple] (B6) to[bend left] node[above] {$seek(8)$} (B8);
\draw[->,purple] (B8) to[bend left] node[above] {$seek(10)$} (Beof);

%Relation C
\node (C) at (-0.3,-2.8) {$t$};
\node (C2) at (1,-2.8) {$2$};
\node (C4) at (2,-2.8) {$4$};
\node (C5) at (2.5,-2.8) {$5$};
\node (C8) at (4,-2.8) {$8$};

\node (C10) at (5,-2.8) {$10$};

\draw[->,purple] (C2) to[bend left] node[above] {$seek(6)$} (C8);
\draw[->,purple] (C8) to[bend left] node[above] {$next()$} (C10);

\draw[-] (A8) to  (B8);
\draw[-] (B8) to  (C8);

%output
\node (intersection) at (-0.3,-3.5) {$r \cap s \cap t$};
\node (result) at (4,-3.5) {$8$};
\draw[-] (C8) to  (result);

\end{tikzpicture}
\end{center}
\hspace*{-0.5cm}
\caption{Joining the unary relations in Figure~\ref{fig:join_unary_relations} using leapfrog join. The only value in the result is $8$.}
\label{fig:leapfrog_join}
\end{figure}

\begin{example}
\rm
We illustrate leapfrog join 
for the query 
$q(a)$ $=$ $r(a), s(a), t(a)$
that computes the intersection of the 
three unary relations
depicted in Figure~\ref{fig:join_unary_relations}.
Figure~\ref{fig:leapfrog_join} visualizes 
how the iterators traverse the three relations. 
The iterators start at the initial positions 
of the ordered lists representing
the relations. 
The values at the initial positions do not match 
and the largest such value is the value $2$ in $t$.
Hence, the algorithm calls $seek(2)$ for $r$, which moves the 
iterator of $r$ to the position with value $3$. 
Now, the largest value of the current iterator positions is
$3$, so the algorithm calls 
$seek(3)$ for $s$. This operation moves the iterator of $s$ to the position with value $6$. Then, it calls $seek(6)$ for $t$, which moves $t$'s iterator  to the position with value $8$.
Afterwards, it calls $seek(8)$ for $r$ and then for $s$, upon which the iterators of both relations move to their respective positions holding value $8$.
Now, all iterators point to the value $8$, so we have a match. 
The algorithm adds $8$ to the output and moves
the iterator of $t$ to the next position, which holds 
the value $10$. 
Then, it calls $seek(10)$ for $r$, which  moves $r$'s iterator to the position with value 
$11$. Finally, it calls $seek(11)$ for $s$, upon which the iterator of $s$ moves to \textbf{EOF}, so the algorithm stops.

We conclude that the only output value is $8$.
\end{example}

\subsection{Navigation over Non-Unary Relations}
\label{sec:trie_representation}
LFTJ navigates non-unary relations using iterators that interpret
the relations as tries that follow attribute orders. 
Each level in a trie corresponds to one attribute.
The iterators support the following operations:
$open()$ moves the iterator to the first child node of the current node;
$up()$ returns the iterator to the parent node; 
$next()$ moves the iterator to the next sibling
or \textbf{EOF} if no such sibling exists;
given a value $v$, $seek(v)$ moves the iterator to the sibling with the least 
value $w$ such that $w \geq  v$, 
or to \textbf{EOF} if no such value exists.

\begin{figure}[t]
\begin{center}
\begin{minipage}{0.3\linewidth}
\begin{tabular}{c}
$r({\color{red} a}, {\color{blue} b}, {\color{goodgreen} c})$  \smallskip \\
\begin{tabular}{ccc}\toprule
\color{red}$0$ & \color{blue} $3$ & \color{goodgreen} $4$ \\
\color{red}$0$ & \color{blue} $3$ & \color{goodgreen} $5$ \\
\color{red}$0$ & \color{blue} $4$ & \color{goodgreen} $0$ \\
\color{red}$0$ & \color{blue} $4$ & \color{goodgreen} $1$ \\
\color{red}$0$ & \color{blue} $4$ & \color{goodgreen} $3$ \\
\color{red}$0$ & \color{blue} $5$ & \color{goodgreen} $2$ \\
\color{red}$1$ & \color{blue} $5$ & \color{goodgreen} $2$  \\\bottomrule
\end{tabular}
\end{tabular}
\end{minipage}
%\hspace{0.5cm}
\begin{minipage}{0.5\linewidth}
\begin{tikzpicture}[xscale=0.55, yscale=0.5]
  \node at (-1,1) (u1) {\textbullet};

  \node at (-4,-1) (a0) {$0$} edge[-] (u1);
	
  	\node at (-7,-3) (b3) {$3$} edge[-] (a0);

		\node at (-8,-5) (c4) {$4$} edge[-] (b3);

  		\node at (-6,-5) (c5) {$5$} edge[-] (b3);

  	\node at (-4,-3) (b4) {$4$} edge[-] (a0);

  		\node at (-5,-5) (c0) {$0$} edge[-] (b4);

  		\node at (-4,-5) (c1) {$1$} edge[-] (b4);

  		\node at (-3,-5) (c3) {$3$} edge[-] (b4);

  	\node at (-2,-3) (b5) {$5$} edge[-] (a0);

		\node at (-2,-5) (c2) {$2$} edge[-] (b5);

  \node at (0,-1) (a1) {$1$} edge[-] (u1);

  	\node at (0,-3) (b5) {$5$} edge[-] (a1);

  		\node at (0,-5) (c2) {$2$} edge[-] (b5);
\end{tikzpicture}
\end{minipage}
\end{center}
\hspace*{-0.5cm}
\caption{A relation and its interpretation as a trie  following the attribute order $a-b-c$.}
\label{fig:trie_representation}
\end{figure}
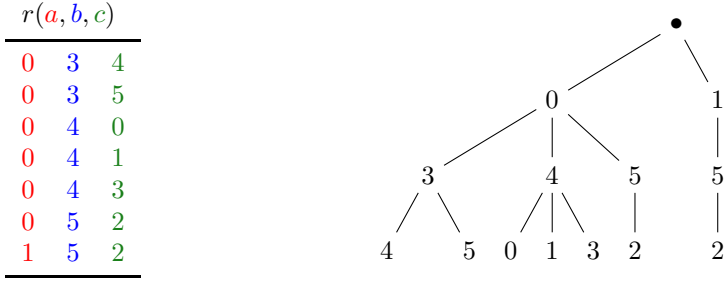

\begin{figure}[b]
\begin{center}
\begin{tikzpicture}[xscale=1, yscale=0.6]
  \node at (-1,1) (u1) {\textbullet};

  \node at (-4,-1) (a0) {$0$} edge[-] (u1);
	
  	\node at (-7,-3) (b3) {$3$} edge[-] (a0);

		\node at (-7.5,-5) (c4) {$4$} edge[-] (b3);

  		\node at (-6.5,-5) (c5) {$5$} edge[-] (b3);

  	\node at (-4,-3) (b4) {$4$} edge[-] (a0);

  		\node at (-5,-5) (c0) {$0$} edge[-] (b4);

  		\node at (-4,-5) (c1) {$1$} edge[-] (b4);

  		\node at (-3,-5) (c3) {$3$} edge[-] (b4);

  	\node at (-2,-3) (b5) {$5$} edge[-] (a0);

		\node at (-2,-5) (c23) {$2$} edge[-] (b5);

  \node at (0,-1) (a1) {$1$} edge[-] (u1);

  	\node at (0,-3) (b5) {$5$} edge[-] (a1);

  		\node at (0,-5) (c23) {$2$} edge[-] (b5);

    \draw[->,purple] (u1) to[bend right] node[above] {$open()$} (a0);
    \draw[->,purple] (a0) to[bend right] node[above] {$open()$} (b3);
    \draw[->,purple] (b3) to node[above,xshift=2em] {$next()$} (b4);
    \draw[->,purple] (b4) to[bend right] node[left] {$open()$} (c0);
    \draw[->,purple] (c0) to[bend right] node[below] {$seek(2)$} (c3);
  \draw[->,purple] (c3) to[bend right] node[right] {$up()$} (b4);    
\end{tikzpicture}
\end{center}
\hspace*{-0.5cm}
\caption{Traversal over the values of relation $r$
in Figure \ref{fig:trie_representation}.}
\label{fig:trie_traversal}
\end{figure}
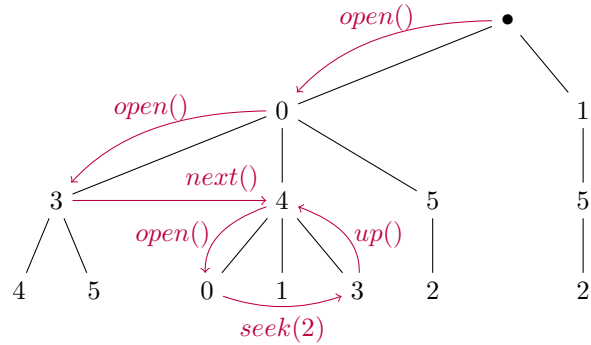

\begin{example}
\rm
Figure~\ref{fig:trie_representation} depicts a relation $r$ and its 
interpretation as a trie that follows the 
attribute order $a-b-c$. The children of each node are sorted. 
The children of the root carry the 
$a$-values of $r$, which are $0$ and $1$. The children of the $a$-value $0$  
carry the values $3$, $4$, and $5$, which are the $b$-values paired with $0$ in $r$.
The children of the $c$-value $3$ carry the values $4$ and $5$, 
which are the $c$-values paired with $0$ and $3$ in $r$. 
The rest of the trie is organized analogously. 
Figure~\ref{fig:trie_traversal} visualizes  how an iterator 
traverses the values in relation $r$
via the operation sequence $open()$, $open()$, $next()$, $open()$, $seek(2)$, $up()$.
\end{example}

\subsection{Leapfrog Triejoin for an Acyclic Query}
\label{sec:lftj_acyclic_query}
Given a set of relations, a global attribute order 
is an ordering of all attributes appearing in the relations.  
LFTJ requires that all input relations can be navigated 
following a global attribute order.
Consider a set of relations and a global attribute order 
$att_1, \ldots , att_n$.
To create all tuples in the join result, LFTJ proceeds as follows. 
It uses leapfrog join to fix the first $att_1$-value that appears in all relations 
containing attribute $att_1$.
Given that the values for $att_1, \ldots , att_i$ with $i < n$
are already fixed to values $a_1, \ldots , a_i$, it uses again leapfrog join to fix the first 
$att_{i+1}$-value 
that appears in all relations containing $att_{i+1}$ when restricted  
to $a_1, \ldots , a_i$.
Once all attributes are fixed 
to values $a_1, \ldots , a_n$, it means that 
$(a_1, \ldots , a_n)$ constitutes a tuple in the join result, so
the algorithm adds it to the output. 
Then, it triggers leapfrog join to traverse the remaining   
$att_{n}$-values that appear in all relations containing $att_{n}$ when restricted  
to $a_1, \ldots , a_{n-1}$. For each such $att_{n}$-value $a_n'$, it 
adds  
$(a_1, \ldots , a_n')$ to the output. Once all $att_{n}$-values are exhausted, it backtracks 
and searches for the next $att_{n-1}$-value that appears in the join 
in the context of $(a_1, \ldots , a_{n-2})$, and so on.
In the following, we illustrate how LFTJ computes an acyclic
join. 
 
\begin{figure}[t]
\centering
\begin{center}
\begin{tabular}{c}
\begin{tabular}{cccc@{\hspace*{2em}}c}
$r({\color{red}a},{\color{blue} b}, {\color{goodgreen} c})$ & $s({\color{red}a},{\color{goodgreen} c})$ & $t({\color{blue}b})$ & 
$u({\color{blue} b},{\color{goodgreen} c})$ &
$\text{join of } r, s, t, \text{and } u$
\smallskip \\
\begin{tabular}{rrr}\toprule
\color{red}$0$ &\color{blue} $2$ &\color{goodgreen} $1$ \\
\color{red}$0$ &\color{blue} $3$ &\color{blue} $1$ \\
\color{red}$1$ &\color{blue} $5$ &\color{goodgreen} $2$\\\bottomrule
\end{tabular}
&
\begin{tabular}{rr}\toprule
\color{red}$0$ & \color{goodgreen}$1$ \\
\color{red}$0$ & \color{goodgreen}$2$ \\
\color{red}$2$ & \color{goodgreen}$3$ \\\bottomrule
\end{tabular}
&
\begin{tabular}{r}\toprule
\color{blue}$0$  \\
\color{blue}$2$  \\
\color{blue}$3$  \\\bottomrule
\end{tabular}
&
\begin{tabular}{rr}\toprule
\color{blue}$0$ & \color{goodgreen}$0$  \\
\color{blue}$0$ & \color{goodgreen}$2$ \\
\color{blue}$2$ & \color{goodgreen}$1$   \\\bottomrule
\end{tabular}
&
\begin{tabular}{rrr}\toprule
\color{red}$0$ & \color{blue}$2$ & \color{goodgreen}$1$\\\bottomrule
\\
\\
\end{tabular}
\end{tabular}
\end{tabular}
\end{center}
\hspace*{-0.5cm}
\caption{
Acyclic join of four relations. 
}
\label{fig:acyclic_join}
\end{figure} 

\begin{figure*}[h]
\centering
%\resizebox{0.64\textwidth}{!}{
% \adjustbox{width=0.72\textwidth}{
 \adjustbox{width=12cm}{
\begin{tabular}{ccccc}
\begin{tikzpicture}[xscale=0.55, yscale=0.4]
  \node at (0,0) (a) {};
  \node at (0,-2) (a) {$a$};
    \node at (0,-4) (b) {$b$} edge[-] (a);
      \node at (0,-6) (c) {$c$} edge[-] (b);

\end{tikzpicture} &
%%%%%%
% FIRST relation
%%%%%%
\begin{tikzpicture}[xscale=0.55, yscale=0.4]
  \node at (0,0.9) (r) {r};
  \node[drop shadow, fill=gray!20, draw] at (0,0) (rr) {\textbullet};
  \node at (-1,-2) (ra0) {$0$} edge[-] (rr);
    \node at (1,-2) (ra1) {$1$} edge[-] (rr);

  \node at (-1.5,-4) (rb2) {$2$} edge[-] (ra0);
	  \node at (-0.5,-4) (rb3) {$3$} edge[-] (ra0);

     \node at (-1.5,-6) (rc1) {$1$} edge[-] (rb2);
	  \node at (-0.5,-6) (rc1) {$1$} edge[-] (rb3);
    
    \node at (1,-4) (rb5) {$5$} edge[-] (ra1);
        \node at (1,-6) (rc2) {$2$} edge[-] (rb5);
\end{tikzpicture}
%\hspace*{-1em}
&
%%%%%%
% SECOND relation
%%%%%%
\begin{tikzpicture}[xscale=0.55, yscale=0.4]
  \node at (0,0.9) (s) {s};    
    \node[drop shadow, fill=gray!20, draw] at (0,0) (rs) {\textbullet};
  \node at (-1,-2) (sa0) {$0$} edge[-] (rs);
    \node at (1,-2) (sa2) {$2$} edge[-] (rr);

     \node at (-1.5,-6) (sc1) {$1$} edge[-] (sa0);
	  \node at (-0.5,-6) (sc2) {$2$} edge[-] (sa0);
    
    \node at (1,-6) (sc3) {$3$} edge[-] (sa2);

\end{tikzpicture}
&
%%%%%%
% THIRD relation
%%%%%%
\begin{tikzpicture}[xscale=0.7, yscale=0.4]
  \node at (0,0.9) (t) {t};          
      \node[drop shadow, fill=gray!20, draw] at (0,0) (rt) {\textbullet};
  \node at (-1,-4) (tb0) {$0$} edge[-] (rt);
    \node at (0,-4) (tb2) {$2$} edge[-] (rt);
 \node at (1,-4) (tb3) {$3$} edge[-] (rt);
\node at (0,-6) {};
\end{tikzpicture}
&
%fourth relation
\begin{tikzpicture}[xscale=0.7, yscale=0.4]
  \node at (0,0.9) (u) {u};    
\node[drop shadow, fill=gray!20, draw] at (0,0) (ru) {\textbullet};
  \node at (-1,-4) (ub0) {$0$} edge[-] (ru);
    \node at (1,-4) (ub2) {$2$} edge[-] (ru);
  \node at (-1.5,-6) (uc0) {$0$} edge[-] (ub0);
   \node at (-0.5,-6) (uc2) {$2$} edge[-] (ub0); 
    \node at (1,-6) (uc1) {$1$} edge[-] (ub2);
\end{tikzpicture} 
\\
%Second round
%%%%%%
% FIRST relation
%%%%%%
\begin{tikzpicture}[xscale=0.55, yscale=0.4]
  \node at (0,0) (a) {};
  \node at (0,-2) (a) {$a$};
    \node at (0,-4) (b) {$b$} edge[-] (a);
      \node at (0,-6) (c) {$c$} edge[-] (b);

\end{tikzpicture} &
\begin{tikzpicture}[xscale=0.55, yscale=0.4]
  \node at (0,0) (rr) {\textbullet};
  \node[drop shadow, fill=gray!20, draw] at (-1,-2) (ra0) {$0$} edge[-] (rr);
    \node at (1,-2) (ra1) {$1$} edge[-] (rr);

  \node at (-1.5,-4) (rb2) {$2$} edge[-] (ra0);
	  \node at (-0.5,-4) (rb3) {$3$} edge[-] (ra0);

     \node at (-1.5,-6) (rc1) {$1$} edge[-] (rb2);
	  \node at (-0.5,-6) (rc1) {$1$} edge[-] (rb3);
    
    \node at (1,-4) (rb5) {$5$} edge[-] (ra1);
        \node at (1,-6) (rc2) {$2$} edge[-] (rb5);
        \draw[->,purple] (rr) to[bend right] node[left] {$open()$} (ra0);
\end{tikzpicture}
%\hspace*{-1em}
&
%%%%%%
% SECOND relation
%%%%%%
\begin{tikzpicture}[xscale=0.55, yscale=0.4]
    \node at (0,0) (rs) {\textbullet};
  \node[drop shadow, fill=gray!20, draw] at (-1,-2) (sa0) {$0$} edge[-] (rs);
    \node at (1,-2) (sa2) {$2$} edge[-] (rr);

     \node at (-1.5,-6) (sc1) {$1$} edge[-] (sa0);
	  \node at (-0.5,-6) (sc2) {$2$} edge[-] (sa0);
    
    \node at (1,-6) (sc3) {$3$} edge[-] (sa2);

	        \draw[->,purple] (rs) to[bend right] node[left] {$open()$} (sa0);
\end{tikzpicture}
&
%%%%%%
% THIRD relation
%%%%%%
\begin{tikzpicture}[xscale=0.7, yscale=0.4]
      \node[drop shadow, fill=gray!20, draw] at (0,0) (rt) {\textbullet};
  \node at (-1,-4) (tb0) {$0$} edge[-] (rt);
    \node at (0,-4) (tb2) {$2$} edge[-] (rt);
 \node at (1,-4) (tb3) {$3$} edge[-] (rt);
\node at (0,-6) {};
\end{tikzpicture}
&
%fourth relation
\begin{tikzpicture}[xscale=0.7, yscale=0.4]
\node[drop shadow, fill=gray!20, draw] at (0,0) (ru) {\textbullet};
  \node at (-1,-4) (ub0) {$0$} edge[-] (ru);
    \node at (1,-4) (ub2) {$2$} edge[-] (ru);
  \node at (-1.5,-6) (uc0) {$0$} edge[-] (ub0);
   \node at (-0.5,-6) (uc2) {$2$} edge[-] (ub0); 
    \node at (1,-6) (uc1) {$1$} edge[-] (ub2);
\end{tikzpicture}
\\
%Third round
%%%%%%
% FIRST relation
%%%%%%
\begin{tikzpicture}[xscale=0.55, yscale=0.4]
  \node at (0,0) (a) {};
  \node at (0,-2) (a) {$a$};
    \node at (0,-4) (b) {$b$} edge[-] (a);
      \node at (0,-6) (c) {$c$} edge[-] (b);

\end{tikzpicture} &
\begin{tikzpicture}[xscale=0.55, yscale=0.4]
  \node at (0,0) (rr) {\textbullet};
  \node at (-1,-2) (ra0) {$0$} edge[-] (rr);
    \node at (1,-2) (ra1) {$1$} edge[-] (rr);

  \node[drop shadow, fill=gray!20, draw] at (-1.5,-4) (rb2) {$2$} edge[-] (ra0);
	  \node at (-0.5,-4) (rb3) {$3$} edge[-] (ra0);

     \node at (-1.5,-6) (rc1) {$1$} edge[-] (rb2);
	  \node at (-0.5,-6) (rc1) {$1$} edge[-] (rb3);
    
    \node at (1,-4) (rb5) {$5$} edge[-] (ra1);
        \node at (1,-6) (rc2) {$2$} edge[-] (rb5);
        \draw[->,purple] (ra0) to[bend right] node[left] {$open()$} (rb2);
\end{tikzpicture}
%\hspace*{-1em}
&
%%%%%%
% SECOND relation
%%%%%%
\begin{tikzpicture}[xscale=0.55, yscale=0.4]
    \node at (0,0) (rs) {\textbullet};
  \node[drop shadow, fill=gray!20, draw] at (-1,-2) (sa0) {$0$} edge[-] (rs);
    \node at (1,-2) (sa2) {$2$} edge[-] (rr);

     \node at (-1.5,-6) (sc1) {$1$} edge[-] (sa0);
	  \node at (-0.5,-6) (sc2) {$2$} edge[-] (sa0);
    
    \node at (1,-6) (sc3) {$3$} edge[-] (sa2);

\end{tikzpicture}
&
%%%%%%
% THIRD relation
%%%%%%
\begin{tikzpicture}[xscale=0.7, yscale=0.4]
      \node at (0,0) (rt) {\textbullet};
  \node[drop shadow, fill=gray!20, draw] at (-1,-4) (tb0) {$0$} edge[-] (rt);
    \node at (0,-4) (tb2) {$2$} edge[-] (rt);
 \node at (1,-4) (tb3) {$3$} edge[-] (rt);
\node at (0,-6) {};
\draw[->,purple] (rt) to[bend right] node[left] {$open()$} (tb0);
\end{tikzpicture}
&
%fourth relation
\begin{tikzpicture}[xscale=0.7, yscale=0.4]
\node at (0,0) (ru) {\textbullet};
  \node[drop shadow, fill=gray!20, draw] at (-1,-4) (ub0) {$0$} edge[-] (ru);
    \node at (1,-4) (ub2) {$2$} edge[-] (ru);
  \node at (-1.5,-6) (uc0) {$0$} edge[-] (ub0);
   \node at (-0.5,-6) (uc2) {$2$} edge[-] (ub0); 
    \node at (1,-6) (uc1) {$1$} edge[-] (ub2);

\draw[->,purple] (ru) to[bend right] node[left] {$open()$} (ub0);
\end{tikzpicture}
\\
%Fourth round
\begin{tikzpicture}[xscale=0.55, yscale=0.4]
  \node at (0,0) (a) {};
  \node at (0,-2) (a) {$a$};
    \node at (0,-4) (b) {$b$} edge[-] (a);
      \node at (0,-6) (c) {$c$} edge[-] (b);

\end{tikzpicture} &
%%%%%%
% FIRST relation
%%%%%%
\begin{tikzpicture}[xscale=0.55, yscale=0.4]
  \node at (0,0) (rr) {\textbullet};
  \node at (-1,-2) (ra0) {$0$} edge[-] (rr);
    \node at (1,-2) (ra1) {$1$} edge[-] (rr);

  \node[drop shadow, fill=gray!20, draw] at (-1.5,-4) (rb2) {$2$} edge[-] (ra0);
	  \node at (-0.5,-4) (rb3) {$3$} edge[-] (ra0);

     \node at (-1.5,-6) (rc1) {$1$} edge[-] (rb2);
	  \node at (-0.5,-6) (rc1) {$1$} edge[-] (rb3);
    
    \node at (1,-4) (rb5) {$5$} edge[-] (ra1);
        \node at (1,-6) (rc2) {$2$} edge[-] (rb5);
\end{tikzpicture}
%\hspace*{-1em}
&
%%%%%%
% SECOND relation
%%%%%%
\begin{tikzpicture}[xscale=0.55, yscale=0.4]
    \node at (0,0) (rs) {\textbullet};
  \node[drop shadow, fill=gray!20, draw] at (-1,-2) (sa0) {$0$} edge[-] (rs);
    \node at (1,-2) (sa2) {$2$} edge[-] (rr);

     \node at (-1.5,-6) (sc1) {$1$} edge[-] (sa0);
	  \node at (-0.5,-6) (sc2) {$2$} edge[-] (sa0);
    
    \node at (1,-6) (sc3) {$3$} edge[-] (sa2);

\end{tikzpicture}
&
%%%%%%
% THIRD relation
%%%%%%
\begin{tikzpicture}[xscale=0.7, yscale=0.4]
      \node at (0,0) (rt) {\textbullet};
  \node at (-1,-4) (tb0) {$0$} edge[-] (rt);
    \node[drop shadow, fill=gray!20, draw] at (0,-4) (tb2) {$2$} edge[-] (rt);
 \node at (1,-4) (tb3) {$3$} edge[-] (rt);
\node at (0,-6) {};
\draw[->,purple] (tb0) to[bend right=90] node[below] {$seek(2)$} (tb2);
\end{tikzpicture}
&
%fifth relation
\begin{tikzpicture}[xscale=0.7, yscale=0.4]
\node at (0,0) (ru) {\textbullet};
  \node[drop shadow, fill=gray!20, draw] at (-1,-4) (ub0) {$0$} edge[-] (ru);
    \node at (1,-4) (ub2) {$2$} edge[-] (ru);
  \node at (-1.5,-6) (uc0) {$0$} edge[-] (ub0);
   \node at (-0.5,-6) (uc2) {$2$} edge[-] (ub0); 
    \node at (1,-6) (uc1) {$1$} edge[-] (ub2);

\end{tikzpicture}
\\
%Fifth round
\begin{tikzpicture}[xscale=0.55, yscale=0.4]
  \node at (0,0) (a) {};
  \node at (0,-2) (a) {$a$};
    \node at (0,-4) (b) {$b$} edge[-] (a);
      \node at (0,-6) (c) {$c$} edge[-] (b);

\end{tikzpicture} &
%%%%%%
% FIRST relation
%%%%%%
\begin{tikzpicture}[xscale=0.55, yscale=0.4]
  \node at (0,0) (rr) {\textbullet};
  \node at (-1,-2) (ra0) {$0$} edge[-] (rr);
    \node at (1,-2) (ra1) {$1$} edge[-] (rr);

  \node[drop shadow, fill=gray!20, draw] at (-1.5,-4) (rb2) {$2$} edge[-] (ra0);
	  \node at (-0.5,-4) (rb3) {$3$} edge[-] (ra0);

     \node at (-1.5,-6) (rc1) {$1$} edge[-] (rb2);
	  \node at (-0.5,-6) (rc1) {$1$} edge[-] (rb3);
    
    \node at (1,-4) (rb5) {$5$} edge[-] (ra1);
        \node at (1,-6) (rc2) {$2$} edge[-] (rb5);
\end{tikzpicture}
%\hspace*{-1em}
&
%%%%%%
% SECOND relation
%%%%%%
\begin{tikzpicture}[xscale=0.55, yscale=0.4]
    \node at (0,0) (rs) {\textbullet};
  \node[drop shadow, fill=gray!20, draw] at (-1,-2) (sa0) {$0$} edge[-] (rs);
    \node at (1,-2) (sa2) {$2$} edge[-] (rr);

     \node at (-1.5,-6) (sc1) {$1$} edge[-] (sa0);
	  \node at (-0.5,-6) (sc2) {$2$} edge[-] (sa0);
    
    \node at (1,-6) (sc3) {$3$} edge[-] (sa2);

\end{tikzpicture}
&
%%%%%%
% THIRD relation
%%%%%%
\begin{tikzpicture}[xscale=0.7, yscale=0.4]
      \node at (0,0) (rt) {\textbullet};
  \node at (-1,-4) (tb0) {$0$} edge[-] (rt);
    \node[drop shadow, fill=gray!20, draw] at (0,-4) (tb2) {$2$} edge[-] (rt);
 \node at (1,-4) (tb3) {$3$} edge[-] (rt);
\node at (0,-6) {};
\end{tikzpicture}
&
%fourth relation
\begin{tikzpicture}[xscale=0.7, yscale=0.4]
\node at (0,0) (ru) {\textbullet};
  \node at (-1,-4) (ub0) {$0$} edge[-] (ru);
    \node[drop shadow, fill=gray!20, draw] at (1,-4) (ub2) {$2$} edge[-] (ru);
  \node at (-1.5,-6) (uc0) {$0$} edge[-] (ub0);
   \node at (-0.5,-6) (uc2) {$2$} edge[-] (ub0); 
    \node at (1,-6) (uc1) {$1$} edge[-] (ub2);
\draw[->,purple] (ub0) to[bend right] node[below] {$seek(2)$} (ub2);
\end{tikzpicture}
\\
%sixth round
\begin{tikzpicture}[xscale=0.55, yscale=0.4]
  \node at (0,0) (a) {};
  \node at (0,-2) (a) {$a$};
    \node at (0,-4) (b) {$b$} edge[-] (a);
      \node at (0,-6) (c) {$c$} edge[-] (b);

\end{tikzpicture} &
%%%%%%
% FIRST relation
%%%%%%
\begin{tikzpicture}[xscale=0.55, yscale=0.4]
  \node at (0,0) (rr) {\textbullet};
  \node at (-1,-2) (ra0) {$0$} edge[-] (rr);
    \node at (1,-2) (ra1) {$1$} edge[-] (rr);

  \node at (-1.5,-4) (rb2) {$2$} edge[-] (ra0);
	  \node at (-0.5,-4) (rb3) {$3$} edge[-] (ra0);

     \node[drop shadow, fill=gray!20, draw] at (-1.5,-6) (rc1) {$1$} edge[-] (rb2);
	  \node at (-0.5,-6) (rc12) {$1$} edge[-] (rb3);
    
    \node at (1,-4) (rb5) {$5$} edge[-] (ra1);
        \node at (1,-6) (rc2) {$2$} edge[-] (rb5);
\draw[->,purple] (rb2) to[bend right] node[left] {$open()$}(rc1);
\end{tikzpicture}
%\hspace*{-1em}
&
%%%%%%
% SECOND relation
%%%%%%
\begin{tikzpicture}[xscale=0.55, yscale=0.4]
    \node at (0,0) (rs) {\textbullet};
  \node at (-1,-2) (sa0) {$0$} edge[-] (rs);
    \node at (1,-2) (sa2) {$2$} edge[-] (rr);

     \node[drop shadow, fill=gray!20, draw] at (-1.5,-6) (sc1) {$1$} edge[-] (sa0);
	  \node at (-0.5,-6) (sc2) {$2$} edge[-] (sa0);
    
    \node at (1,-6) (sc3) {$3$} edge[-] (sa2);

\draw[->,purple] (sa0) to[bend right] node[left] {$open()$}(sc1);
\end{tikzpicture}
&
%%%%%%
% THIRD relation
%%%%%%
\begin{tikzpicture}[xscale=0.7, yscale=0.4]
      \node at (0,0) (rt) {\textbullet};
  \node at (-1,-4) (tb0) {$0$} edge[-] (rt);
    \node[drop shadow, fill=gray!20, draw] at (0,-4) (tb2) {$2$} edge[-] (rt);
 \node at (1,-4) (tb3) {$3$} edge[-] (rt);
\node at (0,-6) {};
\end{tikzpicture}
&
%fourth relation
\begin{tikzpicture}[xscale=0.7, yscale=0.4]
\node at (0,0) (ru) {\textbullet};
  \node at (-1,-4) (ub0) {$0$} edge[-] (ru);
    \node at (1,-4) (ub2) {$2$} edge[-] (ru);
  \node at (-1.5,-6) (uc0) {$0$} edge[-] (ub0);
   \node at (-0.5,-6) (uc2) {$2$} edge[-] (ub0); 
    \node[drop shadow, fill=gray!20, draw] at (1,-6) (uc1) {$1$} edge[-] (ub2);

\draw[->,purple] (ub2) to[bend left] node[right] {$open()$}(uc1);
\end{tikzpicture}
\end{tabular}
}
%\hspace*{-0.5cm}
\caption{
LFTJ execution for the input relations in Figure \ref{fig:acyclic_join} following the attribute order $a - b - c$ (depicted to the left).
}
\label{fig:trie_acyclic_join}
\end{figure*}
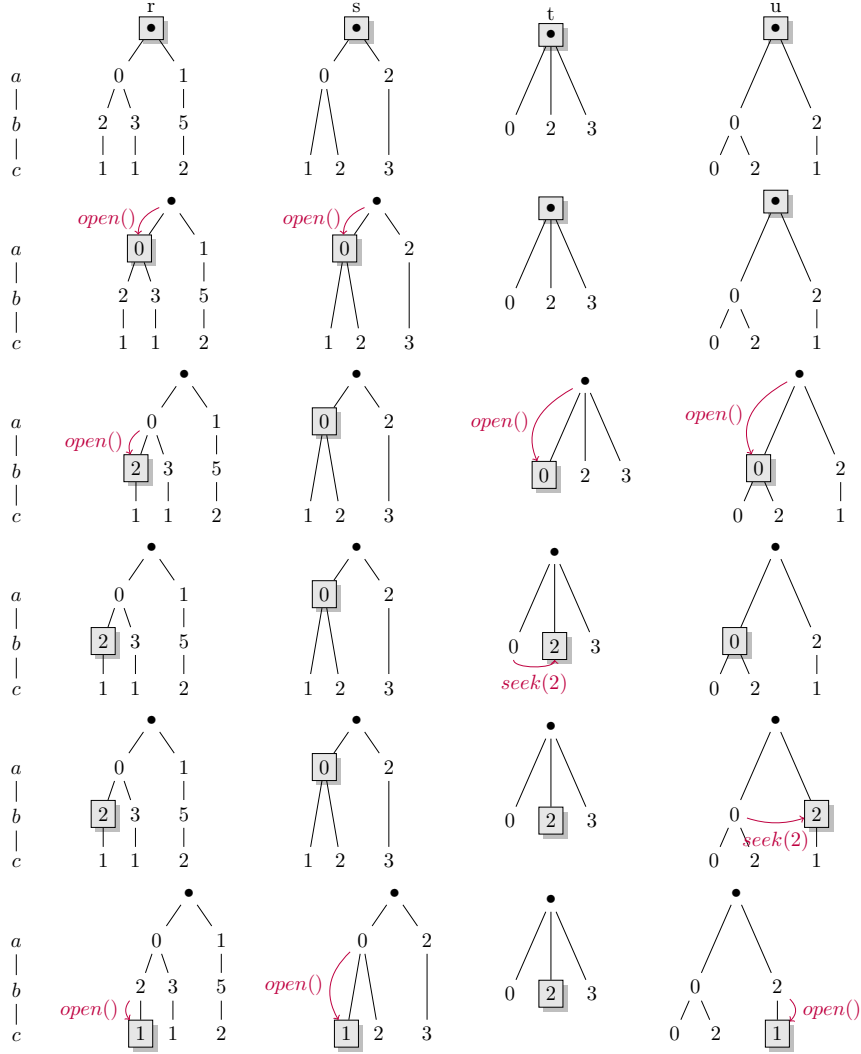

\begin{example}
\label{ex:acyclic}
\rm
Consider the acyclic query $q(a,b,c,) =$ $r(a,b,c),$ $s(a,c),$ $t(b),$ $u(b,c)$
joining the four relations depicted in Figure~\ref{fig:acyclic_join}.
As shown in the figure, the only tuple in the join result is 
$(0,2,1)$.
Figure~\ref{fig:trie_acyclic_join}
shows how LFTJ traverses the four relations 
following the global attribute order $a-b-c$ to compute the join result.
First, the algorithm calls $open()$ for $r$ and $s$, since these are the only 
relations containing attribute $a$. 
The iterators of these relations move to nodes carrying 
$0$, which  means that we have a match for the $a$-values 
(second row in Figure~\ref{fig:trie_acyclic_join}).
Next, the algorithm calls $open()$ for $r$, $t$, and $u$, 
since these are the relations that have attribute $b$. 
The current value of the iterator for $r$ becomes $2$ while the current value 
of the iterators for $t$ and $u$ become $0$
(third row in Figure~\ref{fig:trie_acyclic_join}). This means that the $b$-values
do not match yet. Since the largest current $b$-value is $2$, the algorithm calls $seek(2)$ for $t$, which moves the iterator of $t$ to the sibling node with value $2$ (fourth row in Figure~\ref{fig:trie_acyclic_join}).
In this situation, the $b$-values
still do not match and the largest current $b$-value is $2$. 
The algorithm calls $seek(2)$ for $u$, moving the iterator of $u$ to the sibling node with value $2$ (fifth row in Figure~\ref{fig:trie_acyclic_join}). Now, the iterators of $r$, $t$, and $u$
point to $2$, so the algorithm   
fixes the $b$-value to $2$. Calling $open()$ for
$r$, $s$, and $u$ moves their iterators to child nodes with value $1$.
Hence, the algorithm fixes the $c$-value to $1$ (sixth row in Figure~\ref{fig:trie_acyclic_join}). 
It follows that $(0,2,1)$ constitutes
the first result tuple, which is added to the output. 
After backtracking, the algorithm realizes that there is no further tuple in the 
join result and stops.
\end{example}

\subsection{Leapfrog Triejoin for a Cyclic Query}
\label{sec:lftj_cyclic_query}
The next example illustrates that traditional 
join algorithms are suboptimal 
in the sense that they can produce intermediate results that are larger than the final 
result. In contrast, LFTJ does not produce intermediate results and constructs
one output tuple at a time. The example considers the triangle query 
and showcases a database that  
was used in prior work to demonstrate the suboptimality of traditional 
join algorithms on skewed data~\cite{DBLP:journals/sigmod/NgoRR13}.

\begin{example}
\rm
Consider the triangle query $q(a,b,c)$ $=$ $r(a,b),$ 
$s(a,c),$ $t(b,c)$, which joins the three relations $r$, $s$, and $t$ depicted in Figure~\ref{fig:triangle_join}. 
Each relation has two values of degree $m+1$ and $2m$ values 
of degree $1$. For instance, in relation $r$, each of the values $a_0$ and $b_0$ is paired with $m+1$ distinct values and each of the remaining values is paired with exactly one value.  
Each input relation has $2m + 1$ tuples while  the result (Figure~\ref{fig:triangle_join} right)
has $3m + 1$ tuples. Hence, the size of the result is linear in $m$.

\begin{figure}[h]
\centering
\begin{center}
\begin{tabular}{c}
\begin{tabular}{cccc}
$r({\color{red}a},{\color{blue} b})$ & $s({\color{red}a},{\color{goodgreen} c})$ & $t({\color{blue}b},{\color{goodgreen} c})$ & $\text{join of } r, s,\text{ and } t$\smallskip \\
\begin{tabular}{rr}\toprule
\color{red}$a_0$ &\color{blue} $b_0$ \\
\color{red}$a_0$ &\color{blue} $\ldots$ \\
\color{red}$a_0$ &\color{blue} $b_m$ \\\hline
\color{red}$a_1$ &\color{blue} $b_0$ \\
\color{red}$\ldots$ &\color{blue} $b_0$ \\
\color{red}$a_m$ & \color{blue}$b_0$ \\\bottomrule
\\
\\
\\
\end{tabular}
&
\begin{tabular}{rr}\toprule
\color{red}$a_0$ & \color{goodgreen}$c_0$ \\
\color{red}$a_0$ & \color{goodgreen}$\ldots$ \\
\color{red}$a_0$ & \color{goodgreen}$c_m$ \\\hline
\color{red}$a_1$ & \color{goodgreen}$c_0$ \\
\color{red}$\ldots$ & \color{goodgreen}$c_0$ \\
\color{red}$a_m$ & \color{goodgreen}$c_0$ \\\bottomrule
\\
\\
\\
\end{tabular}
&
\begin{tabular}{rr}\toprule
\color{blue}$b_0$ & \color{goodgreen}$c_0$ \\
\color{blue}$b_0$ & \color{goodgreen}$\ldots$ \\
\color{blue}$b_0$ & \color{goodgreen}$c_m$ \\\hline
\color{blue}$b_1$ & \color{goodgreen}$c_0$ \\
\color{blue}$\ldots$ & \color{goodgreen}$c_0$ \\
\color{blue}$b_m$ & \color{goodgreen}$c_0$ \\\bottomrule
\\
\\
\\
\end{tabular}&
\begin{tabular}{rrr}\toprule
\color{red}$a_0$ & \color{blue}$b_0$ & \color{goodgreen}$c_0$ \\
\color{red}$a_0$ & \color{blue}$b_0$ & \color{goodgreen}$\ldots$ \\
\color{red}$a_0$ & \color{blue}$b_0$ & \color{goodgreen}$c_m$ \\\hline
\color{red}$a_0$ & \color{blue}$b_1$ & \color{goodgreen}$c_0$ \\
\color{red}$a_0$ & \color{blue}$\ldots$ & \color{goodgreen}$c_0$ \\
\color{red}$a_0$ & \color{blue}$b_m$ & \color{goodgreen}$c_0$ \\\hline
\color{red}$a_1$ & \color{blue}$b_0$ & \color{goodgreen}$c_0$ \\
\color{red}$\ldots$ & \color{blue}$b_0$ & \color{goodgreen}$c_0$ \\
\color{red}$a_m$ & \color{blue}$b_0$ & \color{goodgreen}$c_0$ \\\hline
\end{tabular}
\end{tabular}
\end{tabular}
\end{center}
\hspace*{-0.5cm}
\caption{
Relations of the triangle join.
}
\label{fig:triangle_join}
\end{figure} 

Traditional join algorithms first join  two of the three input relations and then join in the remaining one. 
%The join of any two relations has size quadratic in $m$ and hence requires at least quadratic time to compute. 
Figure~\ref{fig:binary_join} shows the result of 
any pairwise join.
In each case, the result contains $(m+1)^2 + m$ tuples, which 
is quadratic in $m$. This means that any join plan that first joins two of the three input relations needs at least quadratic time, while
the size of the final result is linear in $m$.

\begin{figure}[b]
\centering
\begin{center}
\begin{tabular}{c}
\begin{tabular}{ccc}
$\text{join of } r\text{ and } s$ 
& $\text{join of } r \text{ and } t$
& $\text{join of } s \text{ and } t$\smallskip \\
\begin{tabular}{rrr}\toprule
\color{red}$a_0$ & \color{blue}$b_0$ & \color{goodgreen}$c_0$ \\
\color{red}$a_0$ & \color{blue}$b_0$ & \color{goodgreen}$\ldots$ \\
\color{red}$a_0$ & \color{blue}$b_0$ & \color{goodgreen}$c_m$ \\\hline
$\ldots$ & $\ldots$ & $\ldots$ \\\hline
\color{red}$a_0$ & \color{blue}$b_m$ & \color{goodgreen}$c_0$ \\
\color{red}$a_0$ & \color{blue}$b_m$ & \color{goodgreen}$\ldots$ \\
\color{red}$a_0$ & \color{blue}$b_m$ & \color{goodgreen}$c_m$ \\\hline
\color{red}$a_1$ & \color{blue}$b_0$ & \color{goodgreen}$c_0$ \\
\color{red}$\ldots$ & \color{blue}$b_0$ & \color{goodgreen}$c_0$ \\
\color{red}$a_m$ & \color{blue}$b_0$ & \color{goodgreen}$c_0$ \\\hline
\end{tabular}
&
\begin{tabular}{rrr}\toprule
\color{red}$a_0$ & \color{blue}$b_0$ & \color{goodgreen}$c_0$ \\
\color{red}$a_0$ & \color{blue}$b_0$ & \color{goodgreen}$\ldots$ \\
\color{red}$a_0$ & \color{blue}$b_0$ & \color{goodgreen}$c_m$ \\\hline
$\ldots$ & $\ldots$ & $\ldots$ \\\hline
\color{red}$a_m$ & \color{blue}$b_0$ & \color{goodgreen}$c_0$ \\
\color{red}$a_m$ & \color{blue}$b_0$ & \color{goodgreen}$\ldots$ \\
\color{red}$a_m$ & \color{blue}$b_0$ & \color{goodgreen}$c_m$ \\\hline
\color{red}$a_0$ & \color{blue}$b_1$ & \color{goodgreen}$c_0$ \\
\color{red}$a_0$ & \color{blue}$\ldots$ & \color{goodgreen}$c_0$ \\
\color{red}$a_0$ & \color{blue}$b_m$ & \color{goodgreen}$c_0$ \\\hline
\end{tabular}
&
\begin{tabular}{rrr}\toprule
\color{red}$a_0$ & \color{blue}$b_0$ & \color{goodgreen}$c_0$ \\
\color{red}$a_0$ & \color{blue}$\ldots$ & \color{goodgreen}$\ldots$ \\
\color{red}$a_0$ & \color{blue}$b_m$ & \color{goodgreen}$c_0$ \\\hline
$\ldots$ & $\ldots$ & $\ldots$ \\\hline
\color{red}$a_m$ & \color{blue}$b_0$ & \color{goodgreen}$c_0$ \\
\color{red}$a_m$ & \color{blue}$\ldots$ & \color{goodgreen}$c_0$ \\
\color{red}$a_m$ & \color{blue}$b_m$ & \color{goodgreen}$c_0$ \\\hline
\color{red}$a_0$ & \color{blue}$b_0$ & \color{goodgreen}$c_1$ \\
\color{red}$a_0$ & \color{blue}$b_0$ & \color{goodgreen}$\ldots$ \\
\color{red}$a_0$ & \color{blue}$b_0$ & \color{goodgreen}$c_m$ \\\hline
\end{tabular}
\end{tabular}
\end{tabular}
\end{center}
\hspace*{-0.5cm}
\caption{
Pairwise joins of the relations 
$r$, $s$, and $t$ from Figure~\ref{fig:triangle_join}.}
\label{fig:binary_join}
\end{figure}

LFTJ does not produce any intermediate result
and its computation time is proportional to the size of the final result. 
Figure~\ref{fig:trie_traversal_triangle_join} visualizes how LFTJ 
traverses the three relations following the global attribute order 
$a-b-c$. The last column in the Figure
shows how the join result is produced.  
Just as in Example~\ref{ex:acyclic}, LFTJ uses leapfrog join to compute
the intersection of $a$-values, the intersection of $b$-values in the context of 
a given $a$-value, and the intersection of $c$-values in the context 
of a given $(a,b)$-pair. Hence, in the sequel we focus more on the order in which complete 
result tuples are produced.

The first result tuple is constructed by fixing the attributes 
$a$, $b$, and $c$ to the values $a_0$, $b_0$, and $c_0$, respectively
(second row in Figure~\ref{fig:trie_traversal_triangle_join}).
The $c$-values in the context of $a_0$ in $s$ and in the context of 
$b_0$ in $t$ are teh same: $c_1, \ldots , c_m$. So, we obtain the result tuples 
$(a_0, b_0, c_1)$, $\ldots$, $(a_0, b_0, c_m)$
(third row in Figure~\ref{fig:trie_traversal_triangle_join}).
The $b$-values in the context of $a_0$ in $r$ are the same as the $b$-values in 
$t$. All $b$-values in $t$ have $c_0$ as a child, which 
is also child of $a_0$ in $s$. 
So, the algorithm produces the result tuples 
$(a_0, b_1, c_0)$, $\ldots$, $(a_0, b_m, c_0)$
(fourth row in Figure~\ref{fig:trie_traversal_triangle_join}).
At this point, all $b$- and $c$-values in the context of $a_0$
are exhausted. The algorithm moves to the next $a$-value $a_1$
in the intersection of the $a$-values in $r$ and $s$. The value 
$b_0$ appears in $t$ and in the context of $a_1$ in $r$. Moreover, 
the value $c_0$ appears in the context of $a_1$ in $s$ and of $b_0$
in $t$. Since this is a match, the algorithm adds $(a_1, b_0, c_0)$
to the output (fifth row in Figure~\ref{fig:trie_traversal_triangle_join}).
Next, the algorithm iterates over the $a$-values $a_2, \ldots, a_m$. 
Each of these values have $b_0$ as child in $r$ and $c_0$ as child in $s$. 
Hence, the algorithm produces the result tuples
$(a_2, b_0, c_0)$, $\ldots$, $(a_m, b_0, c_0)$
(sixth row in Figure~\ref{fig:trie_traversal_triangle_join}).
\end{example}

\color{black}
\begin{figure*}[h]
\centering
 \adjustbox{width=0.74\textwidth}{
% \adjustbox{width=12.2cm}{
\begin{tabular}{ccccc}
\begin{tikzpicture}[xscale=0.55, yscale=0.46]
\node at (0,1) (name) {};
\node at (0,-1) (a) {$a$};
\node at (0,-3) (b) {$b$} edge[-] (a);
\node at (0,-5) (c) {$c$} edge[-] (b);
\end{tikzpicture} 
&
%%%%%%
% FIRST FACTOR
%%%%%%
\begin{tikzpicture}[xscale=0.55, yscale=0.46]
  \node at (0.2,1.8) (name) {$r$};
  \node[drop shadow, fill=gray!20, draw] at (0.2,1) (u1) {\textbullet};
  \node at (-2,-1) (a0) {$a_0$} edge[-] (u1);

  \node at (0.2,-1) (a1) {$a_1$} edge[-] (u1);

  \node at (1.1,-1) (dots1) {$\ldots$}; 
  
  \node at (2.3,-1) (am) {$a_m$} edge[-] (u1);

  \node at (-2.8,-3) (b0) {$b_0$} edge[-] (a0);

  \node at (-2  ,-3) (dots2) {$\ldots$};
  \node at (-1.1,-3) (bm) {$b_m$} edge[-] (a0);

  \node at (0.2,-3) (b0) {$b_0$} edge[-] (a1);

  	\node at (2.3,-3) (b0) {$b_0$} edge[-] (am);

  \node at (0,-5) (nothing) {};
\end{tikzpicture}
%\hspace*{-1em}
&
%%%%%%
% SECOND FACTOR
%%%%%%
\begin{tikzpicture}[xscale=0.55, yscale=0.46]
\node at (0.2,1.8) (name) {$s$};
  \node[drop shadow, fill=gray!20, draw] at (0.2,1) (u1) {\textbullet};
	  \node at (-2,-1) (a0) {$a_0$} edge[-] (u1);
	  
  \node at (0.2,-1) (a1) {$a_1$} edge[-] (u1);

  \node at (1,-1) (dots1) {$\ldots$};

  \node at (2.3,-1) (am) {$a_m$} edge[-] (u1);
  \node at (2.3,-5) (b0) {$c_0$} edge[-] (am);

  \node at (-2.8,-5) (b0) {$c_0$} edge[-] (a0);

  \node at (-2  ,-5) (dots2) {$\ldots$};
  \node at (-1.1,-5) (bm) {$c_m$} edge[-] (a0);

  \node at (0.2,-5) (b0) {$c_0$} edge[-] (a1);

\end{tikzpicture}
&
%%%%%%
% THIRD FACTOR
%%%%%%
\begin{tikzpicture}[xscale=0.55, yscale=0.46]
\node at (0.2,1.8) (name) {$t$};
  \node[drop shadow, fill=gray!20, draw] at (0.2,1) (u1) {\textbullet};
  \node at (-2,-3) (a0) {$b_0$} edge[-] (u1);

  \node at (0.2,-3) (a1) {$b_1$} edge[-] (u1);	
  \node at (1,-3) (dots1) {$\ldots$};
  \node at (2.3,-3) (am) {$b_m$} edge[-] (u1);

  \node at (-2.8,-5) (b0) {$c_0$} edge[-] (a0);
	
  \node at (-2  ,-5) (dots2) {$\ldots$};
  \node at (-1.1,-5) (bm) {$c_m$} edge[-] (a0);
  
  \node at (0.2,-5) (b0) {$c_0$} edge[-] (a1);
  \node at (2.3,-5) (b0) {$c_0$} edge[-] (am);
\end{tikzpicture}
&
%Result

\\
%Second Round
\begin{tikzpicture}[xscale=0.55, yscale=0.46]
\node at (0,1) (bullet) {};
\node at (0,-1) (a) {$a$};
\node at (0,-3) (b) {$b$} edge[-] (a);
\node at (0,-5) (c) {$c$} edge[-] (b);
\end{tikzpicture} 
&
%%%%%%
% FIRST FACTOR
%%%%%%
\begin{tikzpicture}[xscale=0.55, yscale=0.46]
  \node at (0.2,1) (u1) {\textbullet};
  \node[drop shadow, fill=gray!20, draw] at (-2,-1) (a0) {$a_0$} edge[-] (u1);

  \node at (0.2,-1) (a1) {$a_1$} edge[-] (u1);

  \node at (1.1,-1) (dots1) {$\ldots$}; 
  
  \node at (2.3,-1) (am) {$a_m$} edge[-] (u1);

  \node[drop shadow, fill=gray!20, draw] at (-2.8,-3) (b0) {$b_0$} edge[-] (a0);

  \node at (-2  ,-3) (dots2) {$\ldots$};
  \node at (-1.1,-3) (bm) {$b_m$} edge[-] (a0);

  \node at (0.2,-3) (b0) {$b_0$} edge[-] (a1);

  	\node at (2.3,-3) (b0) {$b_0$} edge[-] (am);

  \node at (0,-5) (nothing) {};
\end{tikzpicture}
%\hspace*{-1em}
&
%%%%%%
% SECOND FACTOR
%%%%%%
\begin{tikzpicture}[xscale=0.55, yscale=0.46]
  \node at (0.2,1) (u1) {\textbullet};
	  \node[drop shadow, fill=gray!20, draw] at (-2,-1) (a0) {$a_0$} edge[-] (u1);
	  
  \node at (0.2,-1) (a1) {$a_1$} edge[-] (u1);

  \node at (1,-1) (dots1) {$\ldots$};

  \node at (2.3,-1) (am) {$a_m$} edge[-] (u1);
  \node at (2.3,-5) (b0) {$c_0$} edge[-] (am);

  \node[drop shadow, fill=gray!20, draw] at (-2.8,-5) (b0) {$c_0$} edge[-] (a0);

  \node at (-2  ,-5) (dots2) {$\ldots$};
  \node at (-1.1,-5) (bm) {$c_m$} edge[-] (a0);

  \node at (0.2,-5) (b0) {$c_0$} edge[-] (a1);

\end{tikzpicture}
&
%%%%%%
% THIRD FACTOR
%%%%%%
\begin{tikzpicture}[xscale=0.55, yscale=0.46]
  \node at (0.2,1) (u1) {\textbullet};
  \node[drop shadow, fill=gray!20, draw] at (-2,-3) (a0) {$b_0$} edge[-] (u1);

  \node at (0.2,-3) (a1) {$b_1$} edge[-] (u1);	
  \node at (1,-3) (dots1) {$\ldots$};
  \node at (2.3,-3) (am) {$b_m$} edge[-] (u1);

  \node[drop shadow, fill=gray!20, draw] at (-2.8,-5) (b0) {$c_0$} edge[-] (a0);
	
  \node at (-2  ,-5) (dots2) {$\ldots$};
  \node at (-1.1,-5) (bm) {$c_m$} edge[-] (a0);
  
  \node at (0.2,-5) (b0) {$c_0$} edge[-] (a1);
  \node at (2.3,-5) (b0) {$c_0$} edge[-] (am);
\end{tikzpicture}
&
%Result
\begin{tikzpicture}[xscale=0.55, yscale=0.46]
    \node at (0.5,1) (u1) {\textbullet};
    \node at (-2,-1) (a0) {$a_0$} edge[-] (u1);
    \node at (-3.6,-3) (b0) {$b_0$} edge[-] (a0);
    \node at(-4.6,-5) (b0tuc0) {$c_0$} edge[-] (b0);
 	
    \node at (2.5,-5) (b0tuc0) {};
\end{tikzpicture}
\\
%Third Round
\begin{tikzpicture}[xscale=0.55, yscale=0.46]
\node at (0,1) (name) {};
\node at (0,-1) (a) {$a$};
\node at (0,-3) (b) {$b$} edge[-] (a);
\node at (0,-5) (c) {$c$} edge[-] (b);
\end{tikzpicture} 
&
%%%%%%
% FIRST FACTOR
%%%%%%
\begin{tikzpicture}[xscale=0.55, yscale=0.46]
  \node at (0.2,1) (u1) {\textbullet};
  \node[drop shadow, fill=gray!20, draw] at (-2,-1) (a0) {$a_0$} edge[-] (u1);

  \node at (0.2,-1) (a1) {$a_1$} edge[-] (u1);

  \node at (1.1,-1) (dots1) {$\ldots$}; 
  
  \node at (2.3,-1) (am) {$a_m$} edge[-] (u1);

  \node[drop shadow, fill=gray!20, draw] at (-2.8,-3) (b0) {$b_0$} edge[-] (a0);

  \node at (-2  ,-3) (dots2) {$\ldots$};
  \node at (-1.1,-3) (bm) {$b_m$} edge[-] (a0);

  \node at (0.2,-3) (b0) {$b_0$} edge[-] (a1);

  	\node at (2.3,-3) (b0) {$b_0$} edge[-] (am);

  \node at (0,-5) (nothing) {};
\end{tikzpicture}
%\hspace*{-1em}
&
%%%%%%
% SECOND FACTOR
%%%%%%
\begin{tikzpicture}[xscale=0.55, yscale=0.46]
  \node at (0.2,1) (u1) {\textbullet};
	  \node[drop shadow, fill=gray!20, draw] at (-2,-1) (a0) {$a_0$} edge[-] (u1);
	  
  \node at (0.2,-1) (a1) {$a_1$} edge[-] (u1);

  \node at (1,-1) (dots1) {$\ldots$};

  \node at (2.3,-1) (am) {$a_m$} edge[-] (u1);
  \node at (2.3,-5) (b0) {$c_0$} edge[-] (am);

  \node at (-2.8,-5) (b0) {$c_0$} edge[-] (a0);

  \node at (-2  ,-5) (dots2) {$\ldots$};
  \node[drop shadow, fill=gray!20, draw] at (-1.1,-5) (bm) {$c_m$} edge[-] (a0);

  \node at (0.2,-5) (b0) {$c_0$} edge[-] (a1);

\end{tikzpicture}
&
%%%%%%
% THIRD FACTOR
%%%%%%
\begin{tikzpicture}[xscale=0.55, yscale=0.46]
  \node at (0.2,1) (u1) {\textbullet};
  \node[drop shadow, fill=gray!20, draw] at (-2,-3) (a0) {$b_0$} edge[-] (u1);

  \node at (0.2,-3) (a1) {$b_1$} edge[-] (u1);	
  \node at (1,-3) (dots1) {$\ldots$};
  \node at (2.3,-3) (am) {$b_m$} edge[-] (u1);

  \node at (-2.8,-5) (b0) {$c_0$} edge[-] (a0);
	
  \node at (-2  ,-5) (dots2) {$\ldots$};
  \node[drop shadow, fill=gray!20, draw] at (-1.1,-5) (bm) {$c_m$} edge[-] (a0);
  
  \node at (0.2,-5) (b0) {$c_0$} edge[-] (a1);
  \node at (2.3,-5) (b0) {$c_0$} edge[-] (am);
\end{tikzpicture}
&
%Result
\begin{tikzpicture}[xscale=0.55, yscale=0.46]
    \node at (0.5,1) (u1) {\textbullet};
    \node at (-2,-1) (a0) {$a_0$} edge[-] (u1);
    \node at (-3.6,-3) (b0) {$b_0$} edge[-] (a0);
    \node at(-4.6,-5) (b0tuc0) {$c_0$} edge[-] (b0);

    \node at(-3.8,-5) (dots) {$\ldots$};
    \node at(-2.8,-5) (b0tucm) {$c_m$} edge[-] (b0);
	
    \node at (2.5,-5) (b0tuc0) {};
\end{tikzpicture}
\\
%fourth round
\begin{tikzpicture}[xscale=0.55, yscale=0.46]
\node at (0,1) (name) {};
\node at (0,-1) (a) {$a$};
\node at (0,-3) (b) {$b$} edge[-] (a);
\node at (0,-5) (c) {$c$} edge[-] (b);
\end{tikzpicture} 
&
%%%%%%
% FIRST FACTOR
%%%%%%
\begin{tikzpicture}[xscale=0.55, yscale=0.46]
  \node at (0.2,1) (u1) {\textbullet};
  \node[drop shadow, fill=gray!20, draw] at (-2,-1) (a0) {$a_0$} edge[-] (u1);

  \node at (0.2,-1) (a1) {$a_1$} edge[-] (u1);

  \node at (1.1,-1) (dots1) {$\ldots$}; 
  
  \node at (2.3,-1) (am) {$a_m$} edge[-] (u1);

  \node at (-2.8,-3) (b0) {$b_0$} edge[-] (a0);

  \node at (-2  ,-3) (dots2) {$\ldots$};
  \node[drop shadow, fill=gray!20, draw] at (-1.1,-3) (bm) {$b_m$} edge[-] (a0);

  \node at (0.2,-3) (b0) {$b_0$} edge[-] (a1);

  	\node at (2.3,-3) (b0) {$b_0$} edge[-] (am);

  \node at (0,-5) (nothing) {};
\end{tikzpicture}
%\hspace*{-1em}
&
%%%%%%
% SECOND FACTOR
%%%%%%
\begin{tikzpicture}[xscale=0.55, yscale=0.46]
  \node at (0.2,1) (u1) {\textbullet};
	  \node[drop shadow, fill=gray!20, draw] at (-2,-1) (a0) {$a_0$} edge[-] (u1);
	  
  \node at (0.2,-1) (a1) {$a_1$} edge[-] (u1);

  \node at (1,-1) (dots1) {$\ldots$};

  \node at (2.3,-1) (am) {$a_m$} edge[-] (u1);
  \node at (2.3,-5) (b0) {$c_0$} edge[-] (am);

  \node[drop shadow, fill=gray!20, draw] at (-2.8,-5) (b0) {$c_0$} edge[-] (a0);

  \node at (-2  ,-5) (dots2) {$\ldots$};
  \node at (-1.1,-5) (bm) {$c_m$} edge[-] (a0);

  \node at (0.2,-5) (b0) {$c_0$} edge[-] (a1);

\end{tikzpicture}
&
%%%%%%
% THIRD FACTOR
%%%%%%
\begin{tikzpicture}[xscale=0.55, yscale=0.46]
  \node at (0.2,1) (u1) {\textbullet};
  \node at (-2,-3) (a0) {$b_0$} edge[-] (u1);

  \node at (0.2,-3) (a1) {$b_1$} edge[-] (u1);	
  \node at (1,-3) (dots1) {$\ldots$};
  \node[drop shadow, fill=gray!20, draw] at (2.3,-3) (am) {$b_m$} edge[-] (u1);

  \node at (-2.8,-5) (b0) {$c_0$} edge[-] (a0);
	
  \node at (-2  ,-5) (dots2) {$\ldots$};
  \node at (-1.1,-5) (bm) {$c_m$} edge[-] (a0);
  
  \node at (0.2,-5) (b0) {$c_0$} edge[-] (a1);
  \node[drop shadow, fill=gray!20, draw] at (2.3,-5) (b0) {$c_0$} edge[-] (am);
\end{tikzpicture}
&
%Result
\begin{tikzpicture}[xscale=0.55, yscale=0.46]
    \node at (0.5,1) (u1) {\textbullet};
    \node at (-2,-1) (a0) {$a_0$} edge[-] (u1);
    \node at (-3.6,-3) (b0) {$b_0$} edge[-] (a0);
    \node at(-4.6,-5) (b0tuc0) {$c_0$} edge[-] (b0);

    \node at(-3.8,-5) (dots) {$\ldots$};
    \node at(-2.8,-5) (b0tucm) {$c_m$} edge[-] (b0);

    \node at (-2,-3) (b1) {$b_1$} edge[-] (a0);
    \node at (-1.3,-3) (dots1) {$\ldots$};
    \node at (-0.5,-3) (bm) {$b_m$} edge[-] (a0);
    \node at(-2,-5) (b1tuc0) {$c_0$} edge[-] (b1);
    \node at(-0.5,-5) (bmtuc0) {$c_0$} edge[-] (bm);
	
    \node at (2.5,-5) (b0tuc0) {};
\end{tikzpicture}
\\
%fifth round
\begin{tikzpicture}[xscale=0.55, yscale=0.46]
\node at (0,1) (name) {};
\node at (0,-1) (a) {$a$};
\node at (0,-3) (b) {$b$} edge[-] (a);
\node at (0,-5) (c) {$c$} edge[-] (b);
\end{tikzpicture} 
&
%%%%%%
% FIRST FACTOR
%%%%%%
\begin{tikzpicture}[xscale=0.55, yscale=0.46]
  \node at (0.2,1) (u1) {\textbullet};
  \node at (-2,-1) (a0) {$a_0$} edge[-] (u1);

  \node[drop shadow, fill=gray!20, draw] at (0.2,-1) (a1) {$a_1$} edge[-] (u1);

  \node at (1.1,-1) (dots1) {$\ldots$}; 
  
  \node at (2.3,-1) (am) {$a_m$} edge[-] (u1);

  \node at (-2.8,-3) (b0) {$b_0$} edge[-] (a0);

  \node at (-2  ,-3) (dots2) {$\ldots$};
  \node at (-1.1,-3) (bm) {$b_m$} edge[-] (a0);

  \node[drop shadow, fill=gray!20, draw] at (0.2,-3) (b0) {$b_0$} edge[-] (a1);

  	\node at (2.3,-3) (b0) {$b_0$} edge[-] (am);

  \node at (0,-5) (nothing) {};
\end{tikzpicture}
%\hspace*{-1em}
&
%%%%%%
% SECOND FACTOR
%%%%%%
\begin{tikzpicture}[xscale=0.55, yscale=0.46]
  \node at (0.2,1) (u1) {\textbullet};
	  \node at (-2,-1) (a0) {$a_0$} edge[-] (u1);
	  
  \node[drop shadow, fill=gray!20, draw] at (0.2,-1) (a1) {$a_1$} edge[-] (u1);

  \node at (1,-1) (dots1) {$\ldots$};

  \node at (2.3,-1) (am) {$a_m$} edge[-] (u1);
  \node at (2.3,-5) (b0) {$c_0$} edge[-] (am);

  \node at (-2.8,-5) (b0) {$c_0$} edge[-] (a0);

  \node at (-2  ,-5) (dots2) {$\ldots$};
  \node at (-1.1,-5) (bm) {$c_m$} edge[-] (a0);

  \node[drop shadow, fill=gray!20, draw] at (0.2,-5) (b0) {$c_0$} edge[-] (a1);

\end{tikzpicture}
&
%%%%%%
% THIRD FACTOR
%%%%%%
\begin{tikzpicture}[xscale=0.55, yscale=0.46]
  \node at (0.2,1) (u1) {\textbullet};
  \node[drop shadow, fill=gray!20, draw] at (-2,-3) (a0) {$b_0$} edge[-] (u1);

  \node at (0.2,-3) (a1) {$b_1$} edge[-] (u1);	
  \node at (1,-3) (dots1) {$\ldots$};
  \node at (2.3,-3) (am) {$b_m$} edge[-] (u1);

  \node[drop shadow, fill=gray!20, draw] at (-2.8,-5) (b0) {$c_0$} edge[-] (a0);
	
  \node at (-2  ,-5) (dots2) {$\ldots$};
  \node at (-1.1,-5) (bm) {$c_m$} edge[-] (a0);
  
  \node at (0.2,-5) (b0) {$c_0$} edge[-] (a1);
  \node at (2.3,-5) (b0) {$c_0$} edge[-] (am);
\end{tikzpicture}
&
%Result
\begin{tikzpicture}[xscale=0.55, yscale=0.46]
    \node at (0.5,1) (u1) {\textbullet};
    \node at (-2,-1) (a0) {$a_0$} edge[-] (u1);
    \node at (-3.6,-3) (b0) {$b_0$} edge[-] (a0);
    \node at(-4.6,-5) (b0tuc0) {$c_0$} edge[-] (b0);

    \node at(-3.8,-5) (dots) {$\ldots$};
    \node at(-2.8,-5) (b0tucm) {$c_m$} edge[-] (b0);

    \node at (-2,-3) (b1) {$b_1$} edge[-] (a0);
    \node at (-1.3,-3) (dots1) {$\ldots$};
    \node at (-0.5,-3) (bm) {$b_m$} edge[-] (a0);
    \node at(-2,-5) (b1tuc0) {$c_0$} edge[-] (b1);
    \node at(-0.5,-5) (bmtuc0) {$c_0$} edge[-] (bm);

    \node at (0.5,-1) (a1) {$a_1$} edge[-] (u1);
    \node at (0.5,-3) (b0) {$b_0$} edge[-] (a1);
    \node at (0.5,-5) (b0tuc0) {$c_0$} edge[-] (b0);

    \node at (2.5,-5) (b0tuc0) {};
\end{tikzpicture}
\\
%sixth round
\begin{tikzpicture}[xscale=0.55, yscale=0.46]
\node at (0,1) (name) {};
\node at (0,-1) (a) {$a$};
\node at (0,-3) (b) {$b$} edge[-] (a);
\node at (0,-5) (c) {$c$} edge[-] (b);
\end{tikzpicture} 
&
%%%%%%
% FIRST FACTOR
%%%%%%
\begin{tikzpicture}[xscale=0.55, yscale=0.46]
  \node at (0.2,1) (u1) {\textbullet};
  \node at (-2,-1) (a0) {$a_0$} edge[-] (u1);

  \node at (0.2,-1) (a1) {$a_1$} edge[-] (u1);

  \node at (1.1,-1) (dots1) {$\ldots$}; 
  
  \node[drop shadow, fill=gray!20, draw] at (2.3,-1) (am) {$a_m$} edge[-] (u1);

  \node at (-2.8,-3) (b0) {$b_0$} edge[-] (a0);

  \node at (-2  ,-3) (dots2) {$\ldots$};
  \node at (-1.1,-3) (bm) {$b_m$} edge[-] (a0);

  \node at (0.2,-3) (b0) {$b_0$} edge[-] (a1);

  	\node[drop shadow, fill=gray!20, draw] at (2.3,-3) (b0) {$b_0$} edge[-] (am);

  \node at (0,-5) (nothing) {};
\end{tikzpicture}
%\hspace*{-1em}
&
%%%%%%
% SECOND FACTOR
%%%%%%
\begin{tikzpicture}[xscale=0.55, yscale=0.46]
  \node at (0.2,1) (u1) {\textbullet};
	  \node at (-2,-1) (a0) {$a_0$} edge[-] (u1);
	  
  \node at (0.2,-1) (a1) {$a_1$} edge[-] (u1);

  \node at (1,-1) (dots1) {$\ldots$};

  \node [drop shadow, fill=gray!20, draw] at (2.3,-1) (am) {$a_m$} edge[-] (u1);
  \node[drop shadow, fill=gray!20, draw] at (2.3,-5) (b0) {$c_0$} edge[-] (am);

  \node at (-2.8,-5) (b0) {$c_0$} edge[-] (a0);

  \node at (-2  ,-5) (dots2) {$\ldots$};
  \node at (-1.1,-5) (bm) {$c_m$} edge[-] (a0);

  \node at (0.2,-5) (b0) {$c_0$} edge[-] (a1);

\end{tikzpicture}
&
%%%%%%
% THIRD FACTOR
%%%%%%
\begin{tikzpicture}[xscale=0.55, yscale=0.46]
  \node at (0.2,1) (u1) {\textbullet};
  \node[drop shadow, fill=gray!20, draw] at (-2,-3) (a0) {$b_0$} edge[-] (u1);

  \node at (0.2,-3) (a1) {$b_1$} edge[-] (u1);	
  \node at (1,-3) (dots1) {$\ldots$};
  \node at (2.3,-3) (am) {$b_m$} edge[-] (u1);

  \node[drop shadow, fill=gray!20, draw] at (-2.8,-5) (b0) {$c_0$} edge[-] (a0);
	
  \node at (-2  ,-5) (dots2) {$\ldots$};
  \node at (-1.1,-5) (bm) {$c_m$} edge[-] (a0);
  
  \node at (0.2,-5) (b0) {$c_0$} edge[-] (a1);
  \node at (2.3,-5) (b0) {$c_0$} edge[-] (am);
\end{tikzpicture}
&
%Result
\begin{tikzpicture}[xscale=0.55, yscale=0.46]
    \node at (0.5,1) (u1) {\textbullet};
    \node at (-2,-1) (a0) {$a_0$} edge[-] (u1);
    \node at (-3.6,-3) (b0) {$b_0$} edge[-] (a0);
    \node at(-4.6,-5) (b0tuc0) {$c_0$} edge[-] (b0);

    \node at(-3.8,-5) (dots) {$\ldots$};
    \node at(-2.8,-5) (b0tucm) {$c_m$} edge[-] (b0);

    \node at (-2,-3) (b1) {$b_1$} edge[-] (a0);
    \node at (-1.3,-3) (dots1) {$\ldots$};
    \node at (-0.5,-3) (bm) {$b_m$} edge[-] (a0);
    \node at(-2,-5) (b1tuc0) {$c_0$} edge[-] (b1);
    \node at(-0.5,-5) (bmtuc0) {$c_0$} edge[-] (bm);

    \node at (0.5,-1) (a1) {$a_1$} edge[-] (u1);
    \node at (0.5,-3) (b0) {$b_0$} edge[-] (a1);
    \node at (0.5,-5) (b0tuc0) {$c_0$} edge[-] (b0);

    \node at (1.35,-1) (dots1) {$\ldots$};
    \node at (2.5,-1) (am) {$a_m$} edge[-] (u1);
    \node at (2.5,-3) (b0) {$b_0$} edge[-] (am);
    \node at (2.5,-5) (b0tuc0) {$c_0$} edge[-] (b0);
\end{tikzpicture}
\end{tabular}
}
%\hspace*{-0.5cm}
\caption{
LFTJ execution for relations $r$, $s$, and $t$ in  Figure~\ref{fig:triangle_join} following the attribute order $a - b - c$ (depicted to the left). The last column shows how the join output is produced one tuple at a time.
}
\label{fig:trie_traversal_triangle_join}
\end{figure*}
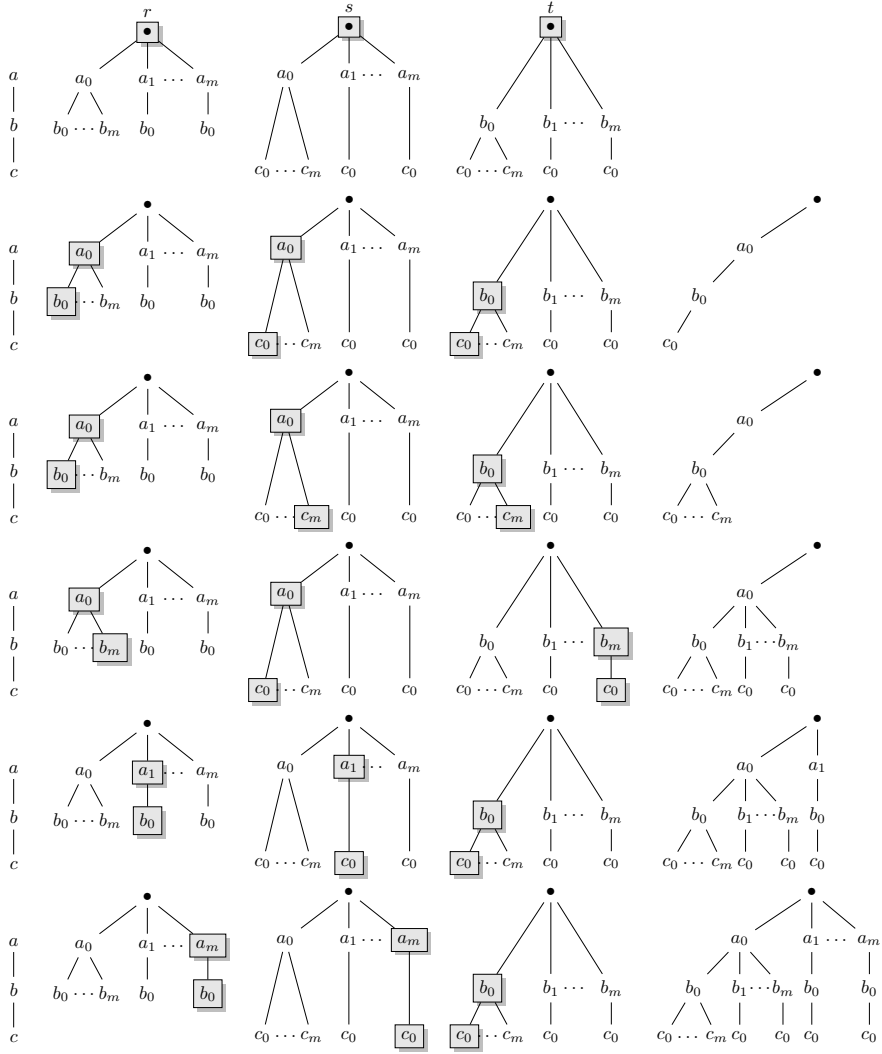

%\color{red}
\section{LFTJ Implementation in ADOPT}
\label{sec:lftjinadopt}

ADOPT uses a variant of the LFTJ, described in Algorithm~\ref{alg:joinOneCube} at a relatively high level of abstraction. The loop from Lines~10 to 17 iterates over values for the current attribute that satisfy all applicable join predicates. Here, ADOPT uses the LFTJ approach, discussed in the preceding subsections, to efficiently identify the next attribute value that appears in all input relations. As discussed previously, identifying the next value ($v$ in Algorithm~\ref{alg:joinOneCube}) may involve repeated ``seek'' operations on all input relations that contain the current attribute. The LFTJ, as presented so far, focuses on equality join predicates. In addition, ADOPT processes other join predicates by simply skipping to the next attribute value if predicates evaluate to false with the current value. Also, ADOPT only considers values for each attribute that fall within the current target cube (by starting from the lower bound and terminating iterations once seek operations return values above the upper bound). Whenever ADOPT proceeds to a new attribute, it performs the equivalent of the ``open'' operation on all input tables containing the new attribute. Finally, whereas the original LFTJ focuses on set data, ADOPT supports multi-sets by iterating over all tuple combinations having the currently selected combination of values in all join columns (variable $M$ in Algorithm~\ref{alg:joinOneCube}), when updating the query result set (Line~5 in Algorithm~\ref{alg:joinOneCube}).

LFTJ variants generally rely on data structures that enable efficient seek operations. Each table is \emph{logically} organized as a trie, where each level corresponds to one attribute and the order of the levels follows the global attribute order. This logical organization can be supported \emph{physically} by sorted tables as done by ADOPT, B$^+$-trees as in the original LFTJ implementation in LogicBlox~\cite{Aref2015}, or nested hashing~\cite{Freitag2020}. More precisely, assume we want to process the global attribute order $a_1$ to $a_m$. For a specific table, denote by $a_{i_1}$ to $a_{i_n}$ the subset of attributes that appear in that table, in the same order as they appear in the global attribute order. ADOPT sorts the rows in that table according to values for attributes $a_{i_1}$ to $a_{i_n}$, prioritizing attribute values in this order during comparisons (e.g., rows are ordered according to their value for $a_{i_1}$ as first priority, considering the value for $a_{i_2}$ only when comparing rows with the same value for $a_{i_1}$). ADOPT avoids materializing the sorted table but merely stores the ordered row indexes in an integer array of the same length as the table. Since ADOPT is a column store and stores columns as arrays, data access via the row index is efficient. Having sorted rows enables ADOPT to implement seek operations efficiently.

The original LFTJ algorithm assumes that tables have been sorted during pre-processing before run time~\cite{DBLP:conf/icdt/Veldhuizen14}. To try out different global attribute orders, ADOPT may require multiple alternative sort orders for the input relations. Whenever ADOPT selects a global attribute order, it determines which local sort orders are required for each table. If the corresponding order (i.e., the array containing sorted row indexes) is not cached, ADOPT creates the corresponding order at run time. In its caching policy, ADOPT distinguishes tables with and without unary predicates. For tables obtained after applying unary predicates, ADOPT stores associated sort orders only while processing the current query. Note that tables tend to be small after applying unary predicates, making sorting them relatively cheap. Tables without unary predicates tend to be large and sorting is more expensive. Here, ADOPT caches corresponding sort orders beyond the duration of the current query, reusing them for future queries if possible. 

% to support processing with a specific attribute order, ADOPT stores 

% LFTJ requires indices on the input tables for efficient intersection of lists of attribute values expressing multiway joins.

% The original LFTJ uses one single attribute order, which only requires a single index per joined table. As ADOPT explores several attribute orders during the processing of a single query, it may need several sorting orders of the input tables. These sorting orders are supported by lists of pointers to the original tables. An input table  with $n$ join attributes has up to $n!$ sorting orders. Whenever a new attribute order is selected by the reinforcement learning algorithm, ADOPT checks for the presence of the required sorting for each table. Before resuming join execution with the selected order, it first creates the required sorting order.

\end{document}